\documentclass[12pt]{article}

\usepackage{adjustbox}
\usepackage{algorithm2e}
\usepackage{amsmath}
\usepackage{amssymb}
\usepackage{amsthm}
\usepackage{appendix}
\usepackage{color}
\usepackage{cprotect}
\usepackage{enumerate} 
\usepackage{fullpage} 
\usepackage{graphicx} 
\usepackage{listings}
\usepackage{mathtools}
\usepackage{multirow}
\usepackage{natbib}
\usepackage{placeins}
\usepackage{subcaption} 
\usepackage{titletoc}
\usepackage{url}
\usepackage{xcolor}
\usepackage{hyperref} 
\usepackage{cleveref} 

\DeclareMathOperator*{\argmin}{arg\,min}

\usepackage{scalerel,stackengine}
\stackMath
\newcommand\reallywidehat[1]{%
	\savestack{\tmpbox}{\stretchto{%
			\scaleto{%
				\scalerel*[\widthof{\ensuremath{#1}}]{\kern-.6pt\bigwedge\kern-.6pt}%
				{\rule[-\textheight/2]{1ex}{\textheight}}
			}{\textheight}%
		}{0.5ex}}%
	\stackon[1pt]{#1}{\tmpbox}%
}

\newtheorem{theorem}{Theorem}
\newtheorem{lemma}{Lemma}
\newtheorem{proposition}{Proposition}
\newtheorem{remark}{Remark}

\newcommand{\beginsupplement}{%
        \setcounter{section}{0}
        \renewcommand{\thesection}{S\arabic{section}}%
        \setcounter{subsection}{0}
        \renewcommand{\thesubsection}{S\arabic{section}.\arabic{subsection}}%
        \setcounter{subsubsection}{0}
        \renewcommand{\thesubsubsection}{S\arabic{section}.\arabic{subsection}.\arabic{subsubsection}}%
        \setcounter{table}{0}
        \renewcommand{\thetable}{S\arabic{table}}%
        \setcounter{figure}{0}
        \renewcommand{\thefigure}{S\arabic{figure}}%
     }

    \title{Optimal Thinning of MCMC Output}
    \author{Marina Riabiz$^{1,2}$, Wilson Ye Chen$^3$, Jon Cockayne$^2$, Pawel Swietach$^{4}$, \\ Steven A. Niederer$^1$, Lester Mackey$^5$, Chris. J. Oates$^{6,2}$\footnote{Address for correspondence: Chris. J. Oates, School of Mathematics, Statistics and Physics, Herschel Building, Newcastle University, Newcastle upon Tyne, NE1 7RU, UK. E-mail: \url{chris.oates@ncl.ac.uk}} \\
	\small $^1$King's College London, UK \qquad $^2$Alan Turing Institute, UK \\
	\small $^3$University of Sydney, Australia \qquad $^4$Oxford University, UK \\
	\small $^5$Microsoft Research, US \qquad $^6$Newcastle University, UK}

\begin{document}
	
	\maketitle
	
	\begin{abstract}
		The use of heuristics to assess the convergence and compress the output of Markov chain Monte Carlo can be sub-optimal in terms of the empirical approximations that are produced.
		Typically a number of the initial states are attributed to ``burn in'' and removed, whilst the remainder of the chain is ``thinned'' if compression is also required.
		In this paper we consider the problem of retrospectively selecting a subset of states, of fixed cardinality, from the sample path such that the approximation provided by their empirical distribution is close to optimal.
		A novel method is proposed, based on greedy minimisation of a kernel Stein discrepancy, that is suitable when the gradient of the log-target can be evaluated and approximation using a small number of states is required.
		Theoretical results guarantee consistency of the method and its effectiveness is demonstrated in the challenging context of parameter inference for ordinary differential equations.
		Software is available in the \verb+Stein Thinning+ package in \verb+Python+, \verb+R+ and \verb+MATLAB+.
		
		\vspace{5pt}
		\noindent
		Keywords: Bayesian computation, greedy optimisation, Markov chain Monte Carlo, reproducing kernel, Stein's method
	\end{abstract}	
	
	\section{Introduction}
	
	The most popular computational tool for non-conjugate Bayesian inference is Markov chain Monte Carlo (MCMC).
	Introduced to statistics from the physics literature in \cite{hastings1970,geman1984stochastic,tanner1987calculation,gelfand1990sampling}, an enormous amount of research effort has since been expended in the advancement of MCMC methodology.
	Such is the breadth of this topic that we do not attempt a survey here, but instead refer the reader to \cite{robert2013monte,green2015bayesian} and the references therein to more advanced material.
	This paper is motivated by the fact that the approaches used for convergence assessment and to post-process the output of MCMC can strongly affect the estimates that are produced. 
	
	Let $P$ be a distribution on a measurable space $\mathcal{X}$ and let $(X_i)_{i \in \mathbb{N}}$ be a Markov chain that is $P$-invariant.
	The Markov chain sample path provides an empirical approximation 
	\begin{align}
	\frac{1}{n} \sum_{i=1}^n \delta(X_i)  \label{eq: MCMC estimator}
	\end{align}
	to $P$, where $\delta(x)$ denotes a point mass centred at $x \in \mathcal{X}$. 
	Our discussion supposes that a practitioner is prepared to simulate a Markov chain up to a maximum number of iterations, $n$, and that simulating further iterations is not practical; a scenario that is often encountered (e.g. see \Cref{subsec: cardiac}).
	In this setting it is common (and indeed recommended) to replace \eqref{eq: MCMC estimator} with an alternative estimator 
	\begin{align}
	\frac{1}{m} \sum_{j=1}^m \delta(X_{\pi(j)}) \label{eq: subset estimator}
	\end{align}
	that is based on a subset of the total MCMC output.
	The $m$ indices $\pi(j) \in \{1,\dots,n\}$ indicate which states are retained and the identification of a suitable index set $\pi$ is informed by the following considerations:
	
	\vspace{5pt}
	\noindent
	\textbf{Removal of Initial Bias:} 
	The distribution of the initial states of the Markov chain may be quite different to $P$.
	To mitigate this, it is desirable to identify a ``burn-in'' $(X_i)_{i=1}^b$ which is then discarded.
	The burn-in period $b$ is typically selected using convergence diagnostics \citep{cowles1996}.
	These are primarily based on the empirical distribution of simple moment, quantile or density estimates across independent chains and making a judgement as to whether the ensemble of chains has converged to the distributional target. 
	The main limitation of convergence diagnostics, as far as we are concerned in this work, is that in taking $b$ large enough to make bias negligible, the number $n-b$ of remaining samples may be rather small, such that the statistical efficiency of the estimator in \eqref{eq: subset estimator} is sub-optimal as an approximation of $P$. 
	Nonetheless, a considerable portion of Bayesian pedagogy is devoted to the identification of the burn-in period,  as facilitated using diagnostic tests that are built into commercial-grade software such as \verb+WinBUGS+ \citep{lunn2000winbugs}, \verb+JAGS+ \citep{plummer2003jags}, \verb+R+ \citep{Rproject}, and \verb+Stan+ \citep{carpenter2017stan}.

	\vspace{5pt}
	\noindent
	\textbf{Increased Statistical Efficiency:}
	It is often stated that discarding part of the MCMC output leads to a reduction in the statistical efficiency of the estimator \eqref{eq: subset estimator} compared to \eqref{eq: MCMC estimator}.
	This argument, made e.g. in 
	\cite{geyer1992practical}, applies only when the procedure used to discard part of the MCMC output does not itself depend on the MCMC output and when the length $n$ of the MCMC output is fixed.
	That estimation efficiency can be \emph{improved} by discarding a portion of the samples in a way that depends on the samples themselves is in fact well-established \citep[see e.g.][]{dwivedi2019power}.

	\vspace{5pt}
	\noindent
	\textbf{Compression of MCMC Output:}
	A third motivation for estimators of the form \eqref{eq: subset estimator} is to control the cost of subsequent computation involving the MCMC output.
	Examples include approximating the expectation of a function $f$, where either evaluation of $f$ or storage of its output is associated with a computational cost, and \textit{Monte Carlo Maximum Likelihood}, where one constructs an approximate likelihood using MCMC, then performs optimisation on this approximate likelihood \citep{geyer1992constrained}.
	In such situations one may want to control the cardinality $m$ of the index set $\pi$ and to use $(X_{\pi(j)})_{j=1}^m$ as an experimental design on which $f$ is evaluated.
	The most popular solution is to retain only every $t^{\text{th}}$ state visited by the Markov chain, a
	procedure known as ``thinning'' of the MCMC output.
	See also the more sophisticated approach in \cite{paige2016super}.

	\vspace{5pt}
	
	Taking these considerations into account, the most common approach used to select an index set $\pi$ is based on the identification of a suitable burn-in period $b$ and/or a suitable thinning frequency $t$, leading to an approximation of the form 
	\begin{align}
	\frac{1}{\lfloor(n-b)/t\rfloor} \sum_{i=1}^{\lfloor(n-b)/t\rfloor} \delta(X_{b+it}) .  \label{eq: std post process}
	\end{align}
	Here $\lfloor r \rfloor$ denotes the integer part of $r$.
	This corresponds to a set of indices $\pi$ in \eqref{eq: subset estimator} that discards the burn-in states and retains only every $t^{\text{th}}$ iteration from the remainder of the MCMC output.
	It includes the case where no states are removed when $b =0$ and $t=1$. 
	Despite their widespread usage, the interplay between the Markov chain sample path and the heuristics used to select $b$ and $t$ is not widely appreciated.
	In general it is unclear how much bias may be introduced by employing a post-processing heuristic that is itself based on the MCMC output.
	Indeed, even the basic question of when the post-processed estimator in \eqref{eq: std post process} is consistent when $b$ and $t$ are chosen based on the MCMC output appears not to have been studied.
	
	In this paper we propose a novel method, called \verb+Stein Thinning+, that selects an index set $\pi$, of specified cardinality $m$, such that the associated discrete approximation in \eqref{eq: subset estimator} is close to optimal among all approximations supported on the MCMC output.
	The method is designed to ensure that \eqref{eq: subset estimator} is a consistent approximation of $P$. 
	This includes situations when the Markov chain on which it is based is not $P$-invariant, but we do of course require that the regions of high probability under $P$ are explored.
	To achieve this we adopt a kernel Stein discrepancy as our optimality criterion.
	The minimisation of kernel Stein discrepancy is performed using a greedy sequential algorithm and the main contribution of our theoretical analysis is to study the interplay of the greedy algorithm with the randomness inherent to the MCMC output.
	The proposed \verb+Stein Thinning+ method is simple (see \Cref{alg: the method}), applicable to most problems where gradients of the log-posterior density can be computed, and implemented as convenient \verb+Python+ and \verb+MATLAB+ packages that require no additional user input other than the number $m$ of states to be selected (see \Cref{subsec: software}).

	\subsection{Related Work}
	
	Our work contributes to an active area of research that attempts to cast post-processing of MCMC as an optimisation problem.
	\cite{Mak2018} proposed a method, called \texttt{Support Points}, which selects a small number of states in order that an empirical measure supported on those states minimises an ``energy distance'' to $P$.
	However, computation of the energy distance requires access to $P$, and minimisation of energy distance requires a challenging non-convex optimisation problem to be solved, meaning that in practice approximations are required.
	Stein discrepancy provides a computable alternative, which was used in 
	\cite{liu2016black}  to optimally weight an arbitrary set $(X_i)_{i=1}^n \subset \mathbb{R}^d$ of states in an manner loosely analogous to importance sampling, at a computational cost of $O(n^3)$.
	The combined effect of applying the approach of \cite{liu2016black} to MCMC output was analysed in \cite{Hodgkinson2020}, who established situations in which the overall procedure will be consistent.

	If a compressed representation of the posterior $P$ is required, but one is not wedded to the use of MCMC for generation of candidate states, then several other methods can be used.
	\cite{Joseph2015,Joseph2017} proposed a criterion to capture how well an empirical measure based on a point set approximates $P$ and applied repeated numerical optimisation over $\mathcal{X}$ to arrive at a suitable point set.
	A similar approach was taken in \cite{Chen2018SteinPoints}, where a Stein discrepancy was numerically minimised.
	The reliance of both of these algorithms on non-convex numerical optimisation over $\mathcal{X}$ renders their implementation and analysis difficult.
	\cite{Chen2019} considered using Markov chains to approximately perform numerical optimisation, allowing a tractable analytic treatment at the expense of a sub-optimal compression of $P$. 
	An elegant alternative approach is to formulate a convex optimisation problem on the set of probability distributions on $\mathcal{X}$.
	In this spirit, \cite{Liu2016SVGD,Liu2017GradientFlow} identified a gradient flow with $P$ as a fixed point that can be approximately simulated using a particle method.
	At convergence, one obtains a compressed representation of $P$, however the theoretical analysis of this approach remains an open and active research topic \citep[see e.g.][]{duncan2019geometry}.
	
	The present paper differs from the contributions cited, in that (1) our algorithm requires only the output from one run of MCMC, which is a realistic requirement in many situations, and (2) we are able to provide a finite sample size error bound (\Cref{thm: main theorem}) and a consistency guarantee (\Cref{cor: bias}) for \verb+Stein Thinning+, that cover precisely the algorithm that we implement.

	\subsection{Outline of the Paper}
	
	The paper proceeds, in \Cref{sec: methods}, to recall the construction of a kernel Stein discrepancy and to present \verb+Stein Thinning+.
	Then in \Cref{sec: theory} we establish a finite sample size error bound, as well as a widely-applicable consistency result that does not require the Markov chain to be $P$-invariant.
	In \Cref{sec: empirical} we present an empirical assessment of \verb+Stein Thinning+ in the context of parameter inference for ordinary differential equation models.
	Conclusions are contained in \Cref{sec: conclusion}.

	\section{Methods} \label{sec: methods}
	
	In this section we introduce and analyse \texttt{Stein Thinning}.
	First, in \Cref{subsec: KSD}, we recall the construction of a kernel Stein discrepancy and its theoretical properties. The \verb+Stein Thinning+ method is presented in \Cref{subsec: greedy alg}, whilst \Cref{subsec: kernel choice} is devoted to implementational detail.

	Before we proceed, we introduce a piece of notation that will often be used and recall the mathematical definition of a reproducing kernel:
	
	\vspace{5pt}
	\noindent
	\textbf{Notation:}
	Let $\mathcal{P}$ denote the set of probability distributions $P$ that admit a positive density $p$, with $\nabla \log p$ Lipschitz on $\mathbb{R}^d$.
	
	\vspace{5pt}
	\noindent
	\textbf{Reproducing Kernel:}
	A \textit{reproducing kernel Hilbert space} (RKHS) of functions on a set $\mathcal{X}$ is a Hilbert space, denoted $\mathcal{H}(k)$, equipped with a function $k: \mathcal{X} \times \mathcal{X} \rightarrow \mathbb{R}$, called a \emph{kernel}, such that $\forall x \in \mathcal{X}$ we have $k(\cdot,x) \in \mathcal{H}(k)$ and $\forall x \in \mathcal{X}, h \in \mathcal{H}(k)$ we have $h(x) = \langle h, k(\cdot,x)\rangle_{\mathcal{H}(k)}$.
	In this paper $\langle\cdot,\cdot\rangle_{\mathcal{H}(k)}$ denotes the inner product in $\mathcal{H}(k)$ and the induced norm will be denoted $\|\cdot\|_{\mathcal{H}(k)}$.
	For further details, see \cite{Berlinet2004}.
	
	\subsection{Kernel Stein Discrepancy} \label{subsec: KSD}
	
	To construct a criterion for the selection of states from the MCMC output we require a notion of optimal approximation for probability distributions.
	To this end, recall that an \textit{integral probability metric} (IPM) \citep{Muller1997}, based on a set $\mathcal{F}$ of measure-determining functions on a measurable space $\mathcal{X}$, is defined as
	\begin{equation}\label{eq:IPMs}
	D_{\mathcal{F}} (P,Q)
	\; := \;
	\sup_{f \in \mathcal{F}} \left| \int_{\mathcal{X}} f \mathrm{d}P - \int_{\mathcal{X}} f \mathrm{d}Q \right|  .
	\end{equation}
	The fact that $\mathcal{F}$ is measure-determining means that $D_{\mathcal{F}}(P,Q) = 0$ if and only if $P=Q$ is satisfied.
	Standard choices for $\mathcal{F}$, e.g. that recover Wasserstein distance as the IPM, cannot be used in the Bayesian context due to the need to compute integrals with respect to $P$ in \eqref{eq:IPMs}.
	
	In the remainder of \Cref{subsec: KSD} we restrict attention to the setting $P \in \mathcal{P}$.
	To circumvent intractability of \eqref{eq:IPMs}, the notion of a \emph{Stein discrepancy} was proposed in \cite{Gorham2015}.
	This was based on Stein's method \citep{Stein1972}, which consists of finding a set $\mathcal{G}$ of sufficiently differentiable $d$-dimensional vector fields and a differential operator $\mathcal{A}_P$, depending on $P$ and acting on elements of $\mathcal{G}$, such that $ \int_{\mathbb{R}^d} \mathcal{A}_P g \; \mathrm{d}P = 0$ for all $g \in \mathcal{G}$. 
	The proposal of \cite{Gorham2015} was to take $\mathcal{F} = \mathcal{A}_P \mathcal{G}$ to be the image of $\mathcal{G}$ under $\mathcal{A}_P$ in \eqref{eq:IPMs}, leading to the \textit{Stein discrepancy}
	\begin{equation}\label{eq:stein_discrepancy}
	D_{\mathcal{A_PG}}(P,Q)
	\; = \;
	\sup_{g \in \mathcal{G}} \left| \int_{\mathbb{R}^d} \mathcal{A}_P g \; \mathrm{d}Q \right|  .
	\end{equation}
	Theoretical analysis had led to sufficient conditions for $\mathcal{A}_P \mathcal{G}$ to be measure-determining \citep{Gorham2015}.
	In this paper we focus on a particular form of \eqref{eq:stein_discrepancy} due to \cite{Liu2016,Chwialkowski2016,Gorham2017}, called a \emph{kernel} Stein discrepancy (KSD).
	In this case, $\mathcal{A}_P$ is the \emph{Langevin Stein operator} $\mathcal{A}_P g \coloneqq p^{-1} \nabla \cdot (pg)$ derived in \citet{Gorham2015},
	where $\nabla \cdot$ denotes the divergence operator in $\mathbb{R}^d$ and $\mathcal{G} \coloneqq \{g : \mathbb{R}^d \rightarrow \mathbb{R}^d \vert \sum_{i=1}^d \|g_i\|_{\mathcal{H}(k)}^2 \leq 1\}$ is the unit ball in a Cartesian product of RKHS.
	It follows from construction that the set $\mathcal{A}_P \mathcal{G}$ is the unit ball of another RKHS, denoted $\mathcal{H}(k_P)$, whose kernel is 
	\begin{align}
	k_P(x,y) & \coloneqq \nabla_x \cdot \nabla_y k(x,y) + \left\langle \nabla_x k(x,y) , \nabla_y \log p(y) \right\rangle \nonumber \\
	& \quad + \left\langle \nabla_y k(x,y) , \nabla_x \log p(x) \right\rangle + k(x,y) \left\langle \nabla_{x} \log p(x) , \nabla_y \log p(y) \right\rangle , \label{eq:stein_kernel}
	\end{align}
	where $\langle \cdot , \cdot \rangle$ denotes the standard Euclidean inner product, $\nabla$ denotes the gradient operator and subscripts have been used to indicate the variables being acted on by the differential operators \citep{Oates2017}.
	Thus KSD is recognised as a maximum mean discrepancy in $\mathcal{H}(k_P)$ \citep{song2008learning} and is fully characterised by the kernel $k_P$; we therefore adopt the shorthand notation $D_{k_P}(Q)$ for $D_{\mathcal{A}_P \mathcal{G}}(P,Q)$.
	
	In the remainder of this section we recall the main properties of KSD.
	The first is a condition on the kernel $k$ that guarantees elements of $\mathcal{H}(k_p)$ have zero mean with respect to $P$.
	In what follows $\|x\| = \langle x , x \rangle^{1/2}$ denotes the Euclidean norm on $\mathbb{R}^d$.
	It will be convenient to abuse operator notation, writing $\nabla_x \nabla_y^\top k$ for the Hessian matrix of a bivariate function $(x,y) \mapsto k(x,y)$.
	
	\begin{proposition}[Proposition 1 of \cite{Gorham2017}]
		Let $P \in \mathcal{P}$ and assume that $\int_{\mathbb{R}^d} \| \nabla \log p \| \mathrm{d}P < \infty$.
		Let $(x,y) \mapsto \nabla_x \nabla_y^\top k(x,y)$ be continuous and uniformly bounded on $\mathbb{R}^d$.
		Then $\int_{\mathbb{R}^d} h \mathrm{d}P = 0$ for all $h \in \mathcal{H}(k_P)$, where $k_P$ is defined in \eqref{eq:stein_kernel}.
	\end{proposition}
	
	The second main property of KSD that we will need is that it can be explicitly computed for an empirical measure $Q = \frac{1}{n} \sum_{i=1}^n \delta(x_i)$, supported on states $x_i \in \mathbb{R}^d$:
	
	\begin{proposition}[Proposition 2 of \cite{Gorham2017}]
		Let $P \in \mathcal{P}$ and let $(x,y) \mapsto \nabla_{x} \nabla_{y}^\top k(x,y)$ be continuous on $\mathbb{R}^d$.
		Then
		\begin{align}
		D_{k_P}\left( \frac{1}{n} \sum_{i=1}^n \delta(x_i) \right) & = \sqrt{\frac{1}{n^2}\sum_{i,j = 1}^n k_P(x_i,x_j)} , \label{eq: KSD first}
		\end{align}
		where $k_P$ was defined in \eqref{eq:stein_kernel}.
	\end{proposition}
	
	The third main property is that KSD provides convergence control.
	Let $Q_n \Rightarrow P$ denote weak convergence of a sequence $(Q_n)$ of measures to $P$.
	Theoretical analysis in \cite{Gorham2017,Chen2018SteinPoints,Huggins2018,Chen2019,Hodgkinson2020,Gorham2020} established sufficient conditions for when convergence of \eqref{eq: KSD first} to zero implies $\frac{1}{n} \sum_{i=1}^n \delta(x_i) \Rightarrow P$. 
	For our purposes we present one such result, from \cite{Chen2019}.

	\begin{proposition}[Theorem 4 in \cite{Chen2019}] \label{prop: convergence control}
		Let $P \in \mathcal{P}$ be \emph{distantly dissipative}, meaning that $\liminf_{r \rightarrow \infty} \kappa(r) > 0$ where
		$$
		\kappa(r) := \inf\left\{ -2 \frac{ \langle \nabla \log p(x) - \nabla \log p(y) , x - y \rangle }{\|x-y\|^2} : \|x-y\| = r \right\} .
		$$
		Consider the kernel $k(x,y) = (c^2 + \|\Gamma^{-1/2} (x-y)\|^2)^\beta$ for some fixed $c > 0$, a fixed positive definite matrix $\Gamma$ and a fixed exponent $\beta \in (-1,0)$.
		Then $D_{k_P}\left( \frac{1}{n} \sum_{i=1}^n \delta(x_i) \right) \rightarrow 0$ implies $\frac{1}{n} \sum_{i=1}^n \delta(x_i) \Rightarrow P$,
		where $k_P$ is defined in \eqref{eq:stein_kernel}.
	\end{proposition}
	
	\noindent The properties just described ensure that KSD is a suitable optimality criterion to consider for the post-processing of MCMC output.
	However, all discrepancies are associated with finite sample size pathologies; see \citet[][Section 3.4]{matsubara2021robust} for a discussion of the pathologies of KSD.
	Our attention turns next to the development of algorithms for minimisation of KSD.

	\subsection{Greedy Minimisation of KSD} \label{subsec: greedy alg}
	
	The convergence control afforded by \Cref{prop: convergence control} motivates the design of methods that select points $(x_i)_{i=1}^n$ such that \eqref{eq: KSD first} is approximately minimised.
	Continuous optimisation algorithms were proposed for this task in \cite{Chen2018SteinPoints} and \cite{Chen2019}.
	In \cite{Chen2018SteinPoints}, deterministic optimisation techniques were considered for low-dimensional problems, whereas in \cite{Chen2019} a Markov chain was used to provide more a practical optimisation strategy when the state space is high-dimensional.
	In each case greedy sequential strategies were considered, wherein at iteration $n$ a new state $x_n$ is appended to the current sequence $(x_1,\dots,x_{n-1})$.
	\cite{Chen2018SteinPoints} also considered the use of conditional gradient algorithms (so-called \emph{Frank-Wolfe}, or \emph{kernel herding} algorithms) but found that greedy algorithms provided better performance across a range of experiments and therefore we focus on greedy algorithms in this manuscript.
	
	The present paper is distinguished from earlier work in that we do not attempt to solve a continuous optimisation problem for selection of the next point $x_n \in \mathcal{X}$. 
	Such optimisation problems are fundamentally difficult and can at best be approximately solved.
	Instead, we exactly solve the discrete optimisation problem of selecting a suitable element $x_n$ from supplied MCMC output.
	In this sense we expect our findings will be more widely applicable than previous work, since we are simply performing post-processing of MCMC output and there exists a variety of commercial-grade software for MCMC.
	The method that we propose, called \verb+Stein Thinning+, is straight-forward to implement, and is stated in \Cref{alg: the method} for a distribution $P$ on a general measurable space $\mathcal{X}$.
	(The convention $\sum_{i=1}^0 = 0$ is employed.)
	
	\begin{algorithm}[h!]
		\KwData{The output $(x_i)_{i=1}^n$ from an MCMC method, a kernel $k_P$ for which the conclusion of \Cref{prop: convergence control} holds, and a desired cardinality $m \in \mathbb{N}$.}
		\KwResult{The indices $\pi$ of a sequence $(x_{\pi(j)})_{j=1}^m \subset \{x_i\}_{i=1}^n$ where the $\pi(j)$ are elements of $\{1,\dots,n\}$. }
		\For{$j = 1,\dots, m$}{
			$\pi(j) \; \in \; \displaystyle\argmin_{i = 1,\dots,n} \; \frac{k_P(x_{i},x_{i})}{2} + \sum_{j'=1}^{j-1} k_P(x_{\pi(j')},x_{i}) $\;
		}
		\cprotect\caption{The proposed method; \verb+Stein Thinning+.}
		\label{alg: the method}
	\end{algorithm}
	
	The algorithm is illustrated on a simple bivariate Gaussian mixture in Figure \ref{fig: illustration}.
	Observe in this figure that the points selected by the \verb+Stein Thinning+ do not belong to the burn-in period (which is visually clear), and that although the MCMC spent a disproportionate amount of time in one of the mixture components, the number of points selected by \verb+Stein Thinning+ is approximately equal across the two components of the target.
	The accuracy of the approximation produced by \texttt{Stein Thinning} is, nevertheless, gated by the quality of the MCMC output to which it is applied.
	A detailed empirical assessment is presented in \Cref{sec: empirical}.
	
	\begin{remark}[Tie-breaking]
		In the event of a tie, a tie-breaking rule should be used to select the next index.
		For example, if the minimum in \Cref{alg: the method} is realised by multiple candidate values $\Pi(j) \subseteq \{1,\dots,n\}$, one could adopt a tie-breaking rule that selects the smallest element of $\Pi(j)$ as the value that is assigned to $\pi(j)$.
		The rule that is used has no bearing on our theoretical analysis in \Cref{sec: theory}.
	\end{remark}
	
	\begin{remark}[Complexity]
		The computation associated with iteration $j$ of \Cref{alg: the method} is $O(n r_j)$ where $r_j \leq \min(j,n)$ is the number of distinct indices in $\{\pi(1),\dots,\pi(j-1)\}$; the computational complexity of \Cref{alg: the method} is therefore $O(n \sum_{j=1}^m r_j)$.
		For typical MCMC algorithms the computational complexity is $O(n)$, so the complexity of \verb+Stein Thinning+ is equal to that for MCMC when $m$ is fixed and higher when $m$ is increasing with $n$, being at most $O(nm^2)$.
	\end{remark}
	
	\begin{remark}[Re-sampling]
		In general the indices in $\pi$ need not be distinct.
		That is, \Cref{alg: the method} may prefer to include a duplicate state rather than to include a state which is not useful for representing~$P$.
		Indeed, if $m > n$ then the sequence $(x_{\pi(j)})_{j=1}^m$ must contain duplicates entries.
		\Cref{prop: optimal converge} in \Cref{sec: theory} clarifies this behaviour.
	\end{remark}
	
	\begin{remark}[Finite sample error bound] \label{rem: error bound}
		The approximation produced by \texttt{\emph{Stein Thinning}} satisfies a finite sample error bound
		$$
		\left| \frac{1}{m} \sum_{j=1}^m f(x_{\pi(j)}) - \int_{\mathbb{R}^d} f(x) \mathrm{d}P(x) \right| \leq D_{k_P}\left( \frac{1}{m} \sum_{j=1}^m \delta(x_{\pi(j)}) \right) \left\| f - \int_{\mathbb{R}^d} f(x) \mathrm{d}P(x) \right\|_{\mathcal{H}(k_P)} 
		$$
		following \citet[][Equation (1.14) applied to $f - \int f \mathrm{d}P$]{hickernell1998generalized}.
		This can be contrasted with the typically asymptotic analysis of MCMC.
		The practical estimation of the final term in this bound was discussed in Section 4 of \cite{south2020semi}.
	\end{remark}

	\begin{figure}[t!]
		\centering
		\begin{subfigure}[t]{0.328\textwidth}
			\centering
			\includegraphics[width=\textwidth,clip,trim = 4.27cm 3.4cm 1.2cm 1.4cm]{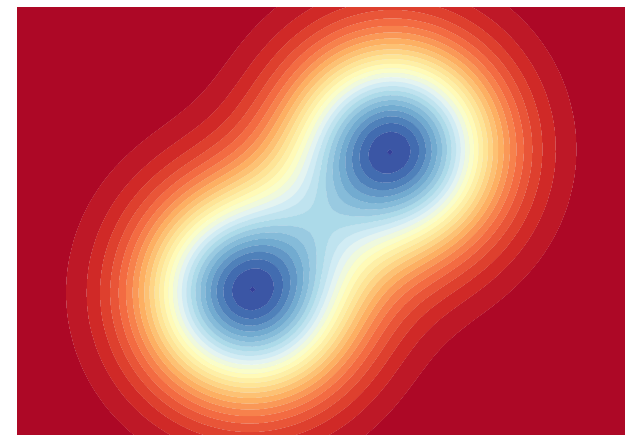}
			\caption{}
		\end{subfigure}%
		\begin{subfigure}[t]{0.33\textwidth}
			\centering
			\includegraphics[width=\textwidth,clip,trim = 3.75cm 3.37cm 1cm 1.1cm]{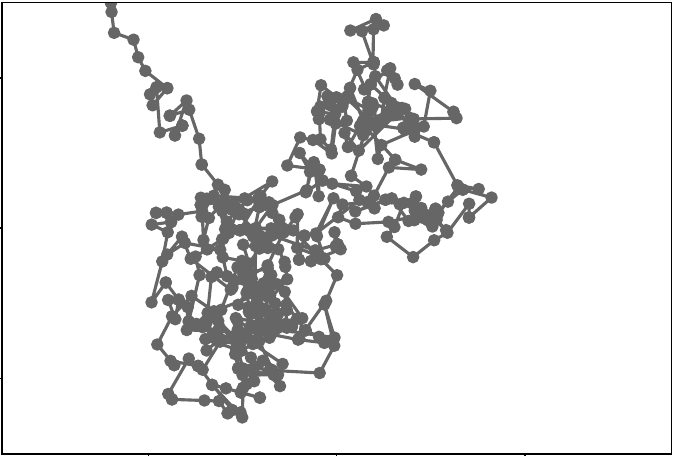}
			\caption{}
		\end{subfigure}%
		\begin{subfigure}[t]{0.33\textwidth}
			\centering
			\includegraphics[width=\textwidth,clip,trim = 3.75cm 3.37cm 1cm 1.1cm]{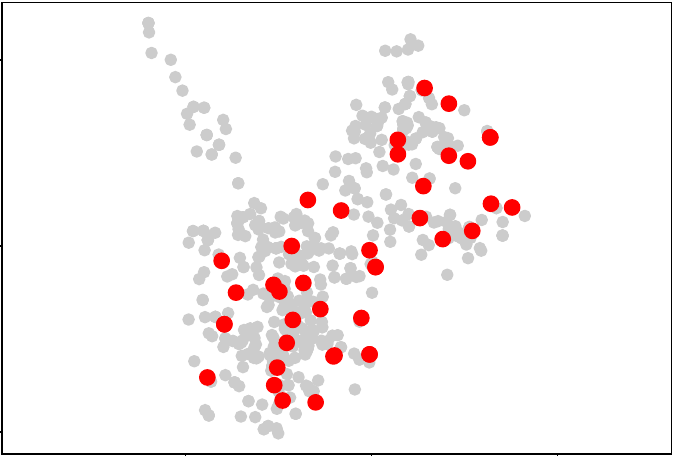}
			\caption{}
		\end{subfigure}%
		\cprotect\caption{Illustration of \verb+Stein Thinning+:
			(a)~Contours of the distributional target $P$.
			(b)~Markov chain Monte Carlo (MCMC) output, limited to 500 iterations to mimic a challenging computational context, exhibiting burn-in and autocorrelation that must be identified and mitigated.
			(c)~A subset of $m=40$ states from the MCMC output selected using \verb+Stein Thinning+, which correctly ignores the burn-in period and stratifies states approximately equally across the two components of the target.
		}
		\label{fig: illustration}
	\end{figure}

		\subsection{Choice of Kernel} \label{subsec: kernel choice}
	
	The suitability of KSD to quantify how well $Q$ approximates $P$ is determined by the choice of the kernel $k$ in \eqref{eq:stein_kernel}.
	Several choices are possible and for $P \in \mathcal{P}$, based on \Cref{prop: convergence control} together with extensive empirical assessment, \cite{Chen2019} advocated the pre-conditioned inverse multi-quadric kernel $	k(x,y) \coloneqq (1 + \|\Gamma^{-1/2}(x-y)\|^2 )^{-1/2}$
	where, compared to \Cref{prop: convergence control}, we have fixed $c=1$ (without loss of generality) and $\beta = -1/2$.
	The suitability of these choices for \texttt{Stein Thinning} is verified in \Cref{ap: hyper param vary}.
	The positive definite matrix $\Gamma$ remains to be specified and it is natural to take a data-driven approach where the MCMC output is used to select $\Gamma$.
	Provided that a fixed number $n_0 \in \mathbb{N}$ of the states $(X_i)_{i = 1}^{n_0}$ from the MCMC output are used in the construction of $\Gamma$, the consistency results for \verb+Stein Thinning+ that we establish in \Cref{sec: theory} are not affected.
	To explore different strategies for the selection of $\Gamma$, we focus on the following candidates:
	\begin{itemize}
		\item \textbf{Median (\texttt{med}):} The scaled identity matrix $\Gamma = \ell^{2} I$, where $\ell = \text{med} := \text{median}\{\|X_i-X_j\| : 1 \leq i < j \leq n_0 \}$	is the median Euclidean distance between states \citep{Garreau2018}.
		In the rare case that $\text{med} = 0$, an exception should be used, such as $\ell = 1$, to ensure a positive definite $\Gamma$ is used.
		\item \textbf{Scaled median (\texttt{sclmed}):} The scaled identity matrix $\Gamma = \ell^2 I$, where $\ell = \text{med} / \sqrt{\log(m)}$.
		This was proposed in \cite{Liu2016SVGD} and can be motivated using the approximation $\sum_{j'=1}^m k_P(x_{\pi(j)},x_{\pi(j')}) \approx m \exp(- \ell^{-2} \text{med}^2 ) = 1$.
		Note the dependence on $m$ means that the preceding theoretical analysis does not apply when this heuristic is used.
		\item \textbf{Sample covariance (\texttt{smpcov}):} The matrix $\Gamma$ can be taken as a sample covariance matrix 
		$$
		\Gamma = \frac{1}{n_0-1} \sum_{i=1}^{n_0} \left( X_i - \bar{X} \right) \left( X_i - \bar{X} \right)^\top, \qquad \bar{X} := \frac{1}{n_0} \sum_{i=1}^{n_0} X_i ,
		$$
		provided that this matrix is non-singular.
	\end{itemize}
	
	\noindent The experiments in \Cref{sec: empirical} shed light on which of these settings is the most effective, but we acknowledge that many other settings could also be considered.
	In what follows, we set $n_0 = \min(n,10^3)$ for the \texttt{med} and \texttt{sclmed} settings, to avoid an $O(n^2)$ cost of computing $\ell$, and otherwise set $n_0 = n$, so that the whole of the MCMC output is used to select $\Gamma$.
	\verb+Python+, \verb+R+ and \verb+MATLAB+ packages are provided and their usage is described in \Cref{subsec: software}.

	\section{Theoretical Assessment} \label{sec: theory}
	
	The theoretical analysis in this section clarifies the limiting behaviour of \verb+Stein Thinning+ as $m,n \rightarrow \infty$.
	Our first main result concerns the behaviour of \verb+Stein Thinning+ on a fixed sequence $(x_i)_{i=1}^n$:
	
	\begin{theorem} \label{prop: optimal converge}
		Let $\mathcal{X}$ be a measurable space and let $P$ be a probability distribution on $\mathcal{X}$.
		Let $k_P : \mathcal{X} \times \mathcal{X} \rightarrow \mathbb{R}$ be a reproducing kernel with $\int_{\mathcal{X}} k_P(x,\cdot) \mathrm{d}P = 0$ for all $x \in \mathcal{X}$.
		Let $(x_i)_{i=1}^n \subset \mathcal{X}$ be fixed and consider an index sequence $\pi$ of length $m$ produced by \Cref{alg: the method}.
		Then we have the bound
		\begin{align*}
		D_{k_P}\left( \frac{1}{m} \sum_{j=1}^m \delta(x_{\pi(j)}) \right)^2 \leq D_{k_P}\left( \sum_{i=1}^n w_i^* \delta(x_i) \right)^2 + \left( \frac{1 + \log(m)}{m} \right) \max_{i=1,\dots,n} k_P(x_i,x_i) ,
		\end{align*}
		where the weights $w^* = (w_1^*,\dots,w_n^*)$ in the first term satisfy
		\begin{align}
		w^* \in \argmin_{\substack{1_n^\top w = 1 \\ w \geq 0}} D_{k_P} \left(\sum_{i=1}^n w_i \delta(x_i)\right)  \label{eq: optimal weights}
		\end{align}
		where $1_n^\top = (1,\dots,1)$ and $w \geq 0$ indicates that $w_i \geq 0$ for $i = 1,\dots,n$.
	\end{theorem}
	
	\noindent The proof of \Cref{prop: optimal converge} is provided in \Cref{ap: thm 1 proof}.
	Its implication is that, given a sequence $(x_i)_{i=1}^n$, \verb+Stein Thinning+ produces an empirical distribution that converges in KSD to the optimal weighted empirical distribution $\sum_{i=1}^n w_i^* \delta(x_i)$ based on that sequence.
	Properties of such optimally weighted empirical measures were studied in \cite{liu2016black,Hodgkinson2020}, and are not the focus of the present paper, where the case $m \ll n$ is of principal interest.
	
	The role of \Cref{prop: optimal converge} is to study the interaction between the greedy algorithm and a given sequence $(x_i)_{i=1}^n$, and this bound is central to our proof of \Cref{thm: main theorem} which deals with the case where $(x_i)_{i=1}^n$ is replaced by MCMC output.
	\Cref{fig:thoreom1} illustrates the terms involved in \Cref{prop: optimal converge}. 
	It is clear that a reduction in KSD is achieved by \texttt{Stein Thinning} of the MCMC output. 
	
	\begin{remark}[Optimal weights] \label{rem: optimal weights}
	    To further improve the empirical approximation, we can consider an optimally-weighted sum $\sum_{j=1}^m w_{m,j}^* \delta(x_{\pi(j)})$ where the $w_{m,j}^*$ solve a convex optimisation problem analogous to \eqref{eq: optimal weights}.
	    Such weights minimise a quadratic function subject to a linear and a non-negativity constraint and can therefore be precisely computed.
		If the non-negativity constraint is removed and the indices in $\pi$ are distinct then
		\begin{align*}
		v_m^* := \argmin_{1_m^\top v = 1} D_{k_P} \left( \sum_{j=1}^m v_j \delta(x_{\pi(j)}) \right) = \frac{K_P^{-1} 1_m}{1_m^\top K_P^{-1} 1_m}, \qquad (K_P)_{i,j} \coloneqq k_P(x_{\pi(i)},x_{\pi(j)}),
		\end{align*}
		as derived in \cite{Oates2017}. 
		\Cref{fig:thoreom1} indicates that the benefit of applying weights $w_m^*$ (red curve) to the output of \verb+Stein Thinning+ (black curve) is limited, likely because the $x_{\pi(j)}$ were selected in a way that avoids redundancy in the point set.
		A larger improvement is provided by the weights $v_m^*$ (blue curve), but in this case the associated empirical measure may not be a probability distribution.
	\end{remark}

	\begin{remark}
		The use of a conditional gradient algorithm, instead of a greedy algorithm, in this context amounts to simply removing the term $k_P(x_{\pi(j)}, x_{\pi(j)})$ in \Cref{alg: the method}.
		As discussed in \cite{Chen2018SteinPoints}, this term can be thought of as a regulariser that lends stability to the algorithm, avoiding selection of $x_i$ that are far from the effective support of $P$.
	\end{remark}
	
	\begin{remark}
		\Cref{prop: optimal converge} is formulated at a high level of generality and can be applied on non-Euclidean domains $\mathcal{X}$.
		In \cite{Barp2018RS,liu2018riemannian,xu2020stein,le2020diffusion} the authors proposed and discussed Stein operators $\mathcal{A}_P$ for the non-Euclidean context.
	\end{remark}

	\begin{figure}[t!]
		\centering
		\includegraphics[width=0.7\textwidth,clip,trim = 0cm 0.5cm 0cm 0cm]{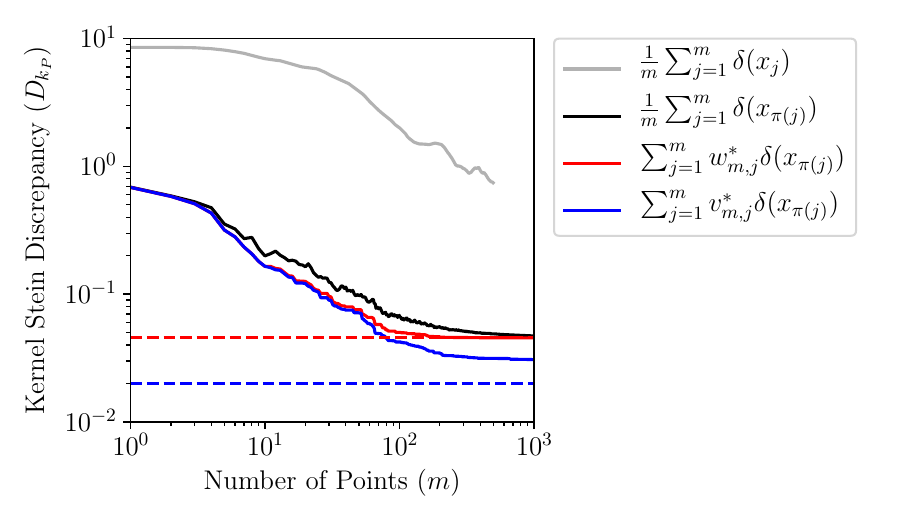}
		\cprotect\caption{Illustration of \Cref{prop: optimal converge}: 
			The gray curve represents the unprocessed output from MCMC in the example of \Cref{fig: illustration}.
			The black curve represents \texttt{Stein Thinning} applied to this same output and, in addition, weighted output of \texttt{Stein Thinning} is shown for weights $w_m^*$ subject subject to $\sum w^*_{m,j} = 1$ and $w^*_{m,j} \geq 0$ (solid red) and weights $v_m^*$ subject only to $\sum v_{m,j}^* = 1$ (solid blue).
			The dashed horizontal lines are the limiting values of their corresponding solid lines as the number $m$ is increased. 
		}
		\label{fig:thoreom1}
	\end{figure}

	Next we consider the properties of \verb+Stein Thinning+ applied to MCMC output.
	Let $V$ be a function $V: \mathcal{X} \to [1,\infty)$ and, for a function $f : \mathcal{X} \rightarrow \mathbb{R}$ and a measure $\mu$ on $\mathcal{X}$, let $\|f\|_V := \sup_{x \in \mathcal{X}} \frac{|f(x)|}{V(x)}$, $\|\mu\|_V := \sup_{\|f\|_V \leq 1} \left| \int_{\mathcal{X}} f \mathrm{d}\mu \right|$.
	Recall that a $\psi$-irreducible and aperiodic Markov chain $(X_i)_{i \in \mathbb{N}} \subset \mathcal{X}$ with $n^{\text{th}}$ step transition kernel $\mathrm{P}^n$ is \emph{$V$-uniformly ergodic} \citep[see Theorem 16.0.1 of][]{Meyn2012} if and only if $\exists R \in [0,\infty), \rho \in (0,1)$ such that 
	\begin{align}
	\| \mathrm{P}^n(x,\cdot)-P \|_{V} \leq R V(x) \rho^n \label{eq: V unif ergod}
	\end{align}
	for all initial states $x \in \mathcal{X}$ and all $n \in \mathbb{N}$.
	The notation $\mathbb{E}$ will be used to denote expectation with respect to the law of the Markov chain in the sequel.
	\Cref{thm: main theorem} establishes a finite sample size error bound for \verb+Stein Thinning+ applied to MCMC output:
	
	\begin{theorem} \label{thm: main theorem}
		Let $\mathcal{X}$ be a measurable space and let $P$ be a probability distribution on $\mathcal{X}$.
		Let $k_P : \mathcal{X} \times \mathcal{X} \rightarrow \mathbb{R}$ be a reproducing kernel with $\int_{\mathcal{X}} k_P(x,\cdot) \mathrm{d}P = 0$ for all $x \in \mathcal{X}$.
		Consider a $P$-invariant, time-homogeneous Markov chain $(X_i)_{i \in \mathbb{N}} \subset \mathcal{X}$ generated using a $V$-uniformly ergodic transition kernel, such that \eqref{eq: V unif ergod} is satisfied with $V(x) \geq \sqrt{k_P(x,x)}$ for all $x \in \mathcal{X}$.
		Suppose that, for some $\gamma > 0$,
		\begin{align*}
		b := \sup_{i \in \mathbb{N}} \mathbb{E} \left[ e^{\gamma k_P(X_i,X_i)} \right] < \infty, \qquad M := \sup_{i \in \mathbb{N}} \mathbb{E}\left[\sqrt{k_P(X_i,X_i)} V(X_i)\right] < \infty .
		\end{align*}
		Let $\pi$ be an index sequence of length $m$ produced by \Cref{alg: the method} applied to the Markov chain output $(X_i)_{i=1}^n$.
		Then, with $C = \frac{2 R \rho}{1-\rho}$, we have that
		\begin{align}
		\mathbb{E}\left[ D_{k_P}\left( \frac{1}{m} \sum_{j=1}^m \delta(X_{\pi(j)}) \right)^2 \right] \leq \frac{\log(b)}{\gamma n} + \frac{CM}{n} + \left( \frac{1 + \log(m)}{m} \right) \frac{\log(nb)}{\gamma} .
		\label{eq: main bound theorem}
		\end{align}
	\end{theorem}
	
	\noindent The proof of \Cref{thm: main theorem} is provided in Appendix \ref{ap: proof of main theorem}.
	
	\begin{remark}
		The upper bound in \eqref{eq: main bound theorem} is asymptotically minimised when (up to log factors) $m$ is proportional to $n$.
		In practice we are interested in the case $m \ll n$, so we may for example set $m = \lfloor \frac{n}{1000} \rfloor$ if we aim for substantial compression.
		It is not claimed that the bound in \eqref{eq: main bound theorem} is tight and indeed empirical results in \Cref{sec: empirical} endorse the use of \verb+Stein Thinning+ in the small $m$ context.
	\end{remark}
	
	\begin{remark}  \label{rem: iso kernel rem}
		For $P \in \mathcal{P}$ and $k_P$ in \eqref{eq:stein_kernel}, based on a radial kernel $k$, meaning that $k(x,y) = \phi(x-y)$ for some function $\phi : \mathbb{R}^d \rightarrow \mathbb{R}$ satisfying $\nabla \phi(0) = 0$, we have that $k_P(x,x) = - \Delta \phi(0) + \phi(0) \|\nabla \log p(x) \|^2$.
		The function $x \mapsto \sqrt{k_P(x,x)}$ appearing in the preconditions of \Cref{thm: main theorem} can therefore be understood in terms of $\|\nabla \log p(x)\|$.
		Further discussion of the preconditions of \Cref{thm: main theorem} is provided in \Cref{ap: satisfy conditions}. 
	\end{remark}
	
	Since convergence in mean-square does not in general imply almost sure convergence, we next strengthen the conclusions of \Cref{thm: main theorem}.
	Our final result, \Cref{cor: bias}, therefore establishes an almost sure convergence guarantee for \verb+Stein Thinning+.
	Furthermore, the result that follows applies also in the ``biased sampler'' case, where $(X_i)_{i \in \mathbb{N}}$ is a $Q$-invariant Markov chain and $Q$ need not equal $P$:
	
	\begin{theorem} \label{cor: bias}
		Let $Q$ be a probability distribution on $\mathcal{X}$ with $P$ absolutely continuous with respect to $Q$.
		Consider a $Q$-invariant, time-homogeneous Markov chain $(X_i)_{i \in \mathbb{N}} \subset \mathcal{X}$ generated using a $V$-uniformly ergodic transition kernel, such that $V(x) \geq \frac{\mathrm{d}P}{\mathrm{d}Q}(x) \sqrt{k_P(x,x)}$.
		Suppose that, for some $\gamma > 0$,
		\begin{align*}
		b := \sup_{i \in \mathbb{N}} \mathbb{E} \left[ e^{\gamma \max\left( 1 ,  \frac{\mathrm{d}P}{\mathrm{d}Q}(X_i)^2 \right) k_P(X_i,X_i)} \right] < \infty, \quad M := \sup_{i \in \mathbb{N}} \mathbb{E}\left[ \frac{\mathrm{d}P}{\mathrm{d}Q}(X_i) \sqrt{k_P(X_i,X_i)} V(X_i)\right] < \infty .
		\end{align*}
		Let $\pi$ be an index sequence of length $m$ produced by \Cref{alg: the method} applied to the Markov chain output $(X_i)_{i=1}^n$.
		If $m \leq n$ and the growth of $n$ is limited to at most $\log(n) = O(m^{\beta/2})$ for some $\beta < 1$, then $D_{k_P} \big( \frac{1}{m} \sum_{j=1}^m \delta(X_{\pi(j)}) \big) \rightarrow 0$
		almost surely as $m,n \rightarrow \infty$.
		Furthermore, if the preconditions of \Cref{prop: convergence control} are satisfied, then $\frac{1}{m} \sum_{j=1}^m \delta(X_{\pi(j)}) \Rightarrow P$ 		almost surely as $m,n \rightarrow \infty$.
	\end{theorem}
	\noindent The proof of \Cref{cor: bias} is provided in \Cref{app: cor bias proof}. 
	The interpretation of \Cref{cor: bias} is that one may sample states from a Markov chain that is not $P$-invariant and yet, under the stated assumptions (which ensure that regions of high probability under $P$ are explored), one can use \texttt{Stein Thinning} to still obtain a consistent approximation of $P$.
	This can be contrasted, for example, with the \texttt{Support Points} method of \cite{Mak2018}, which relies on $P$ being well-approximated by the MCMC output.
	This completes our theoretical analysis of \texttt{Stein Thinning}.

	\section{Empirical Assessment} \label{sec: empirical}
	
	In this section we compare the performance of \verb+Stein Thinning+ with existing methods for post-processing MCMC output.
	Our motivation derives from a problem in which we must infer a 38-dimensional parameter in a calcium signalling model defined by a stiff system of 6 coupled ordinary differential equations (ODEs).
	Posterior uncertainty is required to be propagated through a high-fidelity simulation in a multi-scale and multi-physics model $f$ of the human heart.
	Here, compression of the MCMC output can be used to construct an approximately optimal experimental design on which $f$ can be evaluated.
	The calcium model is, however, unsuitable for conducting a thorough {\it in silico} assessment due to its associated computational cost.
	Therefore in \Cref{subsec: Goodwin} we first consider a simpler ODE model, where $P$ can be accurately approximated.
	Then, as an intermediate example, in \Cref{subsec: Lotka} we consider an ODE model that induces stronger correlations among the parameters in $P$, before addressing the calcium model in \Cref{subsec: cardiac}.
	
	In \Cref{ap: MCMC methods} we describe the generic structure of a parameter inference problem for ODEs.
	In all instances the aim is to post-process the output from MCMC, in order to produce an accurate empirical approximation of the posterior supported on a small number $m \ll n$ of the states that were visited.
	The following methods were compared:
	\begin{itemize}
		\item The standard approach, which estimates a burn-in period using either the \textit{GR diagnostic} $\hat{b}^{\text{GR},L}$, $L > 1$, of \cite{gelman1992, brooks1998general, gelman2014} or the more sophisticated \textit{VK diagnostic} $\hat{b}^{\text{VK},L}$, $L \geq 1$, of \cite{vats2018revisiting}, in each case based on $L$ independent chains as described in \Cref{section: background on diagnostics}, followed by thinning as per \eqref{eq: std post process}.
		\item The \texttt{Support Points} algorithm proposed in \cite{Mak2018}, implemented in the \verb+R+ package \verb+support+. 
		\item The \verb+Stein Thinning+ algorithm that we have proposed, with each of the kernel choices described in \Cref{subsec: kernel choice}.
	\end{itemize}
	
	To ensure that our empirical findings are not sensitive to the choice of MCMC method, we implemented four Metropolis--Hastings samplers that differ qualitatively according to the sophistication of their proposal.
    These were: (i) the Gaussian random walk (\texttt{RW}); (ii) the adaptive Gaussian random walk (\texttt{ADA-RW}), which uses an estimate of the covariance of the target \citep{haario1999adaptive}; (iii) the Metropolis-adjusted Langevin algorithm (\texttt{MALA}), which takes a step in the direction of increasing Euclidean gradient, perturbed by Gaussian noise \citep{Roberts1996}; (iv)  the preconditioned version of \texttt{MALA} (\texttt{P-MALA}), which employs a preconditioner based on the Fisher information matrix \citep{Girolami2011}. 
    Full details are in \Cref{ap: MCMC methods}.
    Metropolis--Hastings algorithms were selected on the basis that we were able to successfully implement them on the challenging calcium signalling model in \Cref{subsec: cardiac}, which required manually interfacing with the numerical integrator to produce reliable output.

	\subsection{Goodwin Oscillator} \label{subsec: Goodwin}
	
	The first example that we consider is a negative feedback oscillator due to \cite{Goodwin1965}. 
	The ODE model and the associated $d=4$ dimensional inference problem are described in \Cref{app:addtional_Goodwin}, where one trace plot for each MCMC method, of length $n=2 \times 10^6$, are presented in \Cref{fig:Goodwin_traceplots_z}.

	First we consider the standard approach to post-processing MCMC output, as per \eqref{eq: std post process}.
	From the trace plots in \Cref{fig:Goodwin_traceplots_z}, it is clear that a burn-in period $b > 0$ is required. 
	For each method we therefore computed the GR and VK diagnostics, to arrive at candidate values $b$ for the burn-in period. 
	Default settings were used for all diagnostics, which were computed both for the multivariate $d$-dimensional state vector and for the univariate marginals, as reported in \Cref{app:addtional_Goodwin}.
	The GR diagnostics were computed using $L=6$ independent chains and the VK diagnostics were computed using both $L=1$ and $L=6$ independent chains; note that when $L > 1$, these diagnostics have access to more information in comparison with \verb+Stein Thinning+, in terms of the number of samples that are available to the method.  
	The estimated values for the burn-in period are reported in \Cref{app:addtional_Goodwin}, \Cref{table:burnin-Goodwin}.
	For all MCMC methods, neither the univariate nor the multivariate GR diagnostics were satisfied, so that $\hat{b}^{\text{GR},6} > n$ and estimation using \eqref{eq: std post process} cannot proceed.
	The VK diagnostic produced values $\hat{b}^{\text{VK},L} < n$, which typically led to about half of the MCMC output being discarded.
	Although well-suited for their intended task of minimising bias in MCMC output, the smaller number of states left after burn-in removal may lead to inefficient approximation of $P$ and derived quantities of interest, strikingly so in the case of the GR diagnostic.
	The use of an optimality criterion enables \verb+Stein Thinning+ to directly address this bias-variance trade-off.
	Of course, one can in principle run more iterations of MCMC to provide more diversity in the remainder of the sample path after burn-in is removed, but in applications such as the calcium model of \Cref{subsec: cardiac} the computational cost associated with each iteration presents a practical limitation in running more iterations of an MCMC method.
	Effective methods to post-process limited output (or, equivalently, a long output from a poorly mixing Markov chain) are therefore important.
	
		\begin{figure}[t!]
		\centering
		\includegraphics[width = 0.4\textwidth]{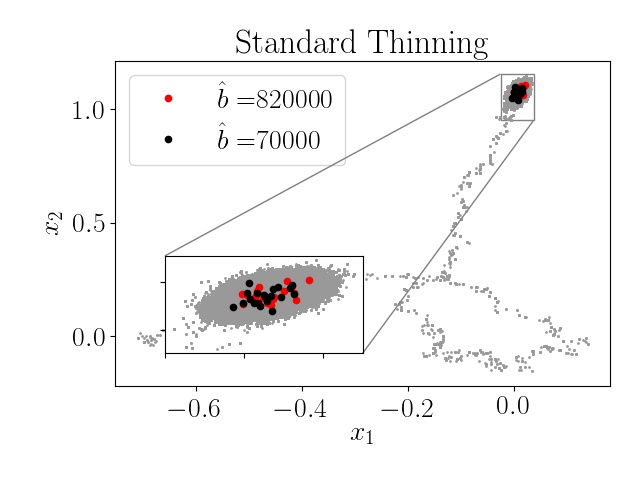}		
		\includegraphics[width = 0.4\textwidth]{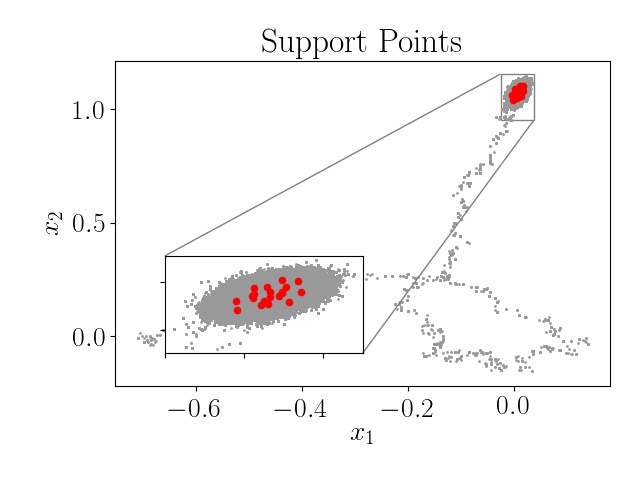}	
		\includegraphics[width = 0.4\textwidth]{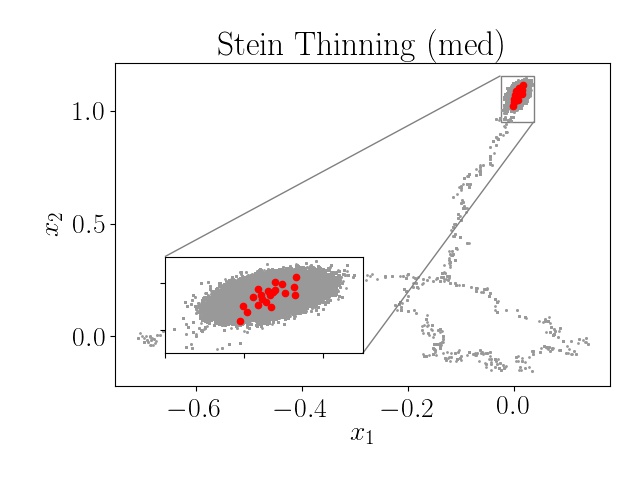}
		\includegraphics[width = 0.4\textwidth]{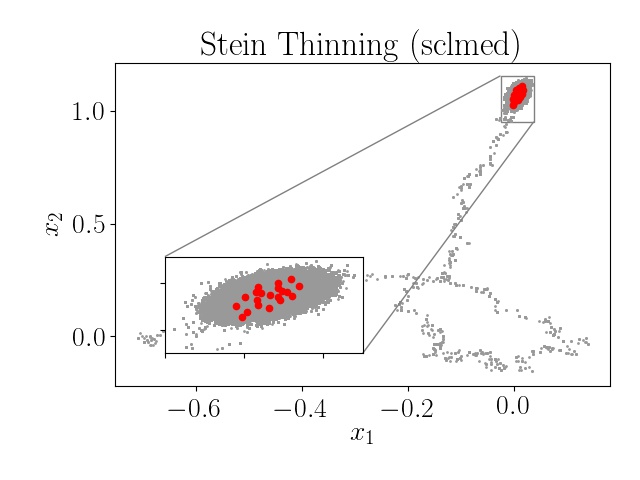}		
		\includegraphics[width = 0.4\textwidth]{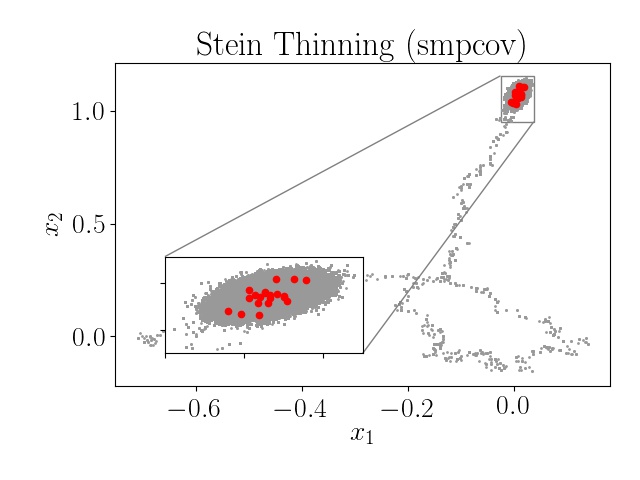}	
		
		\caption{Projections on the first two coordinates of the \texttt{RW} MCMC output from the Goodwin oscillator (gray dots), together with the first $m=20$ points selected using: the standard approach of discarding burn-in and thinning the remainder (the estimated burn in period is indicated in the legend);  the \texttt{Support Points} method; \texttt{Stein Thinning}, for each of the settings \texttt{med}, \texttt{sclmed}, \texttt{smpcov}.
		}
		\label{fig:Goodwin_SP_RW}
	\end{figure}

	Having identified a burn-in period, the standard approach thins the remainder of the sample path according to \eqref{eq: std post process}. 
	In the experiments that follow we focus on the VK diagnostic and consider both the smallest and largest estimates obtained for the burn-in period.
	The resulting index sets $\pi$ are displayed, for $m = 20$ and \texttt{RW} (the simplest MCMC method) in \Cref{fig:Goodwin_SP_RW} (top left panel), and in \Cref{app:addtional_Goodwin}, Figures \ref{fig:Goodwin_SP_ADA-RW} (\texttt{ADA-RW}), \ref{fig:Goodwin_SP_MALA} (\texttt{MALA}), \ref{fig:Goodwin_SP_MALA_PRECOND} (\texttt{P-MALA}).
	In the same figures (top right panel) we show the set of \texttt{Support Points} obtained using algorithm proposed by \cite{Mak2018}.
	The remaining panels display the output from \verb+Stein Thinning+.
	Compared to the standard approach, \texttt{Support Points} and \verb+Stein Thinning+ produce sets that are more structured.

	\begin{figure}
		\centering
		\includegraphics[width = 1
		\textwidth]{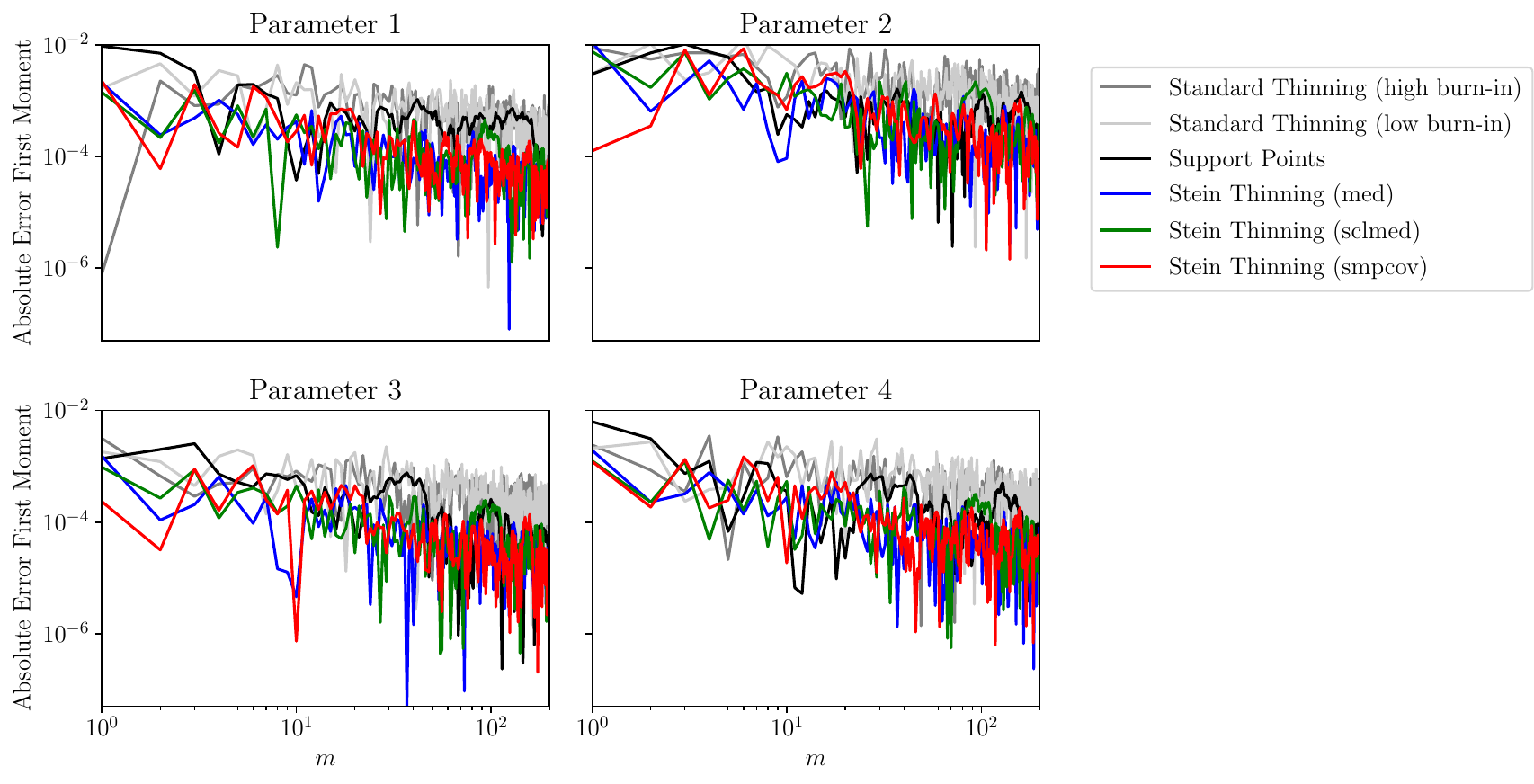}
		\caption{Absolute error of estimates for the posterior mean of each parameter in the Goodwin oscillator, based on output from \texttt{RW} MCMC.
		}
		\label{fig:Goodwin_moments_RW}
	\end{figure}
	
	\begin{figure}
		\centering
		\includegraphics[width = 1\textwidth]{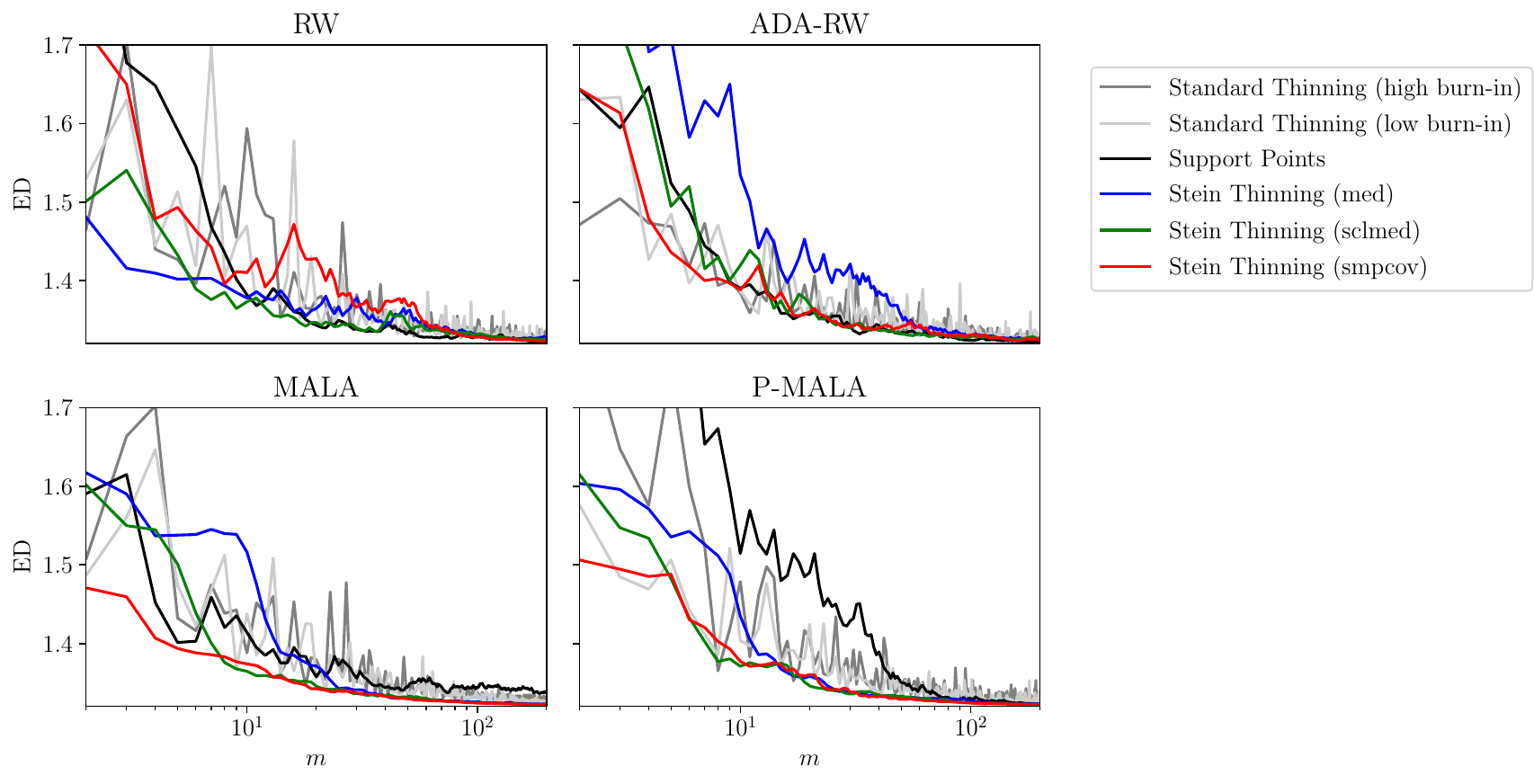}
	
		\caption{Goodwin oscillator: Energy distance (ED) to the posterior, as per \eqref{eq: ED}, for empirical distributions obtained through traditional burn-in and thinning (grey lines), \texttt{Support Points} (black line) and \texttt{Stein Thinning} (colored lines), based on output from four different MCMC methods. 
		}
		\label{fig:Goodwin_MED}
	\end{figure}
	
	\begin{figure}
		\centering
		\includegraphics[width = 1\textwidth]{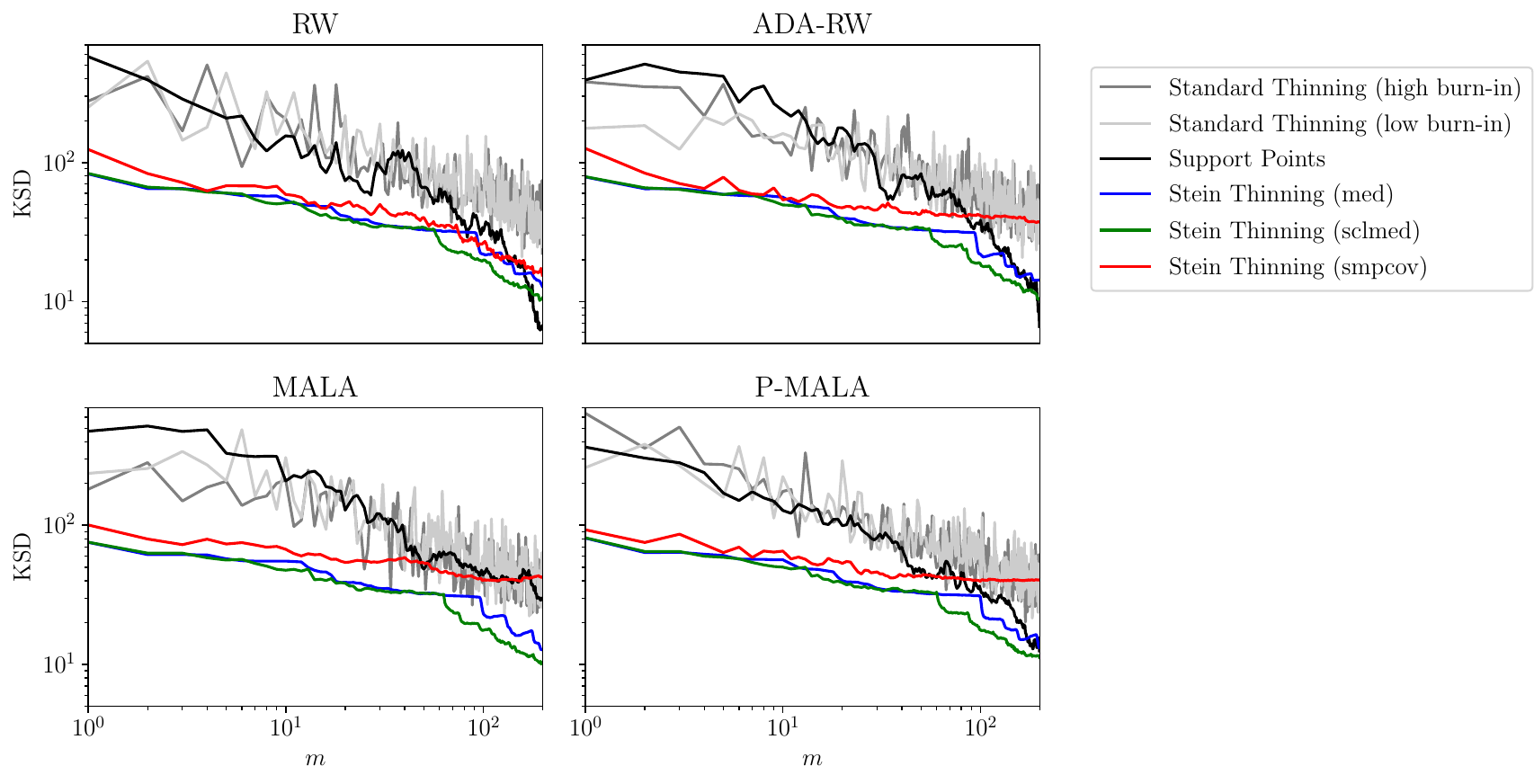}

		\caption{Goodwin oscillator: Kernel Stein discrepancy (KSD) based on \texttt{med}, for empirical distributions obtained through traditional burn-in and thinning (grey lines), \texttt{Support Points} (black line) and \texttt{Stein Thinning} (colored lines), based on output from four different MCMC methods. 
		}
		\label{fig:Goodwin_KSD}
	\end{figure}

	To assess the performance of these competing methods, we first considered the toy problem of approximating the posterior mean of each parameter in the Goodwin oscillator as an average of $m$ points selected from the MCMC output.
	\Cref{fig:Goodwin_moments_RW} displays absolute errors for each method, based on \texttt{RW}; our ground truth was provided by an extended run of MCMC.
	Results for the other MCMC methods are provided in \Cref{app:addtional_Goodwin}, Figures \ref{fig:Goodwin_moments_ADA_RW} (\texttt{ADA-RW}), \ref{fig:Goodwin_moments_MALA} (\texttt{MALA}), \ref{fig:Goodwin_moments_PRECOND-MALA} (\texttt{P-MALA}).
	Broadly speaking, \texttt{Stein Thinning} tends to provide more accurate estimators compared to the alternatives considered.
	From \Cref{fig:Goodwin_moments_RW} it is difficult to see any difference in performance between \texttt{med}, \texttt{sclmed} and \texttt{smpcov}.
	To gain more insight, in \Cref{app:addtional_Goodwin} we plot marginal density estimates in Figures \ref{fig:Goodwin_densities_RW} (\texttt{RW}), \ref{fig:Goodwin_densities_ADA-RW} (\texttt{ADA-RW}), \ref{fig:Goodwin_densities_MALA} (\texttt{MALA}), \ref{fig:Goodwin_densities_PRECOND-MALA} (\texttt{P-MALA}).
	It is apparent that \texttt{Stein Thinning} improves on the standard approach, whilst \texttt{med} and \texttt{sclmed} performed slightly better than \texttt{smpcov}.
	This may be because in \texttt{smpcov} there are more degrees of freedom in $\Gamma$ that must be estimated.
	\texttt{Support Points} performed on a par with \texttt{Stein Thinning} based on \texttt{smpcov}.

	To facilitate a more principled assessment, we computed two quantitative measures for how well the resulting empirical distributions approximate the posterior. 
	These were (a) the energy distance \citep[ED;][]{szekely2004testing,baringhaus2004new}, given up to an additive constant by
	\begin{align}
	\text{ED} := \frac{2}{m}  \sum_{j=1}^m \int  \| x - x_{\pi(j)} \|_{\Sigma} \; \mathrm{d}P(x) - \frac{1}{m^2} \sum_{j,j' = 1}^m \|x_{\pi(j)} - x_{\pi(j')} \|_{\Sigma}, \label{eq: ED}
	\end{align}
	where in this paper we used the norm $\|x\|_\Sigma := \|\Sigma^{-1/2} x\|$ induced by the covariance matrix of $P$, with both $\Sigma$ and \eqref{eq: ED} being estimated from MCMC output, and (b) the KSD based on \texttt{med}, the simplest setting for $\Gamma$. 
	ED serves as an objective performance measure, being closely related to the quantity that \texttt{Support Points} attempts to minimise (\cite{Mak2018} used the $\|\cdot\|$ norm in place of $\|\cdot\|_\Sigma$), while KSD is the performance measure that is being directly optimised in \verb+Stein Thinning+.
	Our decision to include KSD in the assessment is motivated by three factors; (i) ED is somewhat insensitive to detail, making it difficult to rank competing methods; (ii) the empirical approximation of ED in \eqref{eq: ED} relies on access to high-quality MCMC output, but this will not be available in \Cref{subsec: cardiac}; (iii) Stein discrepancies are the only computable performance measures in the Bayesian context, to the best of our knowledge, that have been proven to provide convergence control.

	The results for ED are shown in \Cref{fig:Goodwin_MED}.
	Here \verb+Stein Thinning+ based on \texttt{sclmed} performed at least as well as the other methods considered and, surprisingly, out-performed \texttt{Support Points} when applied to \texttt{MALA} and \texttt{P-MALA} output.
	This may be because \texttt{MALA} and \texttt{P-MALA} provided worse approximations to $P$ compared with \texttt{RW} and \texttt{ADA-RW} (recall that \texttt{Support Points} relies on the MCMC output providing an accurate approximation of $P$).
	Note that neither ED nor KSD values will tend to 0 as $m \rightarrow \infty$ in this experiment, since the number $n$ of MCMC iterations was fixed.
	The corresponding results for KSD are presented in \Cref{fig:Goodwin_KSD} and show a clearer performance ordering of the competing methods, with \verb+Stein Thinning+ based on \texttt{med} and \texttt{sclmed} out-performing all other methods for all but the largest values of $m$ considered.
	The \texttt{smpcov} setting performed well for small $m$ but for large $m$ its performance degraded.
	The performance ordering under KSD was identical across the different MCMC output.

	\subsection{Lotka--Volterra} \label{subsec: Lotka}
	
	\begin{figure}
		\centering
		\includegraphics[width = 1\textwidth]{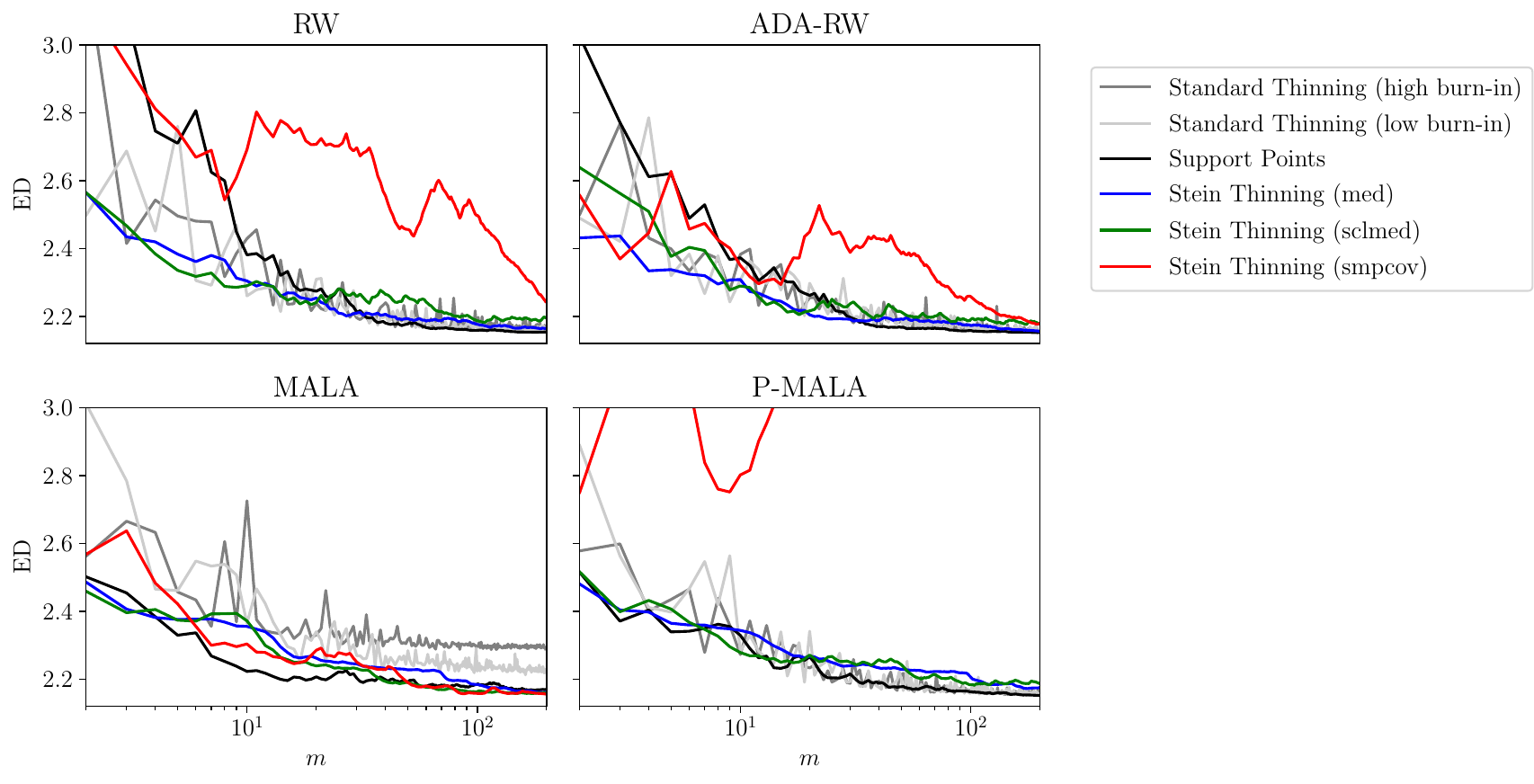}
		\caption{Lotka--Volterra model: Energy distance (ED) to the posterior, as per \eqref{eq: ED}, for empirical distributions obtained through traditional burn-in and thinning (grey lines), \texttt{Support Points} (black line) and \texttt{Stein Thinning} (colored lines), based on output from four different MCMC methods. 
		}
		\label{fig:Lotka_ED}
	\end{figure}
	
	The second example that we consider is the predator-prey model of \cite{Lotka1926} and \cite{Volterra1926}.
	A description of the $d=4$ dimensional inference task, the output from MCMC methods and the implementation of thinning procedures is reserved for \Cref{app:addtional_Lotka}. 
	Compared to the Goodwin oscillator, the Lotka--Volterra posterior $P$ exhibits stronger correlation among parameters.
	This has consequences for our assessment, since now all MCMC methods, and in particular \texttt{MALA}, mixed less well compared to corresponding results for the Goodwin oscillator, as can be seen from the trace plots in \Cref{app:addtional_Lotka}, \Cref{fig:Lotka_traceplots_z}.
	Results are reported for ED in \Cref{fig:Lotka_ED}.
	It can be seen that \texttt{Stein Thinning} based on \texttt{med} and \texttt{sclmed} performed comparably with \texttt{Support Points}, being better for small $m$ in the case of \texttt{RW} and \texttt{ADA-RW} and marginally worse for large $m$ in \texttt{RW}, \texttt{ADA-RW} and \texttt{P-MALA}.
	Interestingly, the setting \texttt{smpcov} was associated with poor performance on output from \texttt{RW}, \texttt{ADA-RW} and especially \texttt{P-MALA}.
	This may be because, when $\Gamma$ is poorly conditioned, any error in an estimate for $\Gamma$ will be amplified when computing $\Gamma^{-1}$.
	However, in the case of \texttt{MALA}, which mixed poorly, the standard approach of burn-in removal and thinning performed poorly and all settings of \texttt{Stein Thinning} provided an improvement.

	Results for KSD are reported in \Cref{fig:Lotka_KSD}. 
    The performance ordering of competing methods under KSD is similar to that reported in \Cref{subsec: Goodwin}, except for the \texttt{smpcov} setting which appears to improve the performance of \texttt{Stein Thinning} for larger values of $m$ in the context of \texttt{MALA}.
    This may be because \texttt{smpcov} serves to ``whiten'' the correlation structure in $P$, such that the resulting geometry is more favourable for the construction of an empirical approximation.
    However, this improved performance was not seen on \texttt{P-MALA}.
	In all cases \texttt{Stein Thinning} out-performed \texttt{Support Points}.

	\begin{figure}[t!]
		\centering
		\includegraphics[width = 1\textwidth]{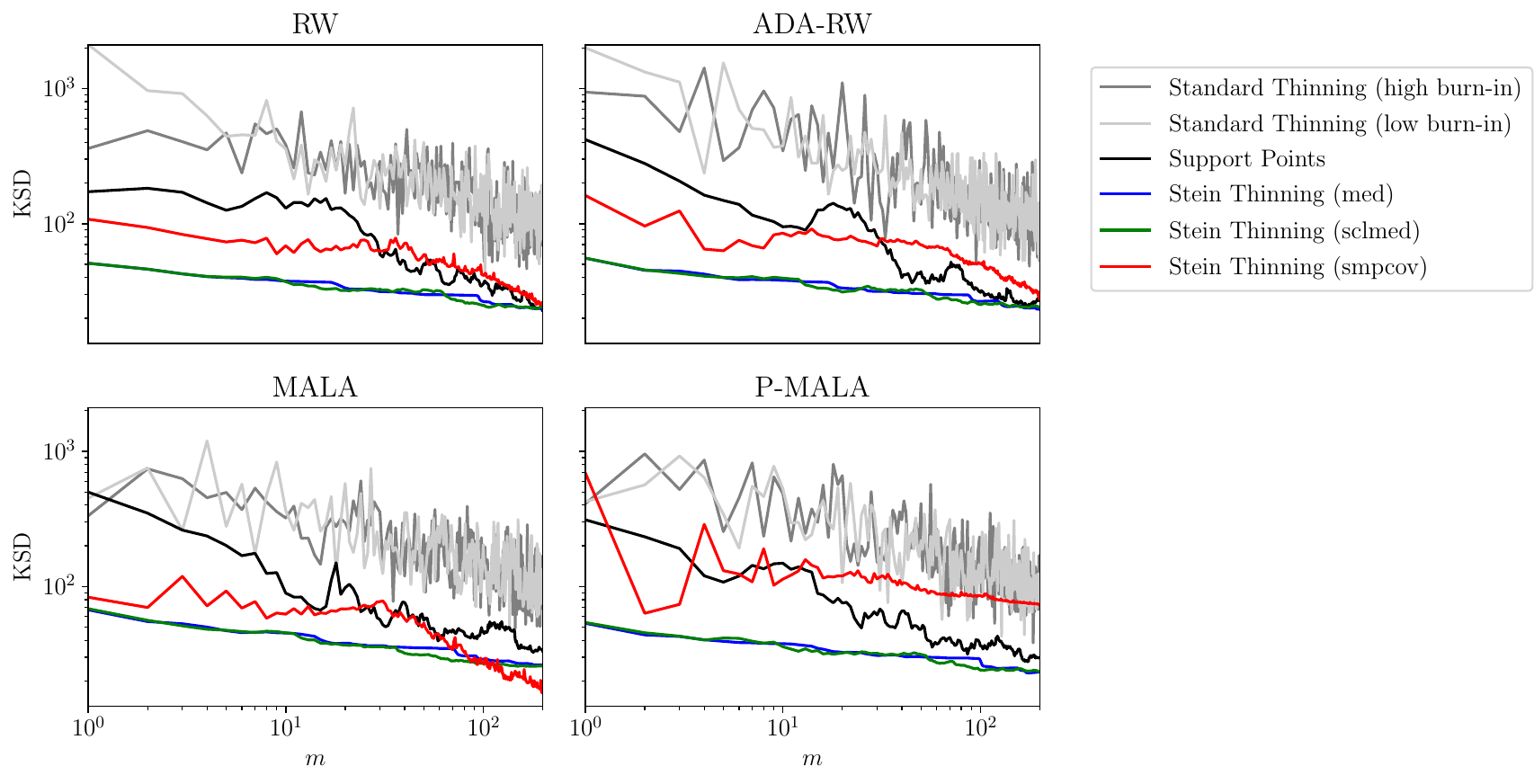}
		
		\caption{Lotka-Volterra model: Kernel Stein discrepancy (KSD) based on \texttt{med}, for empirical distributions obtained through traditional burn-in and thinning (grey lines), \texttt{Support Points} (black line) and \texttt{Stein Thinning} (colored lines), based on output from four different MCMC methods. 
		}
		\label{fig:Lotka_KSD}
	\end{figure}


		\subsection{Calcium Signalling Model} \label{subsec: cardiac}

		Our final example a model for calcium signalling in cardiac cells, illustrated in \Cref{Appendix:Hinch}, \Cref{fig:cell_scheme}.
		The model describes an electrically-activated intracellular calcium signal that in turn activates the sub-cellular sarcomere, causing the muscle cell to contract and the heart to beat. 
		The intracellular calcium signal is crucial for healthy cardiac function.
		However, under pathological conditions, dysregulation of this intra-cellular signal can play a central role in the initiation and sustenance of life-threatening arrhythmias. Computational models are increasingly being applied to study this highly-orchestrated multi-scale signalling cascade to determine how changes in cell-scale calcium regulation, encoded in calcium model parameters, impact whole-organ cardiac function \citep{campos2015stochastic,niederer2019computational,colman2019arrhythmia}. 
		The computational cost of simulating from tissue-scale and organ-scale models is high, with single simulations taking thousands of CPU hours \citep{niederer2011simulating,augustin2016anatomically,strocchi2020simulating}. 
		This limits the capacity to propagate uncertainty in calcium signalling model parameters up to organ-scale simulations, so that at present it remains unclear how uncertainty in calcium signalling parameters impacts the predictions made by a whole-organ model.
		Our motivation for developing \verb+Stein Thinning+ was to obtain a compressed representation of the posterior distribution for the $d=38$ dimensional parameter of a calcium signalling model, based on a cell-scale experimental dataset, which can subsequently be used as an experimental design to propagate uncertainty through a whole-organ model.
		
		This motivating problem entails a second complication in that, compared to the example in \Cref{subsec: Goodwin} and even the example in \Cref{subsec: Lotka}, the development of an efficient MCMC method appears to be difficult.
		Thus, in the experiment that follows, we cannot rely on any of the MCMC methods that we described at the start of \Cref{sec: empirical} to provide anything more than a crude approximation of the posterior, at best.
		This is evidenced by the non-overlapping approximations to the posterior marginals produced when different random seeds are used; see \Cref{fig:Hinch_kdeA,fig:Hinch_kdeB,fig:Hinch_kdeC,fig:Hinch_kdeD}.
		(Of course, it is possible that a more sophisticated sampling method may be designed for this task, but our aim here is not to develop a new sampling method.)
		Tempering of the likelihood provides a straightforward route to improve the mixing of MCMC, but the invariant distribution $Q$ will then no longer equal $P$. 
		Here we explore the potential for \texttt{Stein Thinning} to perform bias-correction for such $Q$-invariant MCMC output, in the spirit of \Cref{cor: bias}. 
		
			\begin{figure}[t!]
			\centering
			\includegraphics[width = 0.7
			\textwidth,clip,trim = 0cm 0cm 0cm 0.6cm]{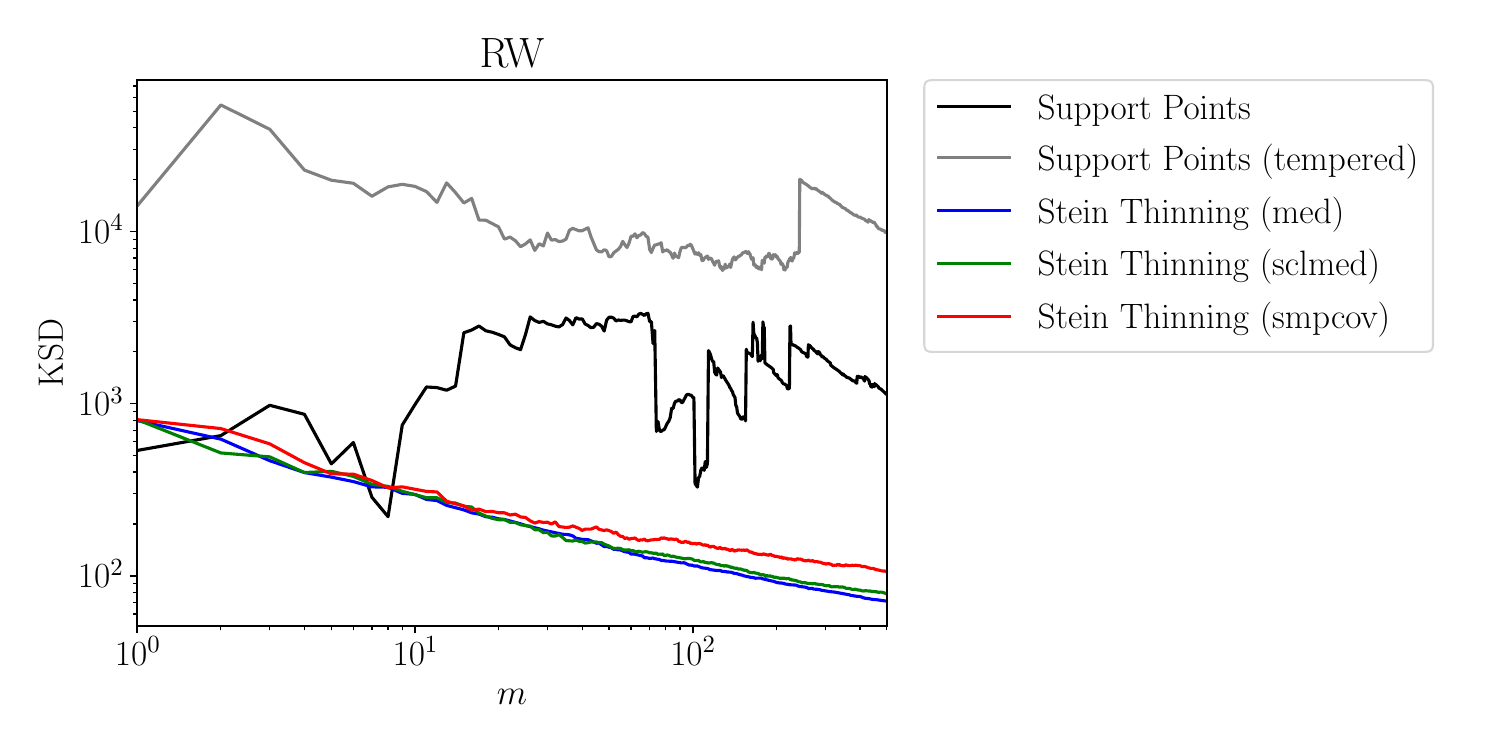}
			\caption{Calcium signalling model. Kernel Stein discrepancy (KSD) based on \texttt{med}, for empirical distributions obtained using \texttt{Support Points} and \texttt{Stein Thinning}, based on output from \texttt{RW} MCMC applied to either $P$ or a tempered version of $P$.
			}
			\label{fig:Hinch_KSD}
		\end{figure}
		
		Our focus in the remainder is on output from the \texttt{RW} MCMC method.
		This MCMC method was selected since (a) gradient-free methods can be easier to tune when the posterior is concentrated \citep{livingstone2019robustness}, and (b) once the sample path has been computed, the associated gradients can be computed in parallel.
		Both standard and tempered MCMC were performed; in the latter case the likelihood was tempered so that the (biased) target $Q$ was just about tractable for MCMC (see \Cref{Appendix:Hinch}).		
		In each case a total of $n = 4 \times 10^6$ iterations of MCMC were preformed, representing two weeks' CPU time. 
		
		\Cref{fig:Hinch_KSD} reports the KSD based on \texttt{med}, for index sets of cardinality up to $m=500$; see \Cref{Appendix:Hinch} for results for KSD based on \texttt{sclmed} (\Cref{fig:Hinch_KSD_sclmed}) and \texttt{smpcov} (\Cref{fig:Hinch_KSD_smpcov}). 
		Considering first the tempered MCMC output, the lower values of KSD achieved by \texttt{Stein Thinning} are consistent with fact that \texttt{Stein Thinning} corrects for bias due to tempering, while \texttt{Support Points} does not.
		Furthermore, \texttt{Stein Thinning} of tempered MCMC results in lower values of KSD compared to \texttt{Support Points} applied to standard MCMC output, with the latter being negatively affected by the non-convergence of the MCMC.
		Inspection of the univariate marginals demonstrates that the combination of tempering and \texttt{Stein Thinning} produces approximations that are robust to changes in the random seed, while the approximations produced by standard MCMC with an equivalent computational budget are not; see \Cref{fig:Hinch_kdeA,fig:Hinch_kdeB,fig:Hinch_kdeC,fig:Hinch_kdeD}.

	\section{Conclusion} \label{sec: conclusion}
	
	In this paper, standard approaches used to post-process and compress output from MCMC were identified as being sub-optimal when one considers the approximation quality of the empirical distribution that is produced.
	A novel method, \verb+Stein Thinning+, was proposed that seeks a subset of the MCMC output, of fixed cardinality, such that the associated empirical approximation is close to optimal.
	The theoretical analysis that we have provided for \verb+Stein Thinning+ handles the effect of the post-processing procedure jointly with the randomness involved in simulating from the Markov chain, such that consistency of the overall estimator is established.
	
	Although we focused on MCMC, the proposed method can be applied to any computational method that provides a collection of states as output.
	These include approximate (biased) MCMC methods,
	where \verb+Stein Thinning+ may be able to provide bias correction in the spirit of \Cref{cor: bias}.
	On the other hand, the main limitation of \texttt{Stein Thinning} is that it requires gradients of the log-target to be computed, which is not always practical.
	
	Our research was motivated by challenging parameter inference problems that arise in ODEs, in particular in cardiac modelling where one is interested in propagating calcium signalling parameter uncertainty through a whole-organ simulation -- a task that would na\"{i}vely be impractical or impossible using the full MCMC output.
	Our ongoing research is exploiting \verb+Stein Thinning+ in this context and is enabling us to perform scientific investigations that were not feasible beforehand.
	Furthermore, in a sequel we demonstrate that approximate implementations of \verb+Stein Thinning+ can massively reduced its implementation cost \citep{teymur2020optimal}.

	\paragraph*{Acknowledgements}
	
	The authors are grateful for support from the Lloyd's Register Foundation programme on data-centric engineering and the programme on health and medical sciences at the Alan Turing Institute.
	MR, SN and CJO were supported by the British Heart Foundation (BHF; SP/18/6/33805).
	JC was supported by the UKRI Strategic Priorities Fund (EP/T001569/1).
	PS was supported by the BHF (RG/15/9/31534).
	SN was supported by the EPSRC (EP/P01268X/1, NS/A000049/1, EP/M012492/1), the BHF (PG/15/91/31812, FS/18/27/33543), the NIHR (II-LB-1116-20001) and the Wellcome Trust (WT 203148/Z/16/Z).
	The authors thank Matthew Graham, Liam Hodgkinson, Rob Salomone, and the anonymous Editor, Associate Editor, and Reviewers, for helpful comments on the manuscript.

	\FloatBarrier

	\newpage
	\setcounter{page}{1}
	\appendix
    \beginsupplement
    \section*{Supplementary Material}
    
    This electronic supplement contains details for our \verb+Python+, \verb+R+ and \verb+MATLAB+ implementations of \verb+Stein Thinning+, proofs for all novel theoretical results reported in \Cref{sec: theory} of the main text, as well as additional material relating to the experimental assessment in \Cref{sec: empirical} of the main text.
    It is structured as follows:
    
    \startcontents[sections]
    \printcontents[sections]{l}{1}{\setcounter{tocdepth}{2}}

    \section{Software} \label{subsec: software}
	
	To assist with applications of \verb+Stein Thinning+ we have provided code in \verb+Python+, \verb+R+ and \verb+MATLAB+.
	The code is available at:
	\begin{center}
	    \url{http://stein-thinning.org/}
	\end{center}
	\noindent In this section we demonstrate \verb+Stein Thinning+ in \verb+Python+, but the syntax for \verb+Stein Thinning+ in \verb+R+ and in \verb+MATLAB+ is almost identical.
	As an illustration of how \verb+Stein Thinning+ can be used to post-process output from \verb+Stan+, consider the following simple \verb+Stan+ script that produces 1000 correlated samples from a bivariate Gaussian model:
	
	\begin{lstlisting}
	from pystan import StanModel
	mc = """
	parameters {vector[2] x;}
	model {x ~ multi_normal([0, 0], [[1, 0.8], [0.8, 1]]);}
	"""
	sm = StanModel(model_code=mc)
	fit = sm.sampling(iter=1000)
	\end{lstlisting}
	
	\noindent The bivariate Gaussian model is used for illustration, but regardless of the complexity of the model being sampled the output of \verb+Stan+ will be a \verb+fit+ object.
	The sampled points $x_i$ and the gradients $\nabla \log p(x_i)$ can be extracted from the returned \verb+fit+ object:
	
	\begin{lstlisting}
	import numpy as np
	smpl = fit['x']
	grad = np.apply_along_axis(fit.grad_log_prob, 1, smpl)
	\end{lstlisting}
	
	\noindent One can then perform \verb+Stein Thinning+ to obtain a subset of $m=40$ states by running the following code:
	
	\begin{lstlisting}
	from stein_thinning.thinning import thin
	idx = thin(smpl, grad, 40)
	\end{lstlisting}
	
	\noindent The \verb+thin+ function returns a \verb+NumPy+ array containing the row indices in \verb+smpl+ (and \verb+grad+) of the selected points. 
	The default usage requires no additional user input and is based on the \texttt{sclmed} setting from \Cref{subsec: kernel choice}, informed by the empirical analysis of \Cref{sec: empirical}.
	Alternatively, the user can choose to specify which setting to use for computing the preconditioning matrix $\Gamma$ by setting the option string \verb+pre+ to either \verb+'med'+,  \verb+'sclmed'+, or \verb+'smpcov'+. 
	For example, the default setting corresponds to
	\begin{lstlisting}
	idx = thin(smpl, grad, 40, pre='sclmed')
	\end{lstlisting}
	The ease with which \verb+Stein Thinning+ can be used makes it possible to consider a wide variety of applications, including the ODE models that we considered in \Cref{sec: empirical} of the main text.
		
		\section{Proofs} \label{ap: theory}
		
		This appendix contains detailed proofs for all novel theoretical results in the main text.
		
		\subsection{Proof of \Cref{prop: optimal converge}} \label{ap: thm 1 proof}
		
		First we state and prove two elementary results that will be useful:
		
		\begin{lemma} \label{lem: basic inequality}
			For all $a,b \geq 0$ it holds that $2 a \sqrt{a^2 + b} \leq 2a^2 + b$.
		\end{lemma}
		\begin{proof}
			Since all quantities are non-negative, we may square both sides to get an equivalent inequality $4 a^2 (a^2 + b) \leq (2a^2+b)^2$.
			Expanding the brackets and cancelling terms leads to $0 \leq b^2$, which is guaranteed to hold.
		\end{proof}
		
		\begin{lemma} \label{lem: harmonic series}
			For all $m \in \mathbb{N}$ it holds that $\sum_{j=1}^m \frac{1}{j} \leq 1 + \log(m)$.
		\end{lemma}
		\begin{proof}
			Since $x \mapsto \frac{1}{x}$ is convex on $x \in (0,\infty)$, we have that the Riemann sum $\sum_{j=2}^m \frac{1}{j}$ is a lower bound for the Riemann integral $\int_1^m \frac{1}{x} \mathrm{d}x = \log(m)$.
			Thus $\sum_{j=1}^m \frac{1}{j} = 1 + \sum_{j=2}^m \frac{1}{j} \leq 1 + \log(m)$, as required.
		\end{proof}
		
		Now we present the proof of \Cref{prop: optimal converge}:
		
		\begin{proof}[Proof of \Cref{prop: optimal converge}]
			
			Let $a_m := m^2 D_{k_P}(\frac{1}{m} \sum_{j=1}^m \delta(x_{\pi(j)}) )^2$, $f_m := \sum_{j=1}^m k_P(x_{\pi(j)}, \cdot)$ and also let $S^2 := \max_{i=1,\dots,n} k_P(x_i,x_i)$, so that
			\begin{align*}
			a_m = \sum_{j=1}^m \sum_{j'=1}^m k_P(x_{\pi(j)},x_{\pi(j')}) &= a_{m-1} + k_P(x_{\pi(m)},x_{\pi(m)}) + 2 \sum_{j=1}^{m-1} k_P(x_{\pi(j)},x_{\pi(m)}) \\
			&\leq a_{m-1} + S^2 + 2 \min_{y \in \{x_i\}_{i=1}^n} f_{m-1}(y) .
			\end{align*}
			Recall that $\mathcal{H}(k_P)$ denotes the reproducing kernel Hilbert space of the kernel $k_P$ and pick an element $h^* \in \mathcal{H}(k_P)$ of the form $h^* := \sum_{i=1}^n w_i^* k_P(x_i,\cdot)$, where the weight vector $w^*$ satisfies \eqref{eq: optimal weights}.
			From this definition it follows that $\|h^*\|_{\mathcal{H}(k_P)} = D_{k_P}(\sum_{i=1}^n w_i^* \delta(x_i))$, which is the minimal KSD attainable under the constraint \eqref{eq: optimal weights}.
			Now, let $\mathcal{M}$ denote the convex hull of $\{k_P(x_i,\cdot)\}_{i=1}^n$, so that $h^* \in \mathcal{M} \subset \mathcal{H}(k_P)$ and therefore
			\begin{align}
			\min_{y \in \{x_i\}_{i=1}^n} f_{m-1}(y) = \inf_{h \in \mathcal{M}} \langle f_{m-1} , h \rangle_{\mathcal{H}(k_P)} \leq \langle f_{m-1} , h^* \rangle_{\mathcal{H}(k_P)} . \label{eq: convex hull}
			\end{align}
			Noting that $a_m = \|f_m\|_{\mathcal{H}(k_P)}^2$, we have from \eqref{eq: convex hull} and Cauchy-Schwarz that
			\begin{align*}
			\min_{y \in \{x_i\}_{i=1}^n} f_{m-1}(y) \leq \sqrt{a_{m-1}} \|h^*\|_{\mathcal{H}(k_P)}
			\end{align*}
			and therefore
			\begin{align}
			a_m \leq a_{m-1} + S^2 + 2 \sqrt{a_{m-1}} \|h^*\|_{\mathcal{H}(k_P)} . \label{eq: initial bound}
			\end{align}
			Letting 
			\begin{align}
			C_m := \frac{1}{m}\left( S^2- \|h^*\|_{\mathcal{H}(k_P)}^2 \right) \sum_{j=1}^m \frac{1}{j} ,  \label{eq: C def}
			\end{align}
			we will establish by induction that
			\begin{align}
			a_m \leq m^2 \left( \|h^*\|_{\mathcal{H}(k_P)}^2 + C_m \right) . \label{eq: induction}
			\end{align}
			This will in turn prove the result, since
			\begin{align*}
			D_{k_P}\left(\frac{1}{m} \sum_{j=1}^m \delta(x_{\pi(j)}) \right)^2 = \frac{a_m}{m^2} & \leq \|h^*\|_{\mathcal{H}(k_P)}^2 + C_m \\
			& = D_{k_P}\left(\sum_{i=1}^n w_i^* \delta(x_i) \right)^2 + C_m \\
			& \leq D_{k_P}\left(\sum_{i=1}^n w_i^* \delta(x_i) \right)^2 + \left( \frac{1 + \log(m)}{m} \right) S^2 ,
			\end{align*}
			where the upper bound on $C_m$ follows from the fact that $\|h^*\|_{\mathcal{H}(k_0)} \leq S$, combined with \Cref{lem: harmonic series}.
			
			The remainder of the proof is dedicated to establishing the induction in \eqref{eq: induction}.
			The base case $m=1$ is satisfied since $a_1 = D_{k_P}(\delta(x_{\pi(1)})) = k_P(x_{\pi(1)},x_{\pi(1)}) \leq S^2$ and $C_1 = S^2 - \|h^*\|_{\mathcal{H}(k_P)}^2$, so that $a_1 \leq \|h^*\|_{\mathcal{H}(k_P)}^2 + C_1$.
			For the inductive step, we assume that \eqref{eq: induction} holds when $m$ is replaced by $m-1$ and aim to derive \eqref{eq: induction}.
			From \eqref{eq: initial bound} and the inductive assumption, we have that
			\begin{align}
			a_m &\leq a_{m-1} + S^2 + 2 \sqrt{a_{m-1}} \|h^*\|_{\mathcal{H}(k_P)} \nonumber \\
			&\leq (m-1)^2 \left( \|h^*\|_{\mathcal{H}(k_P)}^2 + C_{m-1} \right) + S^2 + 2 (m-1) \sqrt{ \|h^*\|_{\mathcal{H}(k_P)}^2 + C_{m-1} } \|h^*\|_{\mathcal{H}(k_P)} \nonumber \\
			&= m^2 \left( \|h^*\|_{\mathcal{H}(k_P)}^2 + C_m \right) + R_m \label{eq: almost}
			\end{align}
			where
			\begin{align*}
			R_m & := (m-1)^2 C_{m-1} - m^2 C_m + (1-2m) \|h^*\|_{\mathcal{H}(k_P)}^2 + S^2 \\
			& \hspace{200pt} + 2 (m-1) \sqrt{ \|h^*\|_{\mathcal{H}(k_P)}^2 + C_{m-1}} \|h^*\|_{\mathcal{H}(k_P)} 
			\end{align*}
			The induction \eqref{eq: induction} will therefore follow from \eqref{eq: almost} if $R_m \leq 0$.
			Now, $R_m \leq 0$ if and only if
			\begin{align*}
			2 \sqrt{ \|h^*\|_{\mathcal{H}(k_P)}^2 + C_{m-1}} \|h^*\|_{\mathcal{H}(k_P)} \leq \frac{ m^2 C_m - (m-1)^2 C_{m-1} }{m-1} - \frac{ S^2 - \|h^*\|_{\mathcal{H}(k_P)}^2 }{m-1} + 2\|h^*\|_{\mathcal{H}(k_P)}^2 .
			\end{align*}
			From \Cref{lem: basic inequality} it must hold that
			\begin{align*}
			2 \sqrt{ \|h^*\|_{\mathcal{H}(k_P)}^2 + C_{m-1}} \|h^*\|_{\mathcal{H}(k_P)} \leq 2 \|h^*\|_{\mathcal{H}(k_P)}^2 + C_{m-1} ,
			\end{align*}
			meaning it is sufficient to show that
			\begin{align}
			2 \|h^*\|_{\mathcal{H}(k_P)}^2 + C_{m-1} \leq \frac{ m^2 C_m - (m-1)^2 C_{m-1} }{m-1} - \frac{ S^2 - \|h^*\|_{\mathcal{H}(k_P)}^2 }{m-1} + 2\|h^*\|_{\mathcal{H}(k_P)}^2 . \label{eq: ineq to prove}
			\end{align}
			Algebraic simplification of \eqref{eq: ineq to prove} reveals that \eqref{eq: ineq to prove} is equivalent to
			\begin{align}
			m C_m - (m-1) C_{m-1} \leq \frac{1}{m} \left( S^2 - \|h^*\|_{\mathcal{H}(k_P)}^2 \right)  \label{eq: final check}
			\end{align}
			and, using \eqref{eq: C def}, we verify that \eqref{eq: final check} is satisfied as an equality. 
			This completes the inductive argument.
		\end{proof}
		
	\begin{remark}
	    Our results can also be applied to maximum mean discrepancies, where $D_{\mathcal{F}}$ in \eqref{eq:IPMs} is given by $\mathcal{F} = \{f \in \mathcal{H}(k) : \|f\|_{\mathcal{H}(k)} \leq 1 \}$ and $k : \mathcal{X} \times \mathcal{X} \rightarrow \mathbb{R}$ denotes a reproducing kernel \citep{song2008learning}.
	    Indeed, we can set
	    $$
	    k_P(x,y) = k(x,y) - \int_{\mathcal{X}} k(x,y) \mathrm{d}P(x) - \int_{\mathcal{X}} k(x,y) \mathrm{d}P(y) + \int_{\mathcal{X}} \int_{\mathcal{X}} k(x,y) \mathrm{d}P(x) \mathrm{d}P(y)
	    $$
	    in order for $\int_{\mathcal{X}} k_P(x,\cdot) \mathrm{d}P = 0$ to be satisfied for all $x \in \mathcal{X}$, and observe that this construction ensures $D_{\mathcal{F}}$ and $D_{k_P}$ are identical.
	    \cite{teymur2020optimal} explores these consequences of our results in detail.
	\end{remark}

		\subsection{Proof of \Cref{thm: main theorem}} \label{ap: proof of main theorem}
		
		First we state and prove a technical lemma that will be useful:
		
		\begin{lemma} 
			\label{Lemma-kernel-bounds}
			Let $\mathcal{X}$ be a measurable space and let $Q$ be a probability distribution on $\mathcal{X}$.
			Let $k_Q : \mathcal{X} \times \mathcal{X} \rightarrow \mathbb{R}$ be a reproducing kernel with $\int_{\mathcal{X}} k_Q(x,\cdot) \mathrm{d}Q = 0$ for all $x \in \mathcal{X}$.
			Consider a $Q$-invariant, time-homogeneous Markov chain $(X_i)_{i \in \mathbb{N}} \subset \mathcal{X}$ generated using a $V$-uniformly ergodic transition kernel, such that $V(x) \geq \sqrt{k_Q(x,x)}$ for all $x \in \mathcal{X}$, with parameters $R \in [0,\infty)$ and $\rho \in (0,1)$ as in \eqref{eq: V unif ergod}. 
			Then with $C = \frac{2 R \rho}{1 - \rho}$ we have that
			$$
			\sum_{i=1}^{n} \sum_{r \in \{1,\dots,n\} \setminus \{i\}}\mathbb{E} \left [k_Q(X_i,X_r) \right] 
			\; \leq \;
			C \sum_{i = 1}^{n - 1} \mathbb{E} \left[ \sqrt{k_Q(X_i,X_i)} V(X_i) \right].
			$$
		\end{lemma} 
		\begin{proof}
			First recall that given random variables $X,Y$ taking values in $\mathcal{X}$, the conditional mean embedding of the distribution $\mathbb{P}[X|Y=y]$ is the function $\mathbb{E}[ k_Q(X,\cdot) | Y=y] \in \mathcal{H}(k_Q)$ \citep{Song2009}. 
			By the reproducing property we have
			$\mathbb{E}[k_Q(X,y)|Y=y] = \langle k_Q(y,\cdot) , \mathbb{E}[ k_Q(X,\cdot) | Y=y] \rangle_{\mathcal{H}(k_Q)}$, hence
			$\mathbb{E}[k_Q(X,Y)|Y] = \langle k_Q(Y,\cdot) , \mathbb{E}[ k_Q(X,\cdot) | Y] \rangle_{\mathcal{H}(k_Q)}$. 
			Thus
			\begin{align*}
			\mathbb{E}[k_Q(X,Y)|Y] &= \langle k_Q(Y,\cdot) , \mathbb{E}[ k_Q(X,\cdot) | Y] \rangle_{\mathcal{H}(k_Q)} \\
			&= \|k_Q(Y,\cdot)\|_{\mathcal{H}(k_Q)} \left\langle \frac{k_Q(Y,\cdot)}{\|k_Q(Y,\cdot)\|_{\mathcal{H}(k_Q)}} , \mathbb{E}[ k_Q(X,\cdot) | Y] \right\rangle_{\mathcal{H}(k_Q)}  \\
			& \leq  \|k_Q(Y,\cdot)\|_{\mathcal{H}(k_Q)} \sup_{\| h \|_{\mathcal{H}(k_Q)} =1} \langle h , \mathbb{E}[ k_Q(X,\cdot) | Y] \rangle_{\mathcal{H}(k_Q)} .
			\end{align*}
			In what follows it is convenient to introduce a new random variable $Z$, independent from the Markov chain, such that $Z \sim Q$.
			Then, since $\mathbb{E}[k_Q(Z,\cdot)]=0$, we have $\mathbb{E}[h(Z)]=0$ for any $h\in \mathcal{H}(k_Q)$. 
			Hence we have that
			$$
			\mathbb{E}[k_Q(X,Y)|Y] 
			\; \leq \; 
			\sqrt{k_Q(Y,Y)}\sup_{\| h \|_{\mathcal{H}(k_Q)} =1} \left( \langle h , \mathbb{E}[ k_Q(X,\cdot) | Y] \rangle - \mathbb{E}[h(Z)] \right).
			$$
			Note $| h(x) | \leq \| h \|_{\mathcal{H}(k_Q)} \sqrt{k_Q(x,x)}$, so $\| h \|_{\mathcal{H}(k_Q)} = 1$ implies $| h(x)| \leq \sqrt{k_Q(x,x)}  $, thus 
			$$
			\mathbb{E}[k_Q(X,Y)|Y] \leq \sqrt{k_Q(Y,Y)}\sup_{| h(x) | 
				\; \leq \;
				\sqrt{k_Q(x,x)} } \left( \mathbb{E}[ h(X) | Y]- \mathbb{E}[h(Z)] \right).
			$$
			Let $\mathrm{Q}^n$ denote the $n^{\text{th}}$ step transition kernel of the Markov chain.
			From $V$-uniform ergodicity it follows that 
			$$ | \mathbb{E}[ h(X_n) | X_0=y]- \mathbb{E}[h(Z)] |= \left| \int h(x)  \mathrm{Q}^n(y,\mathrm{d}x) - \int h(x) \mathrm{d}Q(x) \right| \leq R V(y) \rho^n. $$
			Applying this to $Y=X_i$, $X=X_{i+r}$, we find
			$$
			\mathbb{E}[k_Q(X_{i+r},X_i)|X_i] \leq R \sqrt{k_Q(X_i,X_i)} V(X_i) \rho^r 
			$$
			and taking the expectation on both sides yields 
			\begin{align}
			\mathbb{E}[k_Q(X_{i+r},X_i)] 
			\; \leq \;  R \mathbb{E} \left[ \sqrt{k_Q(X_i,X_i)} V(X_i) \right]\rho^r.  \label{eq: ergod bound kir}
			\end{align}
			Finally, we can use \eqref{eq: ergod bound kir} to obtain that
			\begin{align*}
			\sum_{i=1}^{n} \sum_{r \in \{1,\dots,n\} \setminus \{i\}}\mathbb{E}[k_Q(X_{r},X_i)]
			&= 
			2\sum_{i = 1}^{n - 1} \sum_{r = 1}^{n - i}\mathbb{E}[k_Q(X_{i+r},X_i)] \\
			&\leq 
			2 R \sum_{i = 1}^{n - 1} \mathbb{E} \left[ \sqrt{k_Q(X_i,X_i)} V(X_i) \right] \sum_{r = 1}^{n - i} \rho^r.
			\end{align*}
			Thus for $C = 2 R \sum_{r = 1}^\infty \rho^r = \frac{2 R \rho}{1 - \rho}< \infty$, we have that
			$$
			\sum_{i=1}^{n} \sum_{r \in \{1,\dots,n\} \setminus \{i\}}\mathbb{E}[k_Q(X_{r},X_i)] \leq C \sum_{i = 1}^{n - 1} \mathbb{E} \left[ \sqrt{k_Q(X_i,X_i)} V(X_i) \right]
			$$
			as claimed.
		\end{proof}
		
		We can now prove the main result:
		
		\begin{proof}[Proof of \Cref{thm: main theorem}]
			Taking expectations of the bound in \Cref{prop: optimal converge}, we have that
			\begin{align*}
			\mathbb{E}\left[ D_{k_P}\left( \frac{1}{m} \sum_{j=1}^m \delta(X_{\pi(j)}) \right)^2 \right] & \leq \underbrace{\mathbb{E} \left[ D_{k_P}\left( \sum_{i=1}^n w_i^* \delta(X_i) \right)^2 \right]}_{(*)} \\
			& \hspace{60pt} + \left( \frac{1 + \log(m)}{m} \right) \underbrace{\mathbb{E}\left[ \max_{i=1,\dots,n} k_P(X_i,X_i) \right]}_{(**)} .
			\end{align*}
			In what follows we construct bounds for $(*)$ and $(**)$.
			
			\vspace{5pt}
			\noindent
			\textbf{Bounding $(*)$:}
			To bound the term $(*)$, note that 
			\begin{align*}
			D_{k_P}\left( \sum_{i=1}^n w_i^* \delta(X_i) \right) \leq D_{k_P}\left(  \frac{1}{n} \sum_{i=1}^n \delta(X_i) \right), 
			\end{align*}
			due to the optimality property of the weights $w^*$ presented in \eqref{eq: optimal weights}.
			It is therefore sufficient to study the KSD of the un-weighted empirical distribution $\frac{1}{n} \sum_{i=1}^n \delta(X_i)$.
			To this end, we have that
			\begin{align}
			\mathbb{E} \left[ D_{k_P}\left( \frac{1}{n} \sum_{i=1}^n \delta(X_i) \right)^2 \right] 
			& \; = \; 
			\frac{1}{n^2} \sum_{i=1}^{n} \mathbb{E} [k_P(X_i,X_i)] + \frac{1}{n^2} \sum_{i = 1}^{n} \sum_{r \in \{ 1, \dots, n \} \setminus \{i\} } \mathbb{E} \left[k_P(X_i,X_r)\right] . \label{eq: two terms}
			\end{align}
			To bound the first term in \eqref{eq: two terms} we use Jensen's inequality:
			\begin{eqnarray*}
				\frac{1}{n^2} \sum_{i=1}^n \mathbb{E}\left[ k_P(X_i, X_i) \right] & = & \frac{1}{n^2} \sum_{i=1}^n \mathbb{E}\left[ \frac{1}{\gamma} \log e^{ \gamma k_P(X_i, X_i) } \right]  \\
				& \leq & \frac{1}{\gamma n^2} \sum_{i=1}^n \log \left( \mathbb{E}\left[ e^{\gamma k_P(X_i, X_i) } \right] \right) \; \leq \; \frac{\log(b)}{\gamma n} 
			\end{eqnarray*}
			The second term in \eqref{eq: two terms} can be bounded via \Cref{Lemma-kernel-bounds} with $Q = P$:
			\begin{align*}
			\frac{1}{n^2} \sum_{i = 1}^{n} \sum_{r \in \{ 1, \dots, n \} \setminus \{i\} } \mathbb{E} \left[k_P(X_i,X_r)\right] & \leq \frac{C}{n^2} \sum_{i = 1}^{n - 1} \mathbb{E} \left[ \sqrt{k_P(X_i,X_i)} V(X_i) \right] 
			\leq \frac{C M (n-1)}{n^2} 
			\leq \frac{CM}{n} ,
			\end{align*}
			where $C$ is defined in \Cref{Lemma-kernel-bounds}.
			
			\vspace{5pt}
			\noindent
			\textbf{Bounding $(**)$:}
			We proceed as follows:
			\begin{eqnarray}
			\mathbb{E}\left[ \max_{i = 1,\dots, n} k_P(X_i, X_i) \right] & = & \mathbb{E}\left[ \frac{1}{\gamma} \log \max_{i = 1, \dots, n} e^{ \gamma k_P(X_i, X_i) } \right] \nonumber \\
			& \leq & \mathbb{E}\left[ \frac{1}{\gamma} \log \sum_{i = 1}^{n} e^{ \gamma k_P(X_i, X_i) } \right] \nonumber \\
			& \leq & \frac{1}{\gamma} \log \left( \sum_{i = 1}^{n} \mathbb{E}\left[ e^{\gamma k_P(X_i, X_i) } \right] \right) \; = \; \frac{\log(n b)}{\gamma} \label{eq: Sj dev end}
			\end{eqnarray}
			
			\vspace{5pt}
			\noindent
			\textbf{Overall Bound:}
			Combining our bounds on $(*)$ and $(**)$ leads to the overall bound
			\begin{align*}
			\mathbb{E}\left[ D_{k_P}\left( \frac{1}{m} \sum_{j=1}^m \delta(X_{\pi(j)}) \right)^2 \right] & \leq \frac{\log(b)}{\gamma n} + \frac{CM}{n} + \left( \frac{1 + \log(m)}{m} \right) \frac{\log(nb)}{\gamma}
			\end{align*}
			as claimed.
		\end{proof}

		\subsection{Proof of \Cref{cor: bias}} \label{app: cor bias proof}
		
		To facilitate a neat proof of \Cref{cor: bias} we first present two useful lemmas, the first of which establishes almost sure convergence in KSD of the empirical distribution based on the full MCMC output:
		
		\begin{lemma} \label{lem: AS kQ}
			Let $Q$ be a probability distribution on $\mathcal{X}$.
			Let $k_Q : \mathcal{X} \times \mathcal{X} \rightarrow \mathbb{R}$ be a reproducing kernel with $\int_{\mathcal{X}} k_Q(x,\cdot) \mathrm{d}Q = 0$ for all $x \in \mathcal{X}$.
			Consider a $Q$-invariant, time-homogeneous Markov chain $(X_i)_{i \in \mathbb{N}} \subset \mathcal{X}$, generated using a $V$-uniformly ergodic transition kernel such that $V(x) \geq \sqrt{k_Q(x,x)}$ for all $x \in \mathcal{X}$.
			Suppose that, for some $\gamma > 0$,
			\begin{align*}
			b := \sup_{i \in \mathbb{N}} \mathbb{E} \left[ e^{\gamma k_Q(X_i,X_i)} \right] < \infty, \qquad M := \sup_{i \in \mathbb{N}} \mathbb{E}\left[\sqrt{k_Q(X_i,X_i)} V(X_i)\right] < \infty .
			\end{align*}
			Then
			\begin{align*}
			D_{k_Q}\left( \frac{1}{n} \sum_{i=1}^n \delta(X_i) \right) \rightarrow 0
			\end{align*}
			almost surely as $n \rightarrow \infty$.
		\end{lemma}
		\begin{proof}
			Similarly to the proof of \Cref{thm: main theorem}, we start by bounding
			\begin{align}
			\mathbb{E} \left[ D_{k_Q}\left( \frac{1}{n} \sum_{i=1}^n \delta(X_i) \right)^2 \right] 
			& \; = \; 
			\frac{1}{n^2} \sum_{i=1}^{n} \mathbb{E} [k_Q(X_i,X_i)] + \frac{1}{n^2} \sum_{i = 1}^{n} \sum_{r \in \{ 1, \dots, n \} \setminus \{i\} } \mathbb{E} \left[k_Q(X_i,X_r)\right] . \label{eq: two terms again}
			\end{align}
			To bound the first term in \eqref{eq: two terms again} we use Jensen's inequality:
			\begin{eqnarray*}
				\frac{1}{n^2} \sum_{i=1}^n \mathbb{E}\left[ k_Q(X_i, X_i) \right] & = & \frac{1}{n^2} \sum_{i=1}^n \mathbb{E}\left[ \frac{1}{\gamma} \log e^{ \gamma k_Q(X_i, X_i) } \right]  \\
				& \leq & \frac{1}{\gamma n^2} \sum_{i=1}^n \log \left( \mathbb{E}\left[ e^{\gamma k_Q(X_i, X_i) } \right] \right) \; \leq \; \frac{\log(b)}{\gamma n} 
			\end{eqnarray*}
			The second term in \eqref{eq: two terms again} can be bounded via \Cref{Lemma-kernel-bounds}:
			\begin{align*}
			\frac{1}{n^2} \sum_{i = 1}^{n} \sum_{r \in \{ 1, \dots, n \} \setminus \{i\} } \mathbb{E} \left[k_Q(X_i,X_r)\right] & \leq \frac{C}{n^2} \sum_{i = 1}^{n - 1} \mathbb{E} \left[ \sqrt{k_Q(X_i,X_i)} V(X_i) \right] 
			\leq \frac{C M (n-1)}{n^2} 
			\leq \frac{CM}{n} ,
			\end{align*}
			where $C$ is defined in \Cref{Lemma-kernel-bounds}.
			This establishes that 
			\begin{align}
			\mathbb{E} \left[ D_{k_Q}\left( \frac{1}{n} \sum_{i=1}^n \delta(X_i) \right)^2 \right] & \leq \frac{\log(b)}{\gamma n} + \frac{CM}{n} =: c_1(n) . \label{eq: L2 ksd bound Q}
			\end{align}
			To simplify notation we adopt the shorthand 
			$$
			D_n := D_{k_Q}\left( \frac{1}{n} \sum_{i=1}^n \delta(X_i) \right)
			$$
			in this proof only.
			Fix $\epsilon > 0$.
			If $D_n > \epsilon$ occurs infinitely often (i.o.) then there are infinitely many $r$ such that $\max_{r^2 \leq n < (r+1)^2} D_n^2 > \epsilon$, so that
			\begin{align}
			    \mathbb{P} \left[ D_n^2 > \epsilon \text{ i.o.} \right] \leq \mathbb{P}\left[ \max_{r^2 \leq n < (r+1)^2 } D_n^2 > \epsilon \text{ i.o.} \right]. \label{eq: io bound}
			\end{align}
			Now, consider the bound
			\begin{align*}
			    \sum_{r=1}^\infty \mathbb{P}\left[ \max_{r^2 \leq n < (r+1)^2} D_n^2 > \epsilon \right] \leq  
               \underbrace{ \sum_{r=1}^\infty \mathbb{P}\left[ D_{r^2}^2 > \frac{\epsilon}{2} \right] }_{(*)}
                     + \underbrace{ \sum_{r=1}^\infty \mathbb{P}\left[ \max_{r^2 \leq n < (r+1)^2 } |D_n^2 -  D_{r^2}^2| > \frac{\epsilon}{2} \right] }_{(**)} ,
			\end{align*}
			where the inequality follows from the fact that, for any $a,b \in \mathbb{R}$, if $a > \epsilon$ then either $b > \frac{\epsilon}{2}$ or $|a-b| > \frac{\epsilon}{2}$.
			In the remainder we will show that the sums $(*)$ and $(**)$ are finite, so that from the Borel--Cantelli lemma
			\begin{align}
			\mathbb{P}\left[ \max_{r^2 \leq n < (r+1)^2} D_n^2 > \epsilon \text{ i.o.} \right] = 0 . \label{eq: borel cantelli new}
			\end{align}
			Since \eqref{eq: borel cantelli new} holds for all $\epsilon > 0$, it will follow from \eqref{eq: io bound} that $\mathbb{P}[D_n \rightarrow 0] = 1$, as claimed.
			
			\vspace{5pt}
			\noindent
			\textbf{Bounding $(*)$:}
			From the Markov inequality and \eqref{eq: L2 ksd bound Q} we have that, for any $\epsilon > 0$,
			\begin{align*}
			\mathbb{P}\left[ D_{r^2}^2 > \frac{\epsilon}{2} \right] \leq \frac{ 2 }{ \epsilon } \mathbb{E} [D_{r^2}^2]  \leq \frac{ 2 }{ \epsilon } c_1(r^2)  .
			\end{align*}
			Since $c_1(r^2) = O(1/r^2)$, it follows that
			\begin{align*}
			(*) = \sum_{r=1}^\infty \mathbb{P}\left[ D_{r^2}^2 > \frac{\epsilon}{2} \right] & \leq \frac{2}{\epsilon} \sum_{r=1}^\infty c_1(r^2) < \infty.
			\end{align*}
			
			\vspace{5pt}
			\noindent
			\textbf{Bounding $(**)$:}
			For $r^2 \leq n < (r+1)^2$,
			\begin{align*}
			|D_n^2 - D_{r^2}^2| = \frac{1}{n^2} |n^2 D_n^2 - n^2 D_{r^2}^2| 
			& \leq \frac{1}{n^2} |n^2 D_n^2 - r^4 D_{r^2}^2| + \frac{1}{n^2} |(n^2-r^4) D_{r^2}^2| \\
			& \leq \frac{1}{r^4} |n^2 D_n^2 - r^4 D_{r^2}^2| + \frac{2(r+1)}{r^2} D_{r^2}^2 ,
			\end{align*}
			and also that, using the reproducing property and Cauchy-Schwarz,
			\begin{align*}
			|n^2 D_n^2 - r^4 D_{r^2}^2| & = \left| \sum_{i,j = r^2 + 1}^n k_Q(X_i, X_j) \right| \\
			& \leq \sum_{i,j = r^2 + 1}^n \sqrt{k_Q(X_i, X_i)} \sqrt{k_Q(X_j, X_j)} \\
			& \leq (n-r^2)^2 \max_{1 \leq i \leq n} k_Q(X_i,X_i) \leq 4 r^2 \max_{1 \leq i \leq n} k_Q(X_i,X_i) .
			\end{align*}
			Similarly again to the proof of \Cref{thm: main theorem} we have the bound
			\begin{eqnarray}
			\mathbb{E}\left[ \max_{i = 1,\dots, n} k_Q(X_i, X_i) \right] & = & \mathbb{E}\left[ \frac{1}{\gamma} \log \max_{i = 1, \dots, n} e^{ \gamma k_Q(X_i, X_i) } \right] \nonumber \\
			& \leq & \mathbb{E}\left[ \frac{1}{\gamma} \log \sum_{i = 1}^{n} e^{ \gamma k_Q(X_i, X_i) } \right] \nonumber \\
			& \leq & \frac{1}{\gamma} \log \left( \sum_{i = 1}^{n} \mathbb{E}\left[ e^{\gamma k_Q(X_i, X_i) } \right] \right) \; = \; \frac{\log(n b)}{\gamma} \label{eq: Sj dev end again}
			\end{eqnarray}
			so that, taking expectations, we obtain the bound
			\begin{align*}
			\mathbb{E}\left[ \max_{r^2 \leq n < (r+1)^2} |D_n^2 - D_{r^2}^2| \right] & \leq \frac{4}{r^2} \mathbb{E}\left[ \max_{1 \leq i \leq n} k_Q(X_i,X_i) \right] + \frac{2(r+1)}{r^2} \mathbb{E}\left[ D_{r^2}^2 \right]  \\
			& \leq \frac{4}{r^2} \frac{\log(nb)}{\gamma} + \frac{2(r+1)}{r^2} c_1(r^2)  \\
			& \leq \frac{8}{r^2} \frac{\log((r+1)b)}{\gamma} + \frac{2(r+1)}{r^2} c_1(r^2) =: c_2(r) 
			\end{align*}
			where $c_2(r) = O(\log(r) / r^2)$.
			Using the Markov inequality,
			\begin{align*}
			(**) = \sum_{r=1}^\infty \mathbb{P}\left[ \max_{r^2 \leq n < (r+1)^2} |D_n^2 - D_{r^2}^2| > \frac{\epsilon}{2} \right] & \leq \frac{2}{\epsilon} \sum_{r=1}^\infty \mathbb{E}\left[ \max_{r^2 \leq n < (r+1)^2} |D_n^2 - D_{r^2}^2|  \right] \\
			& \leq \frac{2}{\epsilon} \sum_{r=1}^\infty c_2(r) < \infty .
			\end{align*}
			This completes the proof.
		\end{proof}
		
		Our second lemma is a technical result on almost sure convergence:
		
		\begin{lemma} \label{lem: AS max k Q}
			Let $f$ be a non-negative function on $\mathcal{X}$.
			Consider a sequence of random variables $(X_i)_{i \in \mathbb{N}} \subset \mathcal{X}$ such that, for some $\gamma > 0$,
			\begin{align*}
			b := \sup_{i \in \mathbb{N}} \mathbb{E} \left[ e^{\gamma f(X_i)} \right] < \infty
			\end{align*}
			If $m \leq n$ and the growth of $n$ is limited to at most $\log(n) = O(m^{\beta/2})$ for some $\beta < 1$, then
			\begin{align*}
			\left( \frac{\log(m)}{m} \right) \max_{i=1,\dots,n} f(X_i)  \rightarrow 0
			\end{align*}
			almost surely as $m,n \rightarrow \infty$.
		\end{lemma}
		\begin{proof}
			To simplify notation we adopt the shorthand 
			$$
			E_m := \left( \frac{\log(m)}{m} \right) \max_{i=1,\dots,n} f(X_i)
			$$
			in this proof only, where $n = n(m)$.
			The argument is similar to the proof of \Cref{lem: AS kQ}.
			Fix $\epsilon > 0$.
			If $E_m > \epsilon$ i.o. then there are infinitely many $r$ such that $\max_{r^2 \leq m < (r+1)^2} E_m > \epsilon$, so that
			\begin{align}
			    \mathbb{P} \left[ E_m > \epsilon \text{ i.o.} \right] \leq \mathbb{P}\left[ \max_{r^2 \leq m < (r+1)^2 } E_m > \epsilon \text{ i.o.} \right]. \label{eq: io bound 2}
			\end{align}
			Now, consider the bound
			\begin{align*}
			    \sum_{r=1}^\infty \mathbb{P}\left[ \max_{r^2 \leq m < (r+1)^2} E_m > \epsilon \right] \leq  
               \underbrace{ \sum_{r=1}^\infty \mathbb{P}\left[ E_{r^2} > \frac{\epsilon}{2} \right] }_{(*)}
                     + \underbrace{ \sum_{r=1}^\infty \mathbb{P}\left[ \max_{r^2 \leq m < (r+1)^2 } |E_m -  E_{r^2}| > \frac{\epsilon}{2} \right] }_{(**)} ,
			\end{align*}
			In the remainder we will show that the sums $(*)$ and $(**)$ are finite, so that from the Borel--Cantelli lemma
			\begin{align}
			\mathbb{P}\left[ \max_{r^2 \leq m < (r+1)^2} E_m > \epsilon \text{ i.o.} \right] = 0 . \label{eq: borel cantelli new 2}
			\end{align}
			Since \eqref{eq: borel cantelli new 2} holds for all $\epsilon > 0$, it will follow from \eqref{eq: io bound 2} that $\mathbb{P}[E_m \rightarrow 0] = 1$, as claimed.
			
			\vspace{5pt}
			\noindent
			\textbf{Bounding $(*)$:}
			Similarly to the proof of \Cref{thm: main theorem}, we have the bound
			\begin{eqnarray*}
				\mathbb{E}\left[ \max_{i = 1,\dots, n} f(X_i) \right] & = & \mathbb{E}\left[ \frac{1}{\gamma} \log \max_{i = 1, \dots, n} e^{ \gamma f(X_i) } \right]  \\
				& \leq & \mathbb{E}\left[ \frac{1}{\gamma} \log \sum_{i = 1}^{n} e^{ \gamma f(X_i) } \right]  \\
				& \leq & \frac{1}{\gamma} \log \left( \sum_{i = 1}^{n} \mathbb{E}\left[ e^{\gamma f(X_i) } \right] \right) \; = \; \frac{\log(n b)}{\gamma} 
			\end{eqnarray*}
			and thus from the Markov inequality we have that, for any $\epsilon > 0$,
			\begin{align*}
			\mathbb{P}\left[ E_m > \frac{\epsilon}{2} \right] \leq \frac{ 2 }{ \epsilon} \mathbb{E} [E_m]  \leq \frac{ 2 }{ \epsilon} c_1(m)
			\end{align*}
			where 
			$$
			c_1(m) := \frac{\log(m)}{m} \frac{\log(nb)}{\gamma}.
			$$
			The assumption $\log(n) = O(m^{\beta/2})$ for some $\beta < 1$ implies that $\log(nb) \leq \alpha m^{\beta/2} + \log(b)$ for some constant $\alpha \in (0,\infty)$.
			Thus
			\begin{align*}
			c_1(r^2) \leq \frac{2\log(r)}{r^2} \left( \frac{\alpha r^{\beta} + \log(b)}{\gamma} \right) .
			\end{align*}
			This shows that $c_1(r^2) = O(\log(r) / r^{2-\beta})$, and it follows that
			\begin{align*}
			(*) = \sum_{r=1}^\infty \mathbb{P}\left[ E_{r^2} > \frac{\epsilon}{2} \right] & \leq \frac{2}{\epsilon} \sum_{r=1}^\infty c_1(r^2) < \infty.
			\end{align*}
			
			\vspace{5pt}
			\noindent
			\textbf{Bounding $(**)$:}
			For the second term we argue that, since $E_m \geq 0$,
			\begin{align*}
			\max_{r^2 \leq m < (r+1)^2} |E_m - E_{r^2}| & \leq 2 \max_{r^2 \leq m < (r+1)^2} E_m \\
			& \leq \frac{2 \log(r^2)}{r^2} \max_{i=1,\dots,n((r+1)^2)} f(X_i)
			\end{align*}
			so that, taking expectations,
			\begin{align*}
			\mathbb{E} \left[ \max_{r^2 \leq m < (r+1)^2} |E_m - E_{r^2}| \right] & \leq \frac{ 4 \log(r)}{r^2} \frac{\log(n((r+1)^2) b)}{\gamma} =: c_2(r)
			\end{align*}
			Using the bound $\log(nb) \leq \alpha m^{\beta/2} + \log(b)$, the quantity $c_2(r)$ just defined satisfies 
			\begin{align*}
			c_2(r) & \leq \frac{4 \log(r)}{r^2} \left( \frac{\alpha (r+1)^\beta + \log(b)}{\gamma} \right) ,
			\end{align*}
			which is $O(\log(r) / r^{2 - \beta})$.
			Using the Markov inequality and the fact that $(a-b)^2 \leq |a^2-b^2|$,
			\begin{align*}
			(**) = \sum_{r=1}^\infty \mathbb{P}\left[ \max_{r^2 \leq m < (r+1)^2} |E_m - E_{r^2}| > \frac{\epsilon}{2} \right] & \leq \frac{2}{\epsilon} \sum_{r=1}^\infty \mathbb{E}\left[ \max_{r^2 \leq m < (r+1)^2} |E_m - E_{r^2}|  \right] \\
			& \leq \frac{2}{\epsilon} \sum_{r=1}^\infty c_2(r) < \infty .
			\end{align*}
			This completes the proof.
		\end{proof}

		Now we present the proof of \Cref{cor: bias}:
		
		\begin{proof}[Proof of \Cref{cor: bias}]
			Our starting point is again the bound in \Cref{prop: optimal converge}:
			\begin{align*}
			D_{k_P}\left( \frac{1}{m} \sum_{j=1}^m \delta(X_{\pi(j)}) \right)^2 & \leq \underbrace{ D_{k_P}\left( \sum_{i=1}^n w_i^* \delta(X_i) \right)^2 }_{(*)} + \underbrace{ \left( \frac{1 + \log(m)}{m} \right)  \max_{i=1,\dots,n} k_P(X_i,X_i) }_{(**)} .
			\end{align*}
			For term $(*)$, note that 
			\begin{align*}
			D_{k_P}\left( \sum_{i=1}^n w_i^* \delta(X_i) \right) \leq D_{k_P}\left(   \sum_{i=1}^n w_i \delta(X_i) \right), \qquad w_i := \frac{  \frac{\mathrm{d}P }{ \mathrm{d}Q}(X_i) } { \sum_{i'=1}^n \frac{\mathrm{d}P }{ \mathrm{d}Q}(X_{i'}) }
			\end{align*}
			due to the optimality property of the weights $w^*$ presented in \eqref{eq: optimal weights}.
			Further note that
			\begin{align*}
			D_{k_P}\left(\sum_{i=1}^n w_i \delta(X_i) \right) = \frac{1}{ \frac{1}{n} \sum_{i'=1}^n \frac{\mathrm{d}P}{\mathrm{d}Q}(X_{i'}) } D_{k_Q}\left( \frac{1}{n} \sum_{i=1}^n \delta(X_i) \right)
			\end{align*}
			where $k_Q(x,y) := \frac{\mathrm{d}P}{\mathrm{d}Q}(x) k_P(x,y) \frac{\mathrm{d}P}{\mathrm{d}Q}(y)$ is a reproducing kernel such that $\int_{\mathcal{X}} k_Q(x,\cdot) \mathrm{d}Q = 0$ for all $x \in \mathcal{X}$.
			The preconditions of \Cref{cor: bias} ensure that $V(x) \geq \sqrt{k_Q(x,x)}$ and 
			\begin{align*}
			\sup_{i \in \mathbb{N}} \mathbb{E} \left[ e^{\gamma k_Q(X_i,X_i)} \right] \leq b < \infty, \qquad M = \sup_{i \in \mathbb{N}} \mathbb{E}\left[\sqrt{k_Q(X_i,X_i)} V(X_i)\right] < \infty .
			\end{align*}
			Therefore we may apply \Cref{lem: AS kQ} to obtain that
			\begin{align*}
			D_{k_Q}\left( \frac{1}{n} \sum_{i=1}^n \delta(X_i) \right) \rightarrow 0
			\end{align*}
			almost surely as $n \rightarrow \infty$.
			Moreover, since $\int_{\mathcal{X}} \left| \frac{\mathrm{d}P}{\mathrm{d}Q} \right| \mathrm{d}Q < \infty$, it follows from \citet[][Theorem 17.0.1, part (i)]{Meyn2012} that
			\begin{align*}
			\frac{1}{n} \sum_{i=1}^n \frac{\mathrm{d}P}{\mathrm{d}Q}(X_i) \rightarrow \int_{\mathcal{X}} \frac{\mathrm{d}P}{\mathrm{d}Q} \mathrm{d}Q = 1
			\end{align*}
			almost surely as $n \rightarrow \infty$.
			Standard properties of almost sure convergence thus imply that $(*) \rightarrow 0$ almost surely as $n \rightarrow \infty$.
			
			For term $(**)$, we notice that
			$$
			\sup_{i \in \mathbb{N}} \mathbb{E}\left[ e^{ \gamma k_P(X_i,X_i) } \right] \leq b < \infty
			$$
			and we can therefore use \Cref{lem: AS max k Q} with $f(x) = k_P(x,x)$ to deduce that $(**) \rightarrow 0$ almost surely as $m,n \rightarrow \infty$.
			
			Thus we have established that
			\begin{align}
			D_{k_P} \left( \frac{1}{m} \sum_{j=1}^m \delta(X_{\pi(j)}) \right) \rightarrow 0
			\end{align}
			almost surely as $m,n \rightarrow \infty$.
			The final part of the statement of \Cref{cor: bias} is immediate from \Cref{prop: convergence control}.
		\end{proof}

	\subsection{Satisfying the Conditions of \Cref{thm: main theorem}}
	\label{ap: satisfy conditions}
	
The conditions for \Cref{thm: main theorem} are agnostic to the specific Markov chain used (e.g. Metropolis--Hastings, Gibbs sampling, etc), making it quite general.
In this appendix we discuss how explicit sufficient conditions can be obtained if one restricts attention to a specific MCMC method.
Here we focus on the \textit{Metropolis-adjusted Langevin algorithm} (MALA; whose definition is recalled in \Cref{ap: MCMC methods}).

Let $q(x,y)$ denote the probability density for the proposal $x \rightarrow y$ in MALA, with step size $\epsilon > 0$ fixed.
Let $A(x) \subset \mathbb{R}^d$ be the set of values $y$ which, if a move $x \rightarrow y$ is proposed, then $y$ is always accepted.
Let $I(x) := \{y \in \mathbb{R}^d : \|y\| \leq \|x\|\}$ and let $A \Delta B := (A \cup B) \setminus (A \cap B)$.
MALA is said to be \textit{inwardly convergent} if
\begin{align}
\lim_{\|x\| \rightarrow \infty} \int_{A(x) \Delta I(x)} q(x,y) \mathrm{d}y = 0 , \label{eq: inward conv}
\end{align}
see Section 4 of \cite{Roberts1996}.
The following result will then be established:

\begin{lemma} \label{lem: satisfy assumptions}
Let $P \in \mathcal{P}$ be distantly dissipitive on $\mathcal{X} = \mathbb{R}^d$, let $\mathbb{E}_{X \sim P}[\exp(\beta \|X\|^2)] < \infty$ for some $\beta \in (0,\infty)$, and assume that MALA is inwardly convergent.
Then, with kernel $k(x,y) = (1 + \|x-y\|^2)^{-1/2}$, the conditions of \Cref{thm: main theorem} are satisfied.
\end{lemma}
\begin{proof}
This proof exploits Theorem 9 of \cite{Chen2019}, which establishes $V$-uniform ergodicity of MALA for each of $V(x) = \exp(\gamma \|x\|)$ (any $\gamma >0$), $V(x) = \exp(\gamma \|x\|^2)$ (for $\gamma > 0$ sufficiently small) and $V(x) = 1 + \|x\|^\gamma$ ($\gamma \in \{1,2\}$).
Each choice of $V$ leads to a different set of preconditions for \Cref{thm: main theorem}, and the claimed result follows from taking $V(x) = 1 + \|x\|$.
Note that w.l.o.g. we can consider $V(x) = C(1 + \|x\|)$ for any fixed $C \in [1,\infty)$, since the constant $C$ cancels in the definition of $V$-uniform ergodicity.
The conditions for \Cref{thm: main theorem} now simplify as follows:

\vspace{5pt}
\noindent \textbf{First Condition}: 
The form of $k$ implies that $k_P(x,x) = d + \|\nabla \log p(x) \|^2$ (see \Cref{rem: iso kernel rem}).
Since $\nabla \log p$ is assumed to be Lipschitz we have, for $C$ sufficiently large,
\begin{align*}
    V(x) \geq \sqrt{k_P(x,x)} & \Leftrightarrow C(1+\|x\|) \geq \sqrt{d + \|\nabla \log p(x)\|^2} \\
    & \Leftarrow \|\nabla \log p(x)\|^2 \leq C_1 + C_2 \|x\|^2 \qquad (\text{some }C_1, C_2 > 0) ,
\end{align*}
so that the first condition of \Cref{thm: main theorem} is automatically satisfied.

\vspace{5pt}
\noindent \textbf{Second Condition}:
Now, suppose MALA is also $\tilde{V}$-uniformly ergodic; i.e. $\|\mathrm{P}^i(x,\cdot) - P\|_{\tilde{V}} \leq \tilde{R} \tilde{V}(x) \tilde{\rho}^n$ for some $\tilde{R} \in [0,\infty)$, $\tilde{\rho} \in (0,1)$.
Let $X \sim P$ independent of $X_i$.
Then
\begin{align*}
    \mathbb{E}[\tilde{V}(X_i)] & = \mathbb{E}[\tilde{V}(X)] + \mathbb{E}[\tilde{V}(X_i) - \tilde{V}(X)] \\
    & \leq \mathbb{E}[\tilde{V}(X)] + \|\mathrm{P}^i(X_0,\cdot) - P\|_{\tilde{V}}  \|\tilde{V}\|_{\tilde{V}} \\
    & \leq \mathbb{E}[\tilde{V}(X)] + \tilde{R} \tilde{V}(X_0) \tilde{\rho}^i \rightarrow \mathbb{E}[\tilde{V}(X)] \qquad \text{as } i \rightarrow \infty
\end{align*}
which shows that
\begin{align}
    \mathbb{E}[\tilde{V}(X)] < \infty & \Rightarrow \sup_{i \in \mathbb{N}} \mathbb{E}[\tilde{V}(X_i)] < \infty . \label{eq: remove sup 2}
\end{align}
From Theorem 9 of \cite{Chen2019} we have $\tilde{V}$-uniform ergodicity for $\tilde{V}(x) = \exp(\gamma C_2 \|x\|^2)$ (for $\gamma > 0$ sufficiently small), which shows that, for all $\gamma > 0$ sufficiently small
\begin{align}    
    b = \sup_{i \in \mathbb{N}} \mathbb{E}[\exp(\gamma k_P(X_i,X_i))] < \infty & \Leftrightarrow \sup_{i \in \mathbb{N}} \mathbb{E}[\exp( \gamma \|\nabla \log p(X_i)\|^2)] < \infty \nonumber \\
    & \Leftarrow \sup_{i \in \mathbb{N}} \mathbb{E}[\exp(\gamma C_2 \|X_i\|^2)] < \infty \nonumber \\
    & \Leftarrow \mathbb{E}[\exp(\gamma C_2 \|X\|^2)] < \infty \qquad (\text{due to \eqref{eq: remove sup 2}}). \label{eq: exp moment 2}
\end{align}
Thus the second condition of \Cref{thm: main theorem} is satisfied if $\mathbb{E}[\exp(\beta \|X\|^2)] < \infty$ for some $\beta \in (0,\infty)$.
    
\vspace{5pt}
\noindent \textbf{Third Condition}:
A similar argument used for the second condition can again be used, this time with $\tilde{V}(x) = 1 + \|x\|^s$ for $s \in \{1,2\}$.
Specifically, we have that
\begin{align*}
    M = \sup_{i \in \mathbb{N}} \mathbb{E}[\sqrt{k_P(X_i,X_i)} V(X_i) ] < \infty & \Leftrightarrow \sup_{i \in \mathbb{N}} \mathbb{E}[\sqrt{d + \|\nabla \log p(X_i)\|^2} (1 + \|X_i\|) ] < \infty \\
    & \Leftarrow \sup_{i \in \mathbb{N}} \mathbb{E}[\sqrt{d + C_1 + C_2\|X_i\|^2} (1 + \|X_i\|) ] < \infty \\
    & \Leftarrow \sup_{i \in \mathbb{N}} \mathbb{E}[1 + \|X_i\|], \; \sup_{i \in \mathbb{N}} \mathbb{E}[1 + \|X_i\|^2] < \infty \\
    & \Leftarrow \mathbb{E}[1 + \|X\|], \; \mathbb{E}[1 + \|X\|^2] < \infty \quad \text{(due to \eqref{eq: remove sup 2})} 
\end{align*}
with the latter being implied by the stronger moment condition in \eqref{eq: exp moment 2}.
\end{proof}
	
	The sufficient conditions presented in \Cref{lem: satisfy assumptions} may be explicitly verified, with the possible exception of the inwards convergence condition of \citet{roberts2004general}.

	\section{Experimental Protocol} \label{ap: MCMC methods}
	
	In this appendix we describe the generic structure of a parameter inference problem for a system of ODEs, that forms our empirical test-bed.
	
	Consider the solution $u$ of a system of $q$ coupled ODEs of the form
	\begin{align}
	\frac{\mathrm{d}u_1}{\mathrm{d}t} &= F_1(t,u_1,\dots,u_q; x) \nonumber \\
	& \vdots \nonumber \\
	\frac{\mathrm{d}u_q}{\mathrm{d}t} &= F_q(t,u_1,\dots,u_q; x) , \label{eq: ODE}
	\end{align}
	together with the initial condition $u(0) = u^0 \in \mathbb{R}^q$.
	The functions $F_i$ that define the gradient field are assumed to depend on a number $d$ of parameters, collectively denoted $x \in \mathbb{R}^d$, and the $F_i$ are assumed to be differentiable with respect to $u_1, \dots, u_q$ and $x$.
	It is assumed that $u(t)$ exists and is unique on an interval $t \in [0, T]$ for all values $x \in \mathbb{R}^d$.
	For simplicity in this work we assumed that the initial condition $u^0$ is not dependent on $x$ and is known.
	The goal is to make inferences about the parameters $x$ based on noisy observations of the state vector $u(t_i)$ at discrete times $t_i$; this information is assumed to be contained in a likelihood of the form
	\begin{align}
	\mathcal{L}(x) & := \prod_{i=1}^N \phi_i(u(t_i))
	\label{eq:original_likelihood}
	\end{align}
	where the functions $\phi_i : \mathbb{R}^q \rightarrow [0,\infty)$, describing the nature of the measurement at time $t_i$, are problem-specific and to be specified.
	The parameter $x$ is endowed with a prior density $\pi(x)$ and the posterior of interest $P$ admits a density $p(x) \propto \pi(x) \mathcal{L}(x)$.
	Computation of the gradient $\nabla \log p$ therefore requires computation of $\nabla \log \pi$ and $\nabla \log \mathcal{L}$; the latter can be performed by augmenting the system in \eqref{eq: ODE} with the sensitivity equations, as described next.
	
	Straight-forward application of the chain rule leads to the following expression for the gradient of the log-likelihood:
	\begin{align*}
	(\nabla \log \mathcal{L})(x) &= - \sum_{i=1}^N \frac{\partial u}{\partial x}(t_i) (\nabla \log \phi_i)(u(t_i)),
	\end{align*}
	where $(\partial u/\partial x)_{r,s} \coloneqq \partial u_r / \partial x_s$ is the matrix of \emph{sensitivities} of the solution $u$ to the parameter $x$ and is time-dependent.
	Sensitivities can be computed by augmenting the system in \eqref{eq: ODE} and simultaneously solving the \emph{forward sensitivity equations}
	\begin{align}
	\frac{\mathrm{d}}{\mathrm{d}t} \left( \frac{\partial u_r}{\partial x_s} \right) &= \frac{\partial F_r}{\partial x_s} + \sum_{l=1}^q \frac{\partial F_r}{\partial u_l} \frac{\partial u_l}{\partial x_s} \label{eq:forward_sensitivities}
	\end{align}
	together with the initial condition $(\partial u_r / \partial x_s)(0) = 0$, which follows from the independence of $u^0$ and $x$.
	
	The experiments reported in \Cref{sec: empirical} were based on four distinct Metropolis--Hastings MCMC methods, whose details have not yet been described.
	The generic structure of the proposal mechanism is
	$x^* = x_{n-1} + H \nabla \log p(x_{n-1}) 
	+ G \xi_n$,
	where the $\xi_n \sim \mathcal{N}(0,I)$ are independent.
	The matrices $H$ and $G$ are specified in \Cref{table:transition_kernel_parms}.
	Our implementation of these samplers interfaces with the \texttt{CVODES} library \citep{hindmarsh2005sundials}, which presents a practical barrier to reproducibility.
	Moreover, the CPU time required to obtain MCMC samples was approximately two weeks for the calcium model.
	Since our research focused on post-processing of MCMC output, rather than MCMC itself, we directly make available the full output from each sampler on each model considered at
	\begin{center}
	    \url{https://doi.org/10.7910/DVN/MDKNWM}.
	\end{center}
	This Harvard database download link consists of a single ZIP archive (1.5GB) that contains, for each ODE model and each MCMC method, the states $(x_i)_{i=1}^n$ visited by the Markov chain, their corresponding gradients $\nabla \log p(x_i)$ and the values $p(x_i)$ up to an unknown normalisation constant.
	The \texttt{Stein Thinning} software described in \ref{subsec: software} can be used to post-process these datasets at minimal effort, enabling our findings to be reproduced.

	\begin{table}
		\centering
		\footnotesize
		\begin{tabular}{|p{5cm}||c|c||p{6cm}|} \hline
			\textbf{Proposal} & $H$ & $G$ & \textbf{Details} \\
			\hline 
			\hline 
			\texttt{RW} & $0$& $\epsilon I$ & Step size $\epsilon$ selected following \cite{Roberts2001} \\
			\hline 
			\texttt{ADA-RW} \citep{haario1999adaptive} & $0$&  $\sqrt{\hat{\Sigma}}$ &  $\hat{\Sigma}$ is the sample covariance matrix of preliminary MCMC output \\
			\hline 
			\texttt{MALA} \citep{Roberts1996}	& $\frac{\epsilon^2}{2}I$&  $\epsilon I$ & Step size $\epsilon$ selected following \cite{Roberts2001}  \\
			\hline 
			\texttt{P-MALA} \citep{Girolami2011}	& $\frac{\epsilon^2}{2}
			M^{-1}(x_{n-1})$&  $\epsilon
			\sqrt{M^{-1}(x_{n-1})}$ &  $M(x) = F(x) + \Sigma_0^{-1}$ where $F(x)$ is the Fisher information matrix at $x$ and $\Sigma_0$ is the prior covariance matrix. \\ \hline
		\end{tabular}
		\caption{Parameters $H$ and $G$ used in the Metropolis--Hastings proposal. }
		\label{table:transition_kernel_parms}
	\end{table}

	\section{Convergence Diagnostics for MCMC} \label{section: background on diagnostics}
	
	Rigorous approaches for selecting a burn-in period $b$ have been proposed by authors including \cite{meyn1994computable, rosenthal1995minorization, roberts1999bounds}; see also \cite{jones2001honest}.
	Unfortunately, these often involve conditions that are difficult to establish
	\citep[][discuss how some of the terms appearing in these conditions can be estimated]{biswas2019estimating}, or, when they hold, they provide loose bounds, implying an unreasonably long burn-in period.
	
	Convergence diagnostics have emerged as a practical solution to the need to test for non-convergence of MCMC. 
	Their use is limited to reducing bias in MCMC output; they are not optimised for the fixed $n$ setting, which requires a bias-variance trade-off.
	Nevertheless, convergence diagnostics constitute the principal means by which MCMC output is post-processed.
	In this section we recall standard practice for selection of a burn-in period $b$ in constructing an estimator of the form \eqref{eq: std post process}, focussing on the widely-used diagnostics of \cite{gelman1992, brooks1998general, gelman2014} (the \emph{GR diagnostic}), as well as the more recent work of \cite{vats2018revisiting} (the \emph{VK diagnostic}). 
	
	The GR diagnostic is based on running $L$ independent chains, each of length $n$, with starting points that are over-dispersed with respect to the target. Obtaining initial points with such characterisation is not trivial because the target is not known beforehand; we refer to the original literature for advice on how to select these initial points, but, in practice, it is not uncommon to guess them. 
	When the support of the target distribution is uni-dimensional  (or when $d>1$, but a specific uni-dimensional summary $f(x)$ is used), the GR diagnostics ($\hat{R}^{\text{GR},L}$) is obtained as the square root of the ratio of two estimators of the variance $\sigma^2$ of the target.  
	In particular,
	\begin{align}
	\hat{R}^{\text{GR},L} := \sqrt{\frac{\hat{\sigma}^2}{s^2}},
	\label{eq:GR_diagnostic}
	\end{align}
	where $s^2$ is the (arithmetic) mean of the sample variances $s^2_l$, $l = 1,\dots,L$, of the chains, which typically provides an underestimate of $\sigma^2$, and $\hat{\sigma}^2$ is constructed as an overestimate of the target variance
	\begin{align*}
	\hat{\sigma}^2 := \frac{n-1}{n} s^2 + \frac{B}{n},   
	\end{align*}
	where the term $B/n$ is an estimate of the asymptotic variance of the sample mean of the Markov chain.
	In the original GR diagnostics, this asymptotic variance was estimated as the sample variance of the means $\bar{X}_l$, $l = 1,\dots,L$, from the $L$ chains, leading to
	\begin{align*}
	\frac{B}{n} = \frac{1}{L-1} \sum_{l=1}^L \left( \bar{X}_{l} - \frac{1}{L} \sum_{l'=1}^L \bar{X}_l \right)^2 .
	\end{align*}
	The improved VK diagnostic, $\hat{R}^{\text{VK},L}$, is formally obtained in the same way as \eqref{eq:GR_diagnostic}, but with more efficient estimators $\tau^2/n$ for the asymptotic variance used in place of $B/n$. 
	A number of options are available here, but the (lugsail) batch mean estimator of \cite{vats2018lugsail} is recommended because it is guaranteed to be biased from above, while still being 
	consistent (in our simulations we use  batches of size $\sqrt[3]{n}$). 
	This gain in efficiency leads to improved performance of the VK diagnostic over the GR diagnostic, in the sense that it is less sensitive to the randomness in the Markov chains and the number of chains used. 
	In particular, $\hat{R}^{\text{VK},L}$ 
	can be computed using one chain only $(L=1)$, which has clear practical appeal.

	For an ergodic Markov chain, $\hat{R}^{\text{GR},L}$ and $\hat{R}^{\text{VK},L}$ converge to 1 as $n\rightarrow\infty$, so that selection of a suitable burn-in period $b$ amounts to observing when these diagnostics are below $1 + \delta$, where $\delta$ is a suitable threshold. 
	In the literature on $\hat{R}^{\text{GR},L}$, the somewhat arbitrary choice $\delta = 0.1$ is commonly used, see \citet[Ch.\ 11.5]{gelman2014} and the survey in \cite{vats2018revisiting}. 
	In the literature on $\hat{R}^{\text{KV},L}$, \cite{vats2018revisiting} showed how $\delta$ can be selected by exploiting the relationship between $\hat{R}^{\text{VK},L}$ and the effective sample size (ESS) when estimating the mean of the target.
	In particular, it is possible to re-write 
	\begin{align}
	\hat{R}^{\text{VK},L} = \sqrt{\left( \frac{n-1}{n}\right) + \frac{L}{\reallywidehat{\text{ESS}}}}
	\label{eq:ESS_R_hat}
	\end{align}
	where $\reallywidehat{\text{ESS}}$ is a strongly consistent estimator of the ESS.
	One can therefore (approximately) select a $\delta$ threshold that corresponds to a pre-specified value of the ESS.
	The literature on error assessment for MCMC provides guidance on how large the ESS ought to be in order that the width of a $(1-\alpha)\%$ confidence interval for the mean is less that a specified threshold $\epsilon$; see \cite{jones2001honest, flegal2008markov, vats2019multivariate}:
	\begin{align}
	\reallywidehat{\text{ESS}} \geq M_{\alpha, \epsilon} := \frac{2^{2} \pi}{(\Gamma(1/2))^{2}}\frac{\chi^2_{1-\alpha}}{\epsilon^2}, \label{eq:bound_ESS}
	\end{align}
	where  $\Gamma(\cdot)$ is the Gamma function,  $\chi^2_{1-\alpha}$ is the $(1-\alpha)^{\text{th}}$ quantile of the $\chi^2$ distribution with one degree of freedom. 
	Plugging \eqref{eq:bound_ESS} in \eqref{eq:ESS_R_hat} leads to the conclusion that, after the first iteration for which $\hat{R}^{\text{VK},L}$ is below $1 + \delta$, where 
	\begin{align}
	\delta \equiv \delta(L, \alpha, \epsilon) = \sqrt{1+ \frac{L}{M_{\alpha, \epsilon}}} - 1, 
	\label{eq:delta_Vats}
	\end{align}
	the chain will provide an estimate of the mean with small Monte Carlo error, when compared to the variability of the target. 
	The default choices $\alpha = 0.05$ and $\epsilon = 0.05$ were suggested in \cite{vats2018revisiting}, and were used in our work.
	For experiments reported in this paper we used \eqref{eq:delta_Vats} to select an appropriate threshold for both $\hat{R}^{\text{GR},L}$ and $\hat{R}^{\text{VK},L}$, which leads to estimated burn-in periods that we denote $\hat{b}^{\text{GR}, L}$, and $\hat{b}^{\text{VK},L}$, respectively, in the main text.

	The above discussion focussed on the univariate case, but generalisations of these convergence diagnostics are available and can be found in 
	\cite{brooks1998general}  and  \cite{vats2018revisiting}.
	All convergence diagnostics in this work were computed using the \texttt{R} packages \verb+coda+ \citep{Rcoda} and \verb+stableGR+ \citep{RstableGR}\footnote{The GR diagnostic in the software package uses the original definition in \cite{gelman1992}, that differs slightly from \eqref{eq:GR_diagnostic}; however, this difference is not expected to strongly affect the simulation results that we present.}.

	\section{Empirical Assessment: Additional Results} \label{app: add results}
	
	This section first explores the effect of the choice of kernel, then collects together additional empirical results that accompany the assessment in \Cref{sec: empirical}.

	\subsection{Choice of Kernel} \label{ap: hyper param vary}
	
	This section concerns the selection of the kernel parameter $\beta \in (-1,0)$, which enters into the kernel $k(x,y) = (c^2 + \|\Gamma^{-1/2}(x-y)\|^2)^\beta$ used within \texttt{Stein Thinning}.
	The experiment of \Cref{fig: illustration} was repeated for different values of $\beta \in \{-0.1,-0.5,-0.9\}$, and for each setting of the $\Gamma$ preconditioner matrix (\texttt{med}, \texttt{sclmed}, \texttt{smpcov}), with the results displayed in \Cref{fig:beta_effect}.
	It can be seen that the states selected when $\beta = -0.5$ are qualitatively reasonable. 
	However, poor results were observed for $\beta \in \{ -0.1, -0.9\}$.
    At $\beta = -0.1$, states were selected that were closer together compared to what would intuitively have been expected.
    At $\beta = -0.9$, stratification of selected states across the two components of the target was not achieved.
    These results reflect that $\beta \in \{-0.1,-0.9\}$ are edge cases for KSD, which has been shown to enjoy convergence control only for $\beta \in (-1,0)$ \citep[see Theorem 4 in][]{Chen2019}.
	These results support the use of $\beta = -0.5$ as a default in \texttt{Stein Thinning}.

	\begin{figure}[t!]
			\centering
			
			\begin{subfigure}[b]{0.3\textwidth}
			\includegraphics[width = \textwidth]{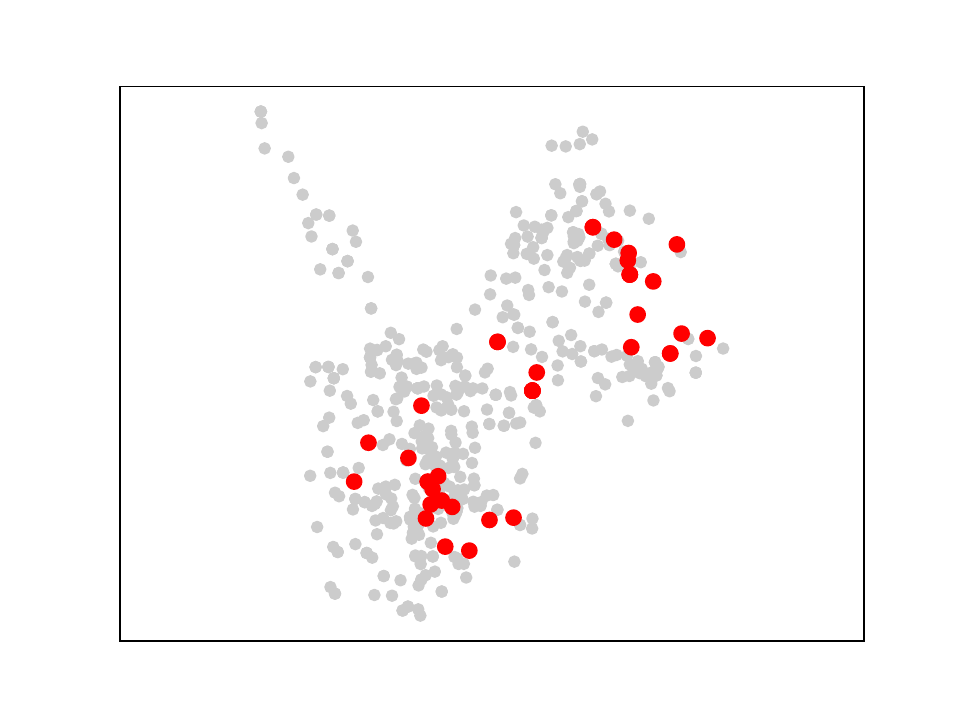}
			\caption{\texttt{med}, $\beta = -0.1$}
			\end{subfigure}
			\begin{subfigure}[b]{0.3\textwidth}
			\includegraphics[width = \textwidth]{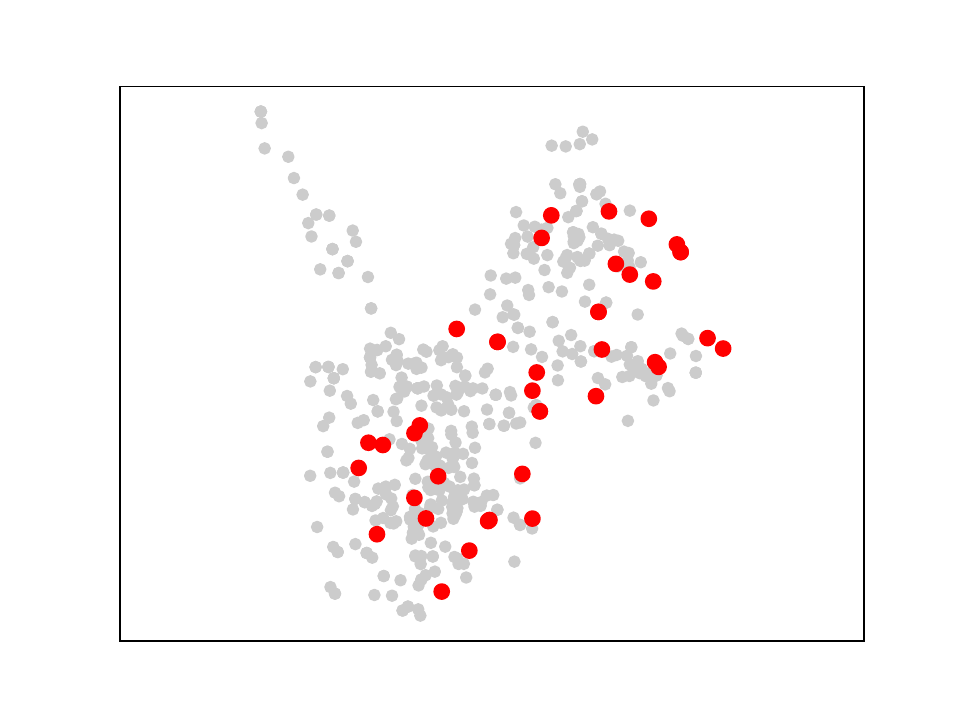}
			\caption{\texttt{med}, $\beta = -0.5$}
			\end{subfigure}
			\begin{subfigure}[b]{0.3\textwidth}
			\includegraphics[width = \textwidth]{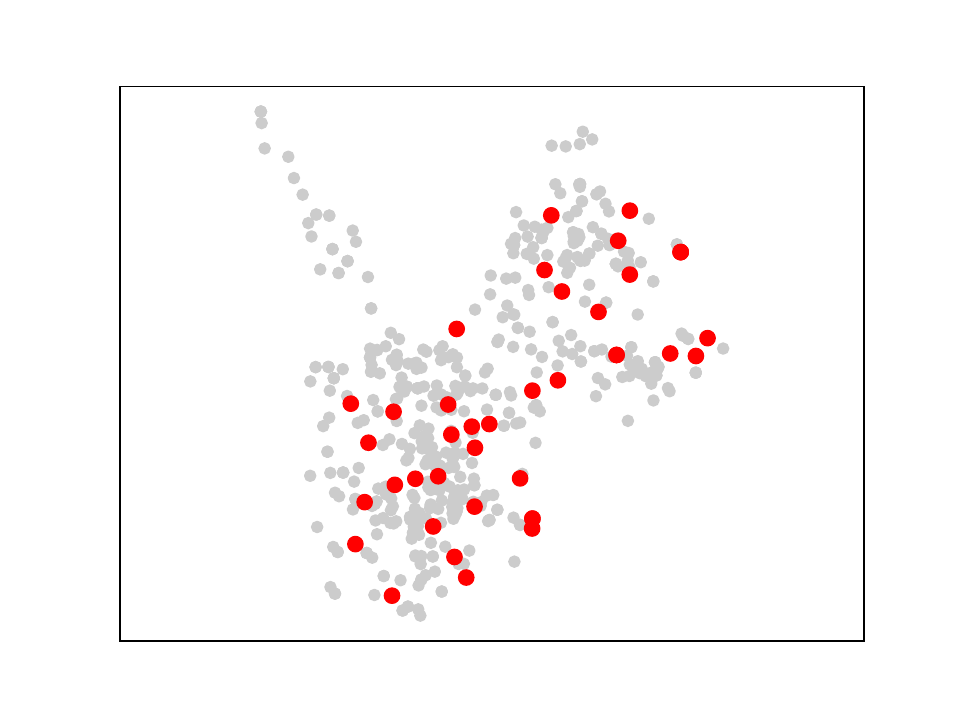}
			\caption{\texttt{med}, $\beta = -0.9$}
			\end{subfigure}
			
			\begin{subfigure}[b]{0.3\textwidth}
			\includegraphics[width = \textwidth]{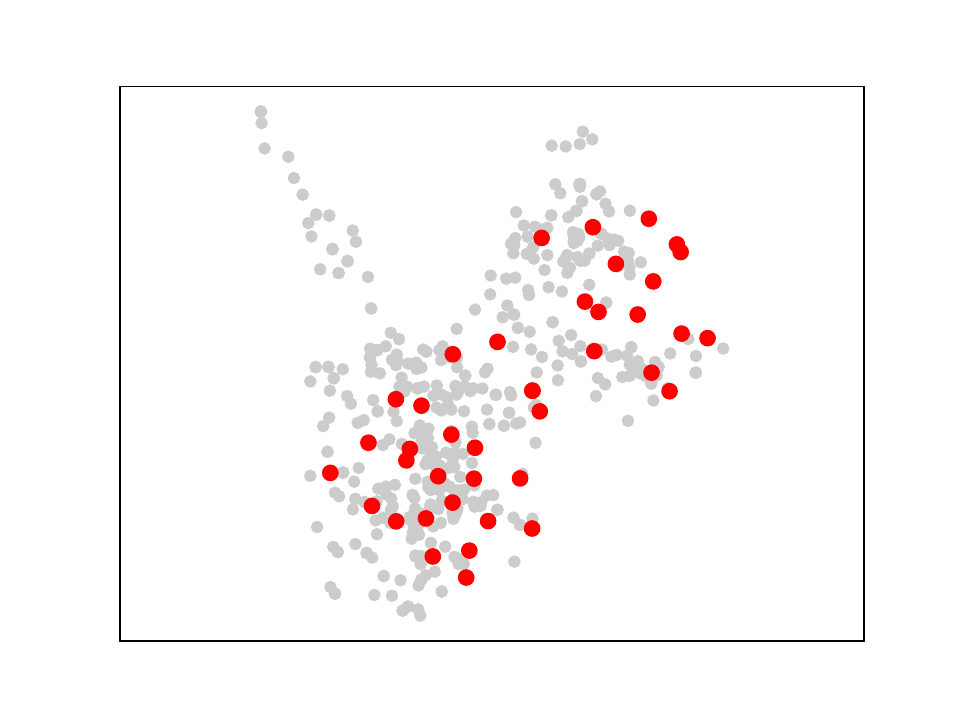}
			\caption{\texttt{sclmed}, $\beta = -0.1$}
			\end{subfigure}
			\begin{subfigure}[b]{0.3\textwidth}
			\includegraphics[width = \textwidth]{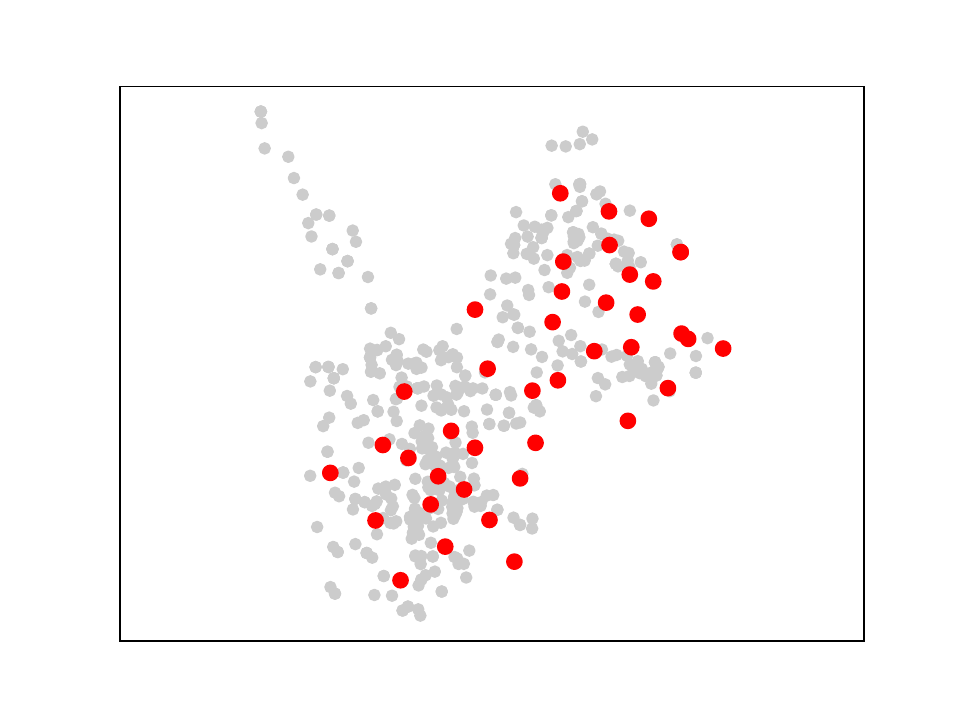}
			\caption{\texttt{sclmed}, $\beta = -0.5$}
			\end{subfigure}
			\begin{subfigure}[b]{0.3\textwidth}
			\includegraphics[width = \textwidth]{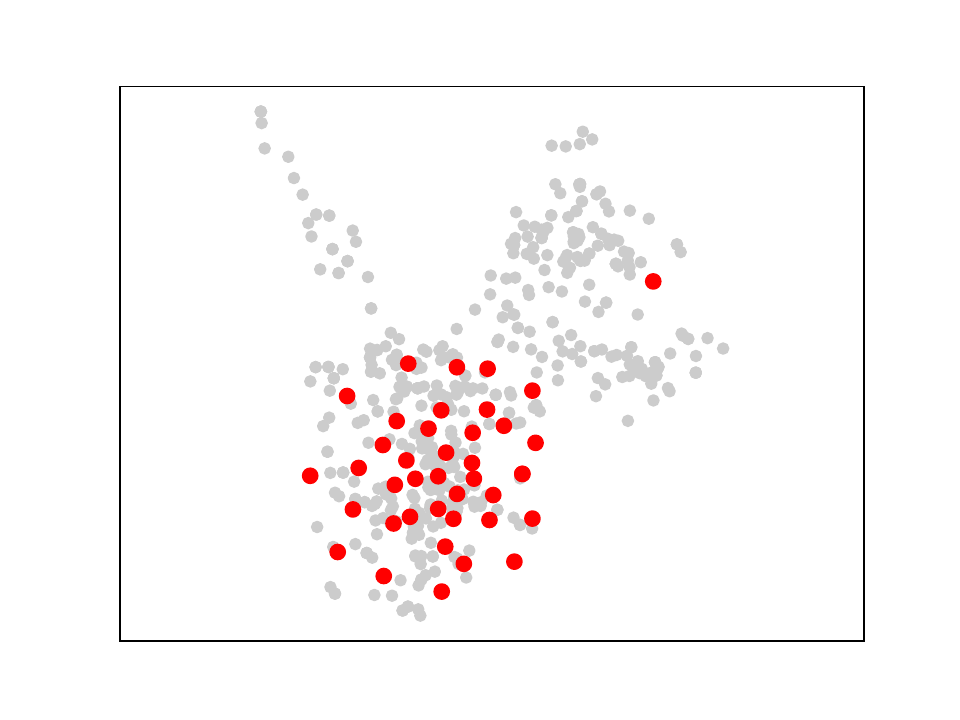}
			\caption{\texttt{sclmed}, $\beta = -0.9$}
			\end{subfigure}
			
			\begin{subfigure}[b]{0.3\textwidth}
			\includegraphics[width = \textwidth]{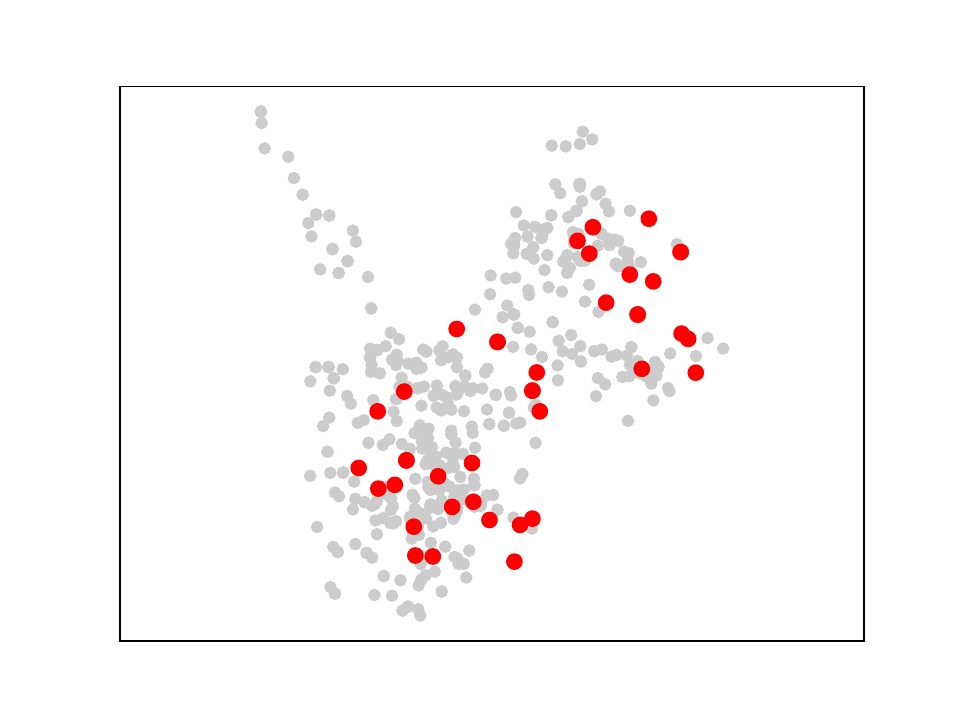}
			\caption{\texttt{smpcov}, $\beta = -0.1$}
			\end{subfigure}
			\begin{subfigure}[b]{0.3\textwidth}
			\includegraphics[width = \textwidth]{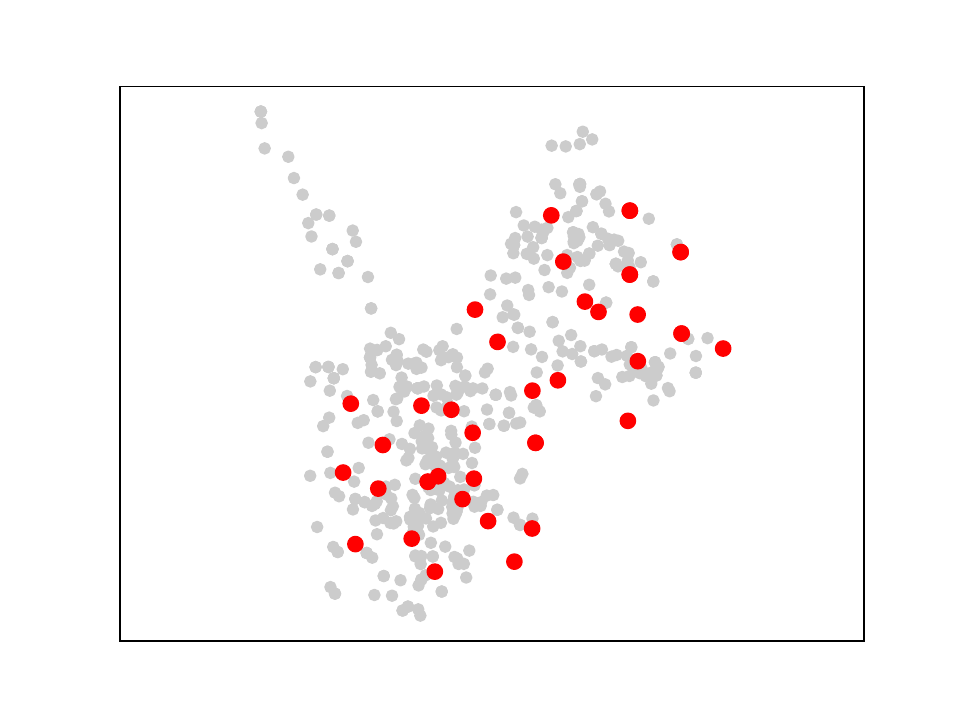}
			\caption{\texttt{smpcov}, $\beta = -0.5$}
			\end{subfigure}
			\begin{subfigure}[b]{0.3\textwidth}
			\includegraphics[width = \textwidth]{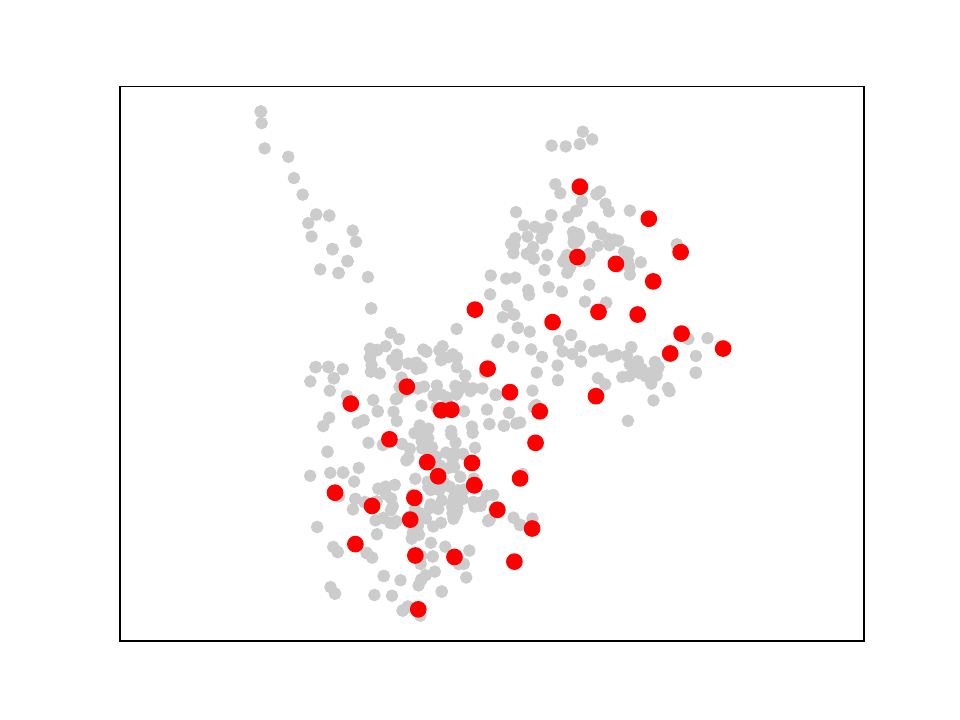}
			\caption{\texttt{smpcov}, $\beta = -0.9$}
			\end{subfigure}
			
			\caption{
				Exploring the role of $\beta$ in \texttt{Stein Thinning}.
				[The experiment in \Cref{fig: illustration} was reproduced, varying the parameter $\beta$ in the kernel. The results reported in \Cref{fig: illustration} correspond to $\beta = -0.5$.]}
			\label{fig:beta_effect}
		\end{figure}

	\subsection{Goodwin Oscillator}\label{app:addtional_Goodwin}
	
	The Goodwin oscillator is a phenomenological model for genetic regulatory processes in a cell and is described by $g$ coupled ODEs of the form
	\begin{align*}
	\frac{\mathrm{d} u_1}{\mathrm{d} t} &= \frac{a_1}{1+a_2 u_g^{\rho}} - \alpha u_1, 
	\\
	\frac{\mathrm{d} u_2}{\mathrm{d} t} &= k_1 u_1- \alpha u_2,
	\\
	& \vdots
	\\
	\frac{\mathrm{d} u_g}{\mathrm{d} t} &= k_{g-1} u_{g-1}- \alpha u_g ,
	\end{align*}
	where the first component $u_1$ represents the concentration of mRNA, $u_2$ that of its corresponding protein product, while $u_3, \ldots u_g$ represent concentrations of proteins in a signalling cascade, that can either be present ($g>2$) or absent ($g=2$), and with the $g^{\text{th}}$ protein having a negative feedback on the production of mRNA, by means the Hill curve in the first equation. 
	The nontrivial oscillations that result have led to the Goodwin oscillator being used in previous studies to assess the performance of Bayesian computational methods \citep{Calderhead2009a,Oates2016,Chen2019}.
	The parameters $a_1, k_1, \ldots k_{g-1} > 0$ represent synthesis rates and $a_2, \alpha > 0$ representing degradation rates.
	To cast this model in the setting of \Cref{sec: methods} we set $x \in \mathbb{R}^{g+2}$ to be the vector whose entries are $\log(a_1)$, $\log(k_1)$, and so forth, so that we have a $d = g+2$ dimensional parameter for which inference is performed.

	The experiment that we report considers synthetic data $y_i \in \mathbb{R}^g$ generated in the simple case $g=2$, which are then corrupted by Gaussian noise such that the terms $\phi_i$ in \eqref{eq:original_likelihood}
	are equal to 
	\begin{equation}
	\phi_i(u(t_i)) \propto \exp\left(- \frac{1}{2} (y_i - u(t_i))^\top C^{-1} (y_i - u(t_i)) \right) \label{eq:2d_likelihood}
	\end{equation}
	with $C = \text{diag}(0.1^2, 0.05^2)$.
	The initial condition was $u(0) = (0,0)$  and the data-generating parameters were $(a_1, a_2, \alpha, k_1) = (1, 3, 0.5, 1)$.
	The times $t_i$, $i = 1, \ldots, 2400$, at
	which data were obtained were taken to be uniformly spaced on [1,25], in order to capture both the oscillatory behaviour of the system and its steady state. 
	This relatively high frequency of observation and corresponding informativeness of the dataset was used to pre-empt a similarly high frequency observation process in the calcium signalling model of \Cref{subsec: cardiac}. 
	\Cref{fig:Goodwin_data} displays the dataset.
	A standard Gaussian prior $\pi(x)$ was placed on the parameter $x$ and each MCMC method was applied to approximately sample from the posterior $P$.
	
	Exemplar trace plots for the MCMC methods are presented in \Cref{fig:Goodwin_traceplots_z}.
	The over-dispersed initial states used for the $L$ chains are reported in \Cref{table:initial-points-Goodwin}, while the univariate and multivariate convergence diagnostics, computed  every 1000 iterations, are shown respectively in \Cref{fig:Goodwin_conv_diagnostics_z} and \Cref{fig:Goodwin_conv_diagnostics_z_multi}.
	The values of the thresholds $\delta(L,\alpha, \epsilon)$ are reported in \Cref{table:delta-Lotka}.
	For each MCMC method, the estimated burn-in period is presented in Table \ref{table:burnin-Goodwin}. 
	The GR diagnostic did not fall below the $1 + \delta$ threshold in the allowed number of iterations, which is consistent with the empirical observations of \cite{vats2018revisiting}. 
	
	The additional results for the Goodwin oscillator that we present in this appendix are as follows: 
	
	\begin{itemize}
		\item Figures \ref{fig:Goodwin_SP_RW} (\texttt{RW}), \ref{fig:Goodwin_SP_MALA} (\texttt{MALA}) and \ref{fig:Goodwin_SP_MALA_PRECOND} (\texttt{P-MALA}) display point sets of size $m=20$ selected using traditional burn in and thinning methods, \texttt{Support Points} and \texttt{Stein Thinning}, based on MCMC output. 
		Note that the gray regions are not \emph{necessarily} regions of high posterior probability; they are the regions explored by the sample path and, moreover, these panels are two-dimensional projections from $\mathbb{R}^4$.
	    Therefore we are hesitant to draw strong conclusions from these figures.
	    \item Figures \ref{fig:Goodwin_moments_RW} (\texttt{RW}), \ref{fig:Goodwin_moments_MALA} (\texttt{MALA}) and \ref{fig:Goodwin_moments_PRECOND-MALA} (\texttt{P-MALA}) display the absolute error in estimating the first moment of each parameter, for each of the competing methods, where an extended run from MCMC provided the ground truth.
		\item Figures \ref{fig:Goodwin_densities_RW} (\texttt{RW}), \ref{fig:Goodwin_densities_ADA-RW} (\texttt{ADA-RW}),  \ref{fig:Goodwin_densities_MALA} (\texttt{MALA}) and \ref{fig:Goodwin_densities_PRECOND-MALA} (\texttt{P-MALA}) show marginal density estimates, for each parameter and each of the competing methods, where an extended run from MCMC provided the ground truth. 
	\end{itemize}

	\begin{figure}[t!]
		\centering
		\includegraphics[width = 0.45\textwidth]{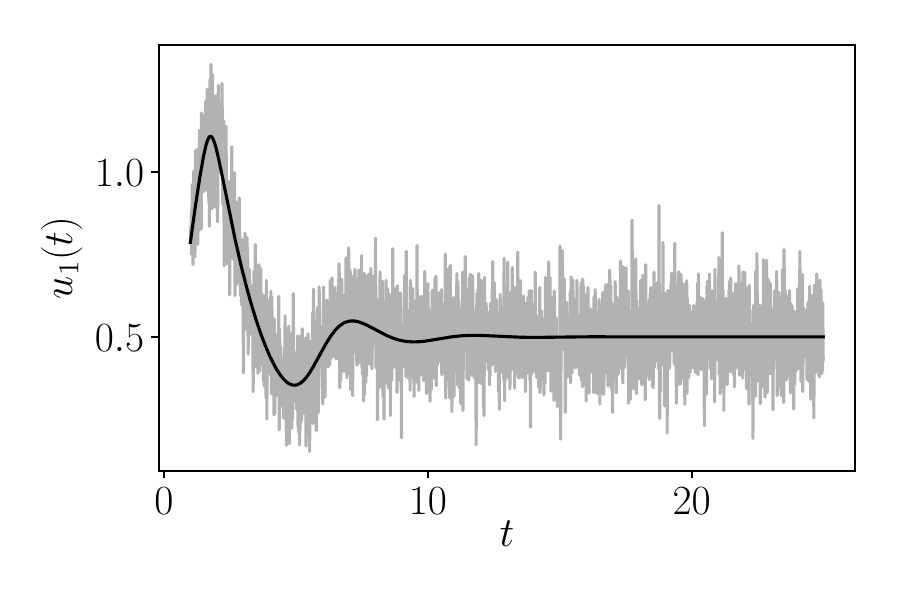}
		\includegraphics[width = 0.45\textwidth]{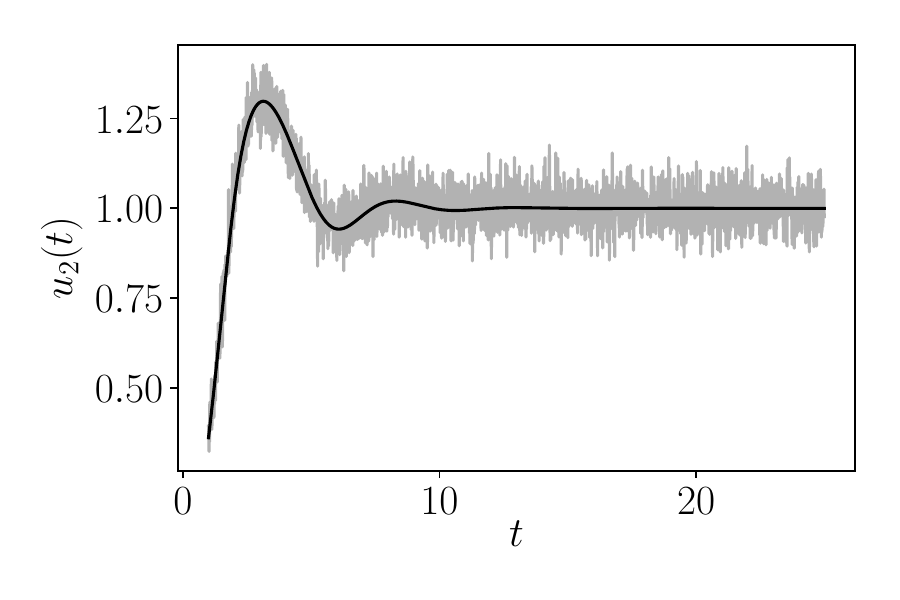}
		\caption{Data (gray) and  ODE solution corresponding to the true data-generating parameters (black) for the Goodwin oscillator.}
		\label{fig:Goodwin_data}
	\end{figure}

	\begin{table}[t!]
		\centering
		\footnotesize
		\begin{tabular}{|c||c | }
			\hline
			\textbf{Chain Number} &  \textbf{Initial State for Parameters}  $(a_1, a_2, \alpha, k)$ \\
			\hline
			\hline
			1 & (0.5, 1, 3, 2)  \\
			2 & (0.001, 0.2, 0.1, 10)   \\
			3 & (10, 0.1, 0.9, 0.1) \\
			4 & (0.1, 30, 0.1, 0.3)     \\
			5 & (2, 2, 2, 2)    \\
			6 & (5, 5, 1, 1)    \\
			\hline 
			
		\end{tabular}
		\caption{Initial states, over-dispersed with respect to the posterior, for the $L = 6$ independent Markov chains used in the Goodwin oscillator.  
		The parameters used to generate the data were $ (a_1, a_2, \alpha, k) = (1,	3,	0.5,	1)$. }
		\label{table:initial-points-Goodwin}
	\end{table}

	\begin{figure}[ht!]
		\centering
		\includegraphics[width = 0.9\textwidth]{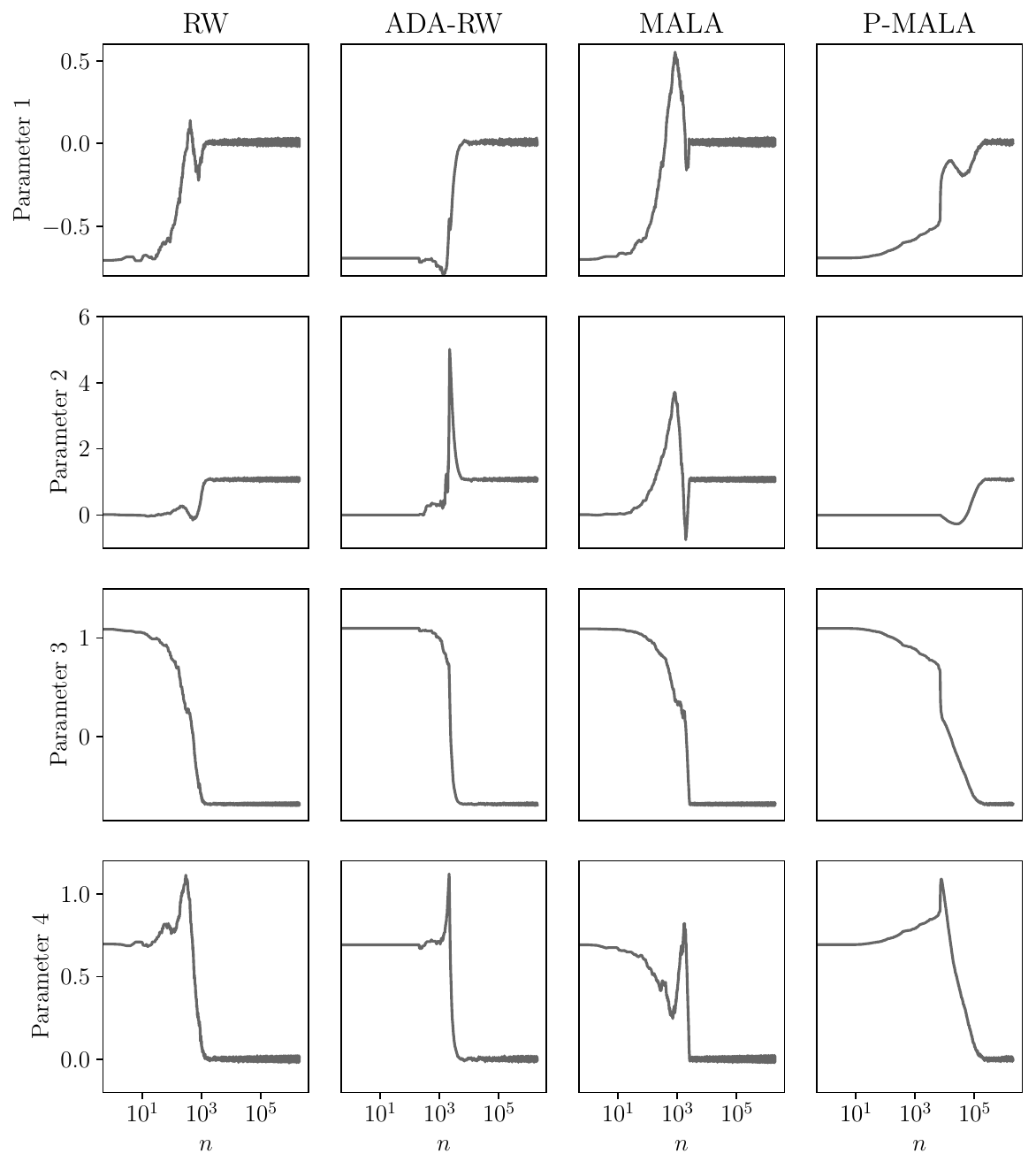}
		\caption{Trace plots for each parameter in the Goodwin oscillator, plotted against the MCMC iteration number. 
			Each row corresponds to one of the four parameters, while each column corresponds to one of the four MCMC methods considered.
			(Note the logarithmic scale on the horizontal axis, used to better visualise the initial part of the MCMC sample path.)
		}
		\label{fig:Goodwin_traceplots_z}
	\end{figure}

	\begin{table}[t!]
		\centering
		\footnotesize
		\begin{tabular}{|c||c|c| }
			\hline
			& \textbf{Univariate Diagnostics}& \textbf{Multivariate Diagnostics}
			\\
			\hline
			\hline
			$L = 6$ &  4.88 $\times 10^{-4}$ & 3.56 $\times 10^{-4}$ 
			\\
			\hline
			$L = 5$ & 4.07 $\times 10^{-4}$ &  2.96 $\times 10^{-4}$\\
			\hline
			$L = 1$& 8.13 $\times 10^{-5}$&  5.93 $\times 10^{-5}$\\
			\hline
		\end{tabular}
		\caption{The values of the threshold $\delta(L, \alpha, \epsilon)$ used in analysis of the Goodwin and Lotka--Volterra models, with $\alpha = 0.05$, $\epsilon = 0.05$, when changing $L$, and considering the univariate and multivariate convergence diagnostics $\hat{R}^{\text{VK},L}$.}
		\label{table:delta-Lotka}
	\end{table}

	\begin{figure}[t!]
		\centering
		\includegraphics[width = 1\textwidth]{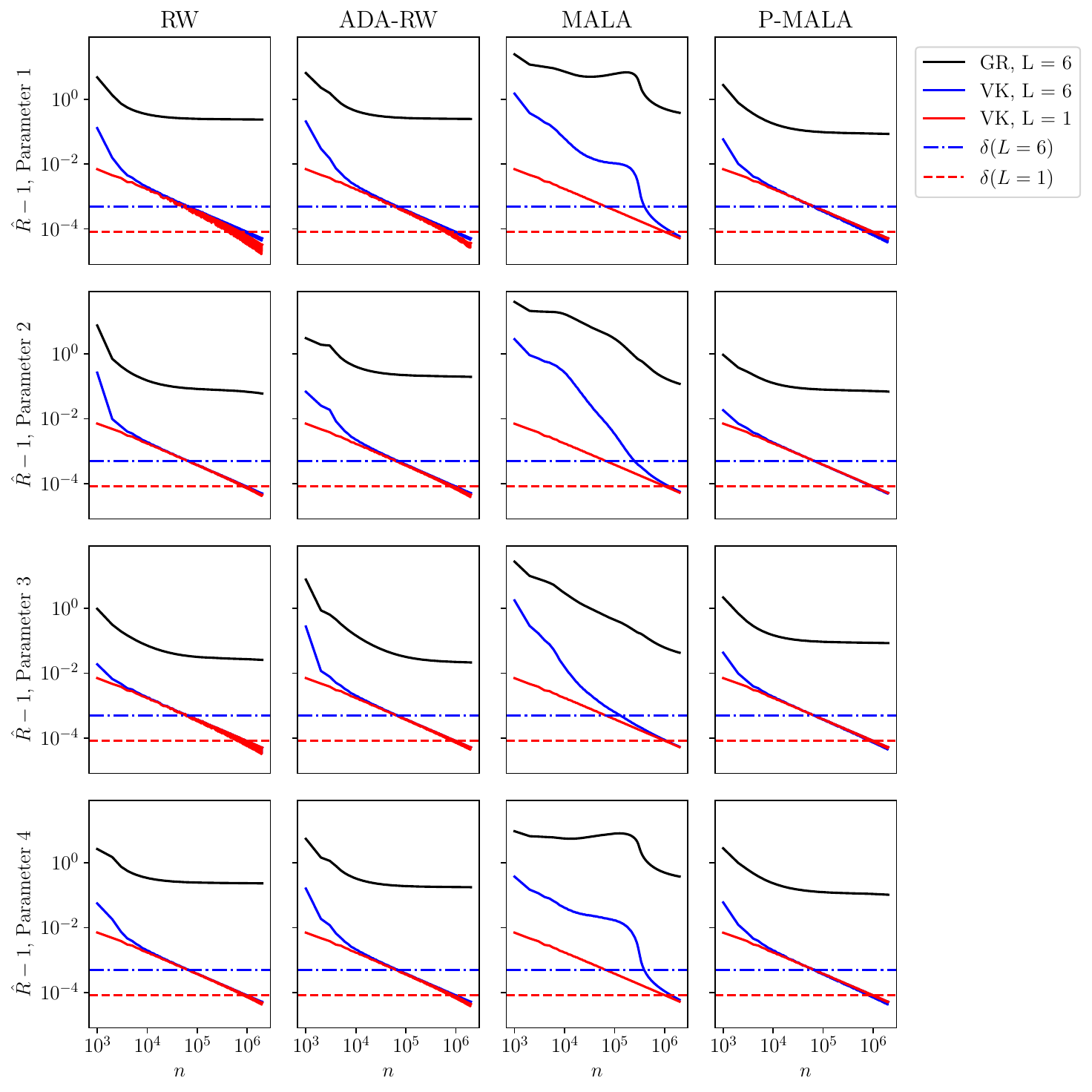}
	
		\caption{Univariate convergence diagnostics, for the Goodwin oscillator, plotted against the MCMC iteration number. 
			The black line represents the GR diagnostic (based on $L=6$ chains), while the blue and red lines represent the VK diagnostic (based on $L = 6$ and $L=1$ chains, respectively). 
			The dash-dotted ($L=6$) and dashed ($L=1$) horizontal lines correspond to the critical values $\delta(L, \alpha, \epsilon)$, used to determine the burn-in period; see Table~\ref{table:delta-Lotka}.}
		\label{fig:Goodwin_conv_diagnostics_z}
	\end{figure}

	\begin{figure}[ht!]
		\centering
		\includegraphics[width = 1\textwidth]{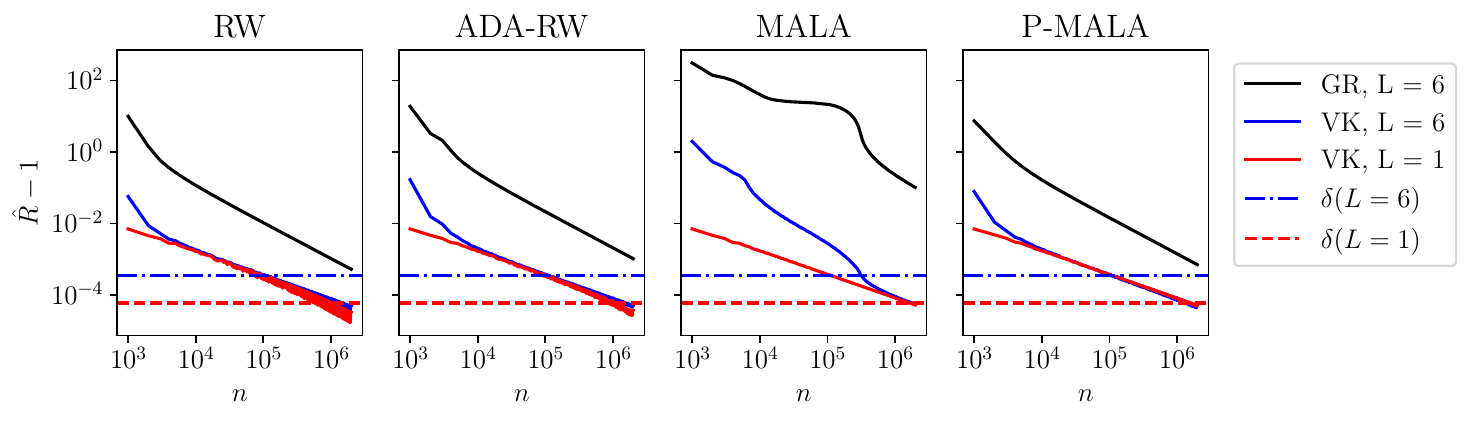}
		\caption{Multivariate convergence diagnostics for the Goodwin oscillator, plotted against the MCMC iteration number $n$. 
			The black line is the GR diagnostic (based on $L=6$ chains), while the blue and red lines are the VK diagnostic (based on $L = 6$ and $L=1$ chains, respectively). 
			The dotted ($L=6$) and dashed ($L=1$) horizontal lines correspond to the threshold $\delta$ that is used to determine the burn-in period; see Table~\ref{table:delta-Lotka} in \Cref{app:addtional_Goodwin}.}
		\label{fig:Goodwin_conv_diagnostics_z_multi}
	\end{figure}

	\begin{table}[t!]
		\centering
		\footnotesize
		\begin{tabular}{|c||c||c | c | c|}
			\hline
			\textbf{MCMC Diagnostics} & \textbf{Sampler} &  $\hat{b}^{\text{GR}, 6}$ & $\hat{b}^{\text{VK}, 6}$ & $\hat{b}^{\text{VK}, 1}$   \\
			\hline 
			\hline 
				\multirow{4}{*}{\textbf{Univariate}} & \texttt{RW} &   $>n$ &  70,000& 820,000\\
			\cline{2-5}
			& \texttt{ADA-RW} & $>n$& 71,000& 816,000\\	
		    \cline{2-5}
			& \texttt{MALA} 	 &  $>n$&  397,000& 1,020,000 \\
			\cline{2-5}
			& \texttt{P-MALA} 	&   $>n$&  68,000 & 987,000\\
			\hline \hline
			\multirow{4}{*}{\textbf{Multivariate}} & \texttt{RW} & $>n$  & 93,000 &  578,000\\
			\cline{2-5}	
			& \texttt{ADA-RW} &  $>n$  &107,000 & 824,000\\
			\cline{2-5} 
			& \texttt{MALA} 	 &  $>n$ &  316,000 & 1,615,000\\
			\cline{2-5} 
			& \texttt{P-MALA} 	&  $>n$  &103,000 &  1,475,000 \\
			\hline
		\end{tabular}
		\caption{Estimated burn-in period for the Goodwin oscillator, using the GR diagnostic based on $L$ chains, $\hat{b}^{\text{GR},L}$ ($L=6$), and the VK diagnostic based on $L$ chains,  $\hat{b}^{\text{VK},L}$, ($L = 1,6$). 		In each case both univariate and multivariate convergence diagnostics are presented; in the univariate case we report the largest value obtained when looking at each of the $d$ parameters individually to estimate the burn-in period. 			The symbol ``$>n$'' indicates the case in which a diagnostic did not go below the $1+\delta$ threshold. 			}
		\label{table:burnin-Goodwin}
	\end{table}

	\begin{figure}[t!]
		\centering
		\includegraphics[width = 0.4\textwidth]{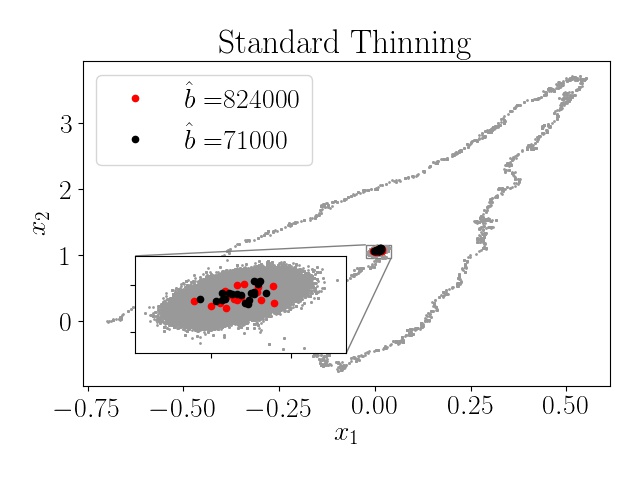}		
		\includegraphics[width = 0.4\textwidth]{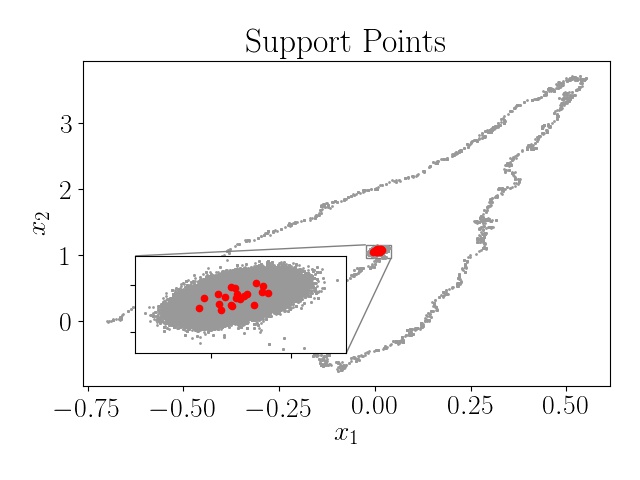}	
		\includegraphics[width = 0.4\textwidth]{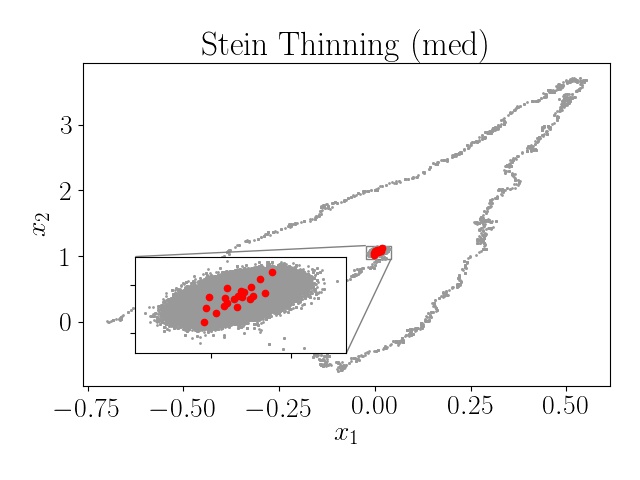}
		\includegraphics[width = 0.4\textwidth]{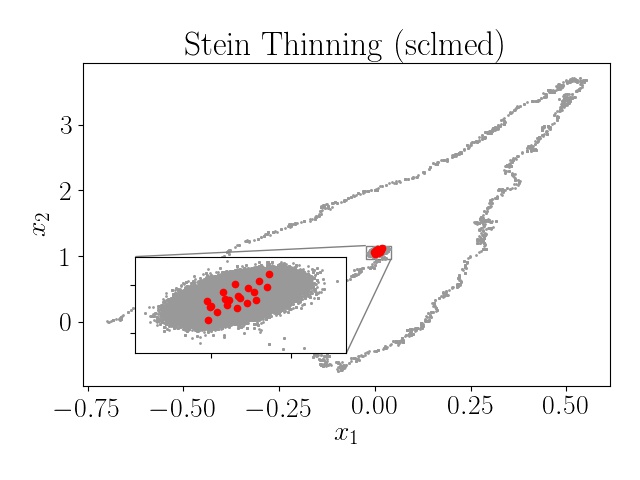}		
		\includegraphics[width = 0.4\textwidth]{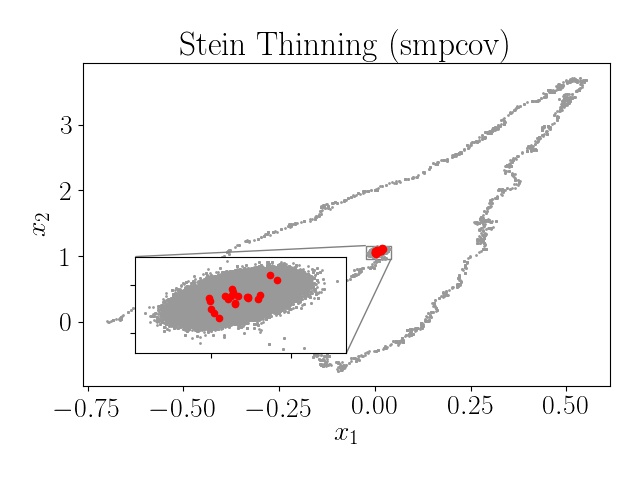}		
		
		\caption{Projections on the first two coordinates of the \texttt{ADA-RW} MCMC output for the Goodwin oscillator (grey dots), together with the first $m=20$ points selected through: traditional burn-in and thinning (the amount of burn in is indicated in the legend);  the \texttt{Support Points} method; \texttt{Stein Thinning}, for each of the settings \texttt{med}, \texttt{sclmed}, \texttt{smpcov}.
		}
		\label{fig:Goodwin_SP_ADA-RW}
	\end{figure}

	\begin{figure}[t!]
		\centering
		\includegraphics[width = 0.4\textwidth]{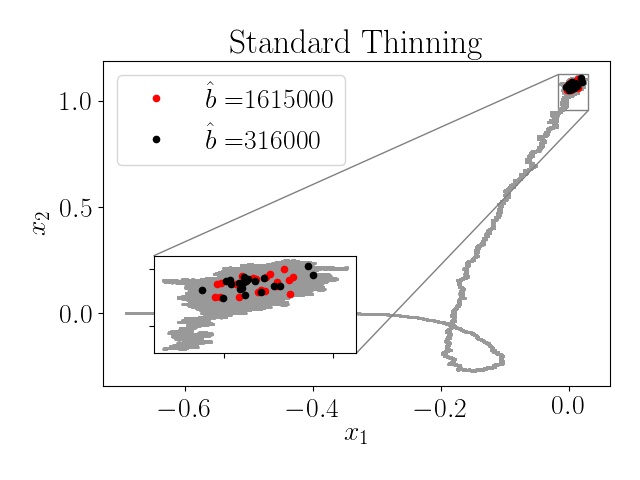}		
		\includegraphics[width = 0.4\textwidth]{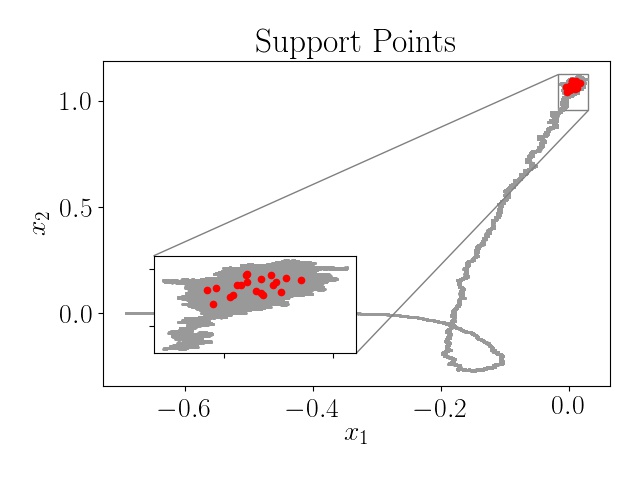}	
		\includegraphics[width = 0.4\textwidth]{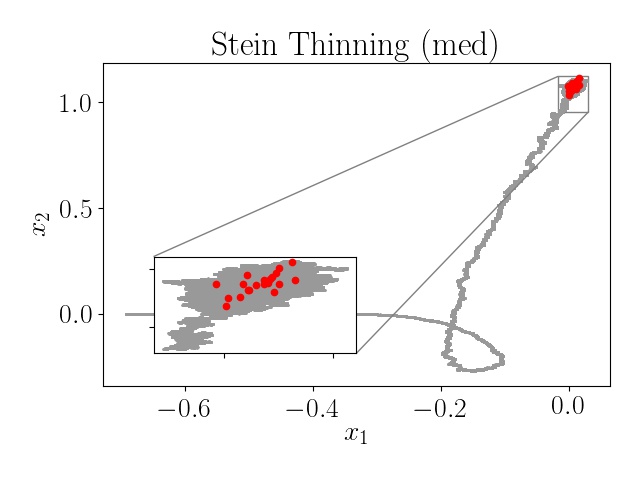}
		\includegraphics[width = 0.4\textwidth]{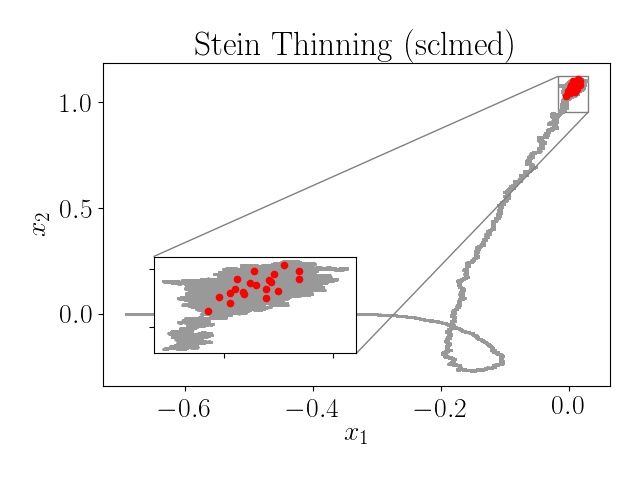}		
		\includegraphics[width = 0.4\textwidth]{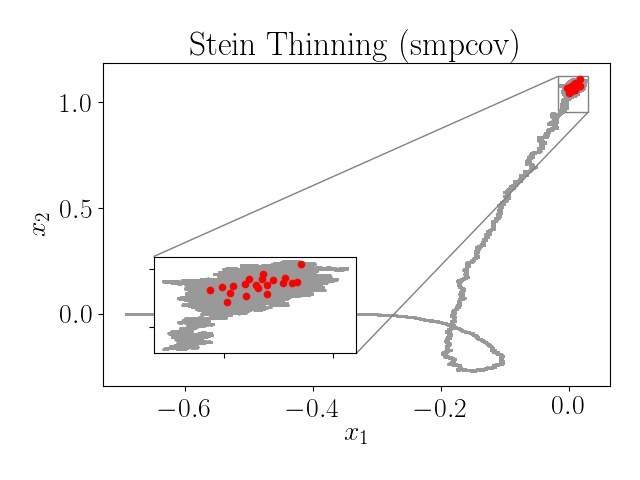}
		
		\caption{Projections on the first two coordinates of the \texttt{MALA} MCMC output for the Goodwin oscillator (grey dots), together with the first $m=20$ points selected through: traditional burn-in and thinning (the amount of burn in is indicated in the legend);  the \texttt{Support Points} method; \texttt{Stein Thinning}, for each of the settings \texttt{med}, \texttt{sclmed}, \texttt{smpcov}.
		}		\label{fig:Goodwin_SP_MALA}
	\end{figure}

	\begin{figure}[t!]
		\centering
		\includegraphics[width = 0.4\textwidth]{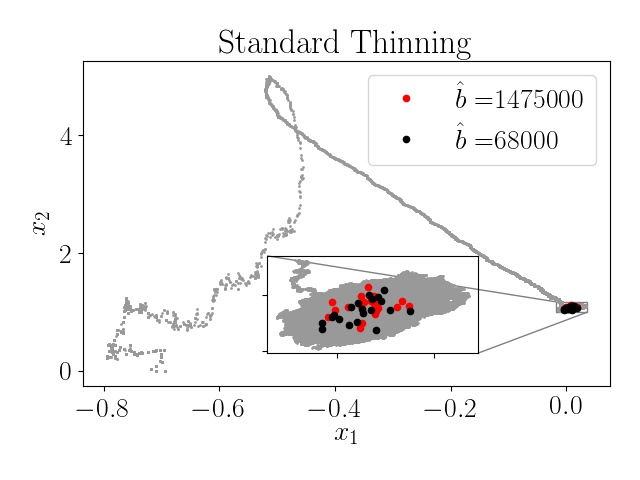}		
		\includegraphics[width = 0.4\textwidth]{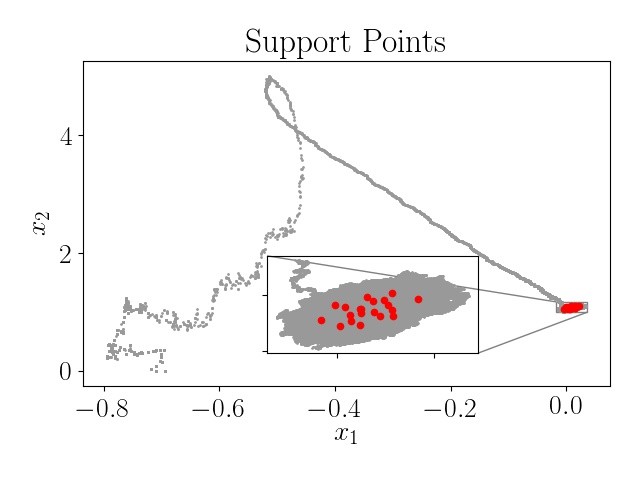}	
		\includegraphics[width = 0.4\textwidth]{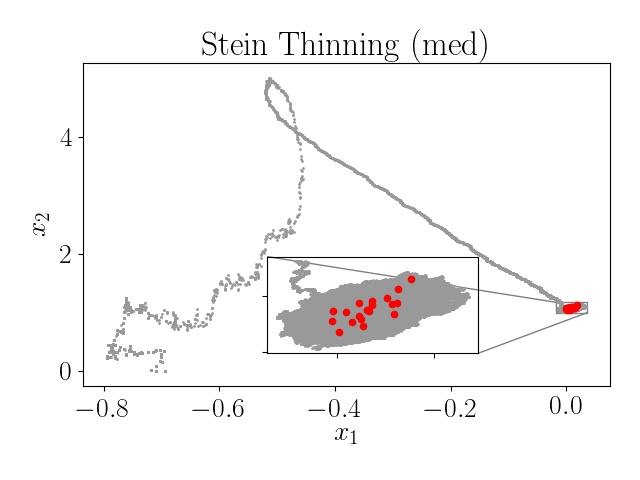}
		\includegraphics[width = 0.4\textwidth]{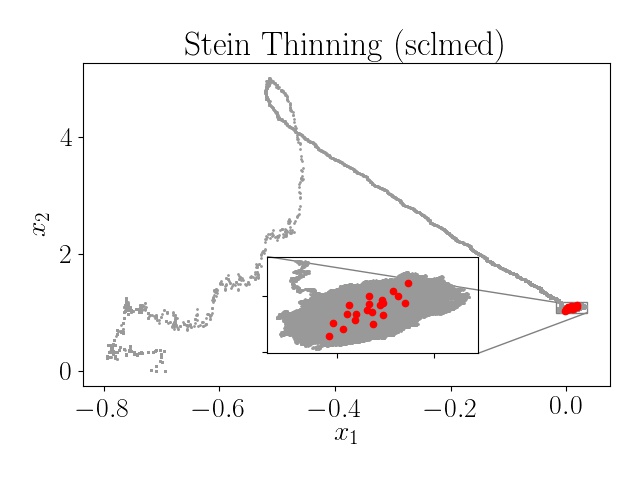}		
		\includegraphics[width = 0.4\textwidth]{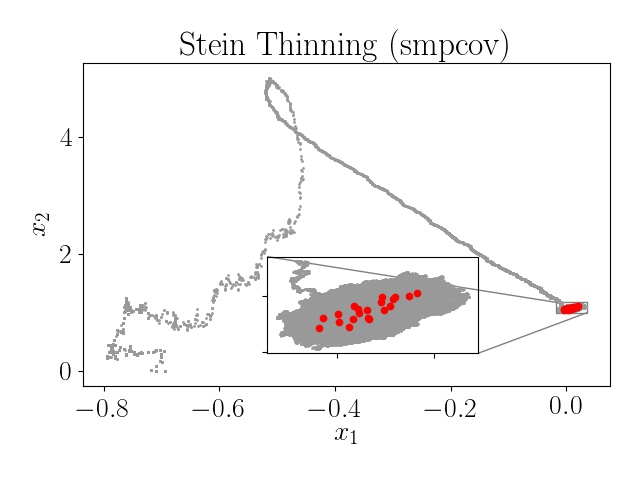}

		\caption{Projections on the first two coordinates of the \texttt{P-MALA} MCMC output for the Goodwin oscillator (grey dots), together with the first $m=20$ points selected through: traditional burn-in and thinning (the amount of burn in is indicated in the legend);  the \texttt{Support Points} method; \texttt{Stein Thinning}, for each of the settings \texttt{med}, \texttt{sclmed}, \texttt{smpcov}.
		}
	\label{fig:Goodwin_SP_MALA_PRECOND}
	\end{figure}

	\begin{figure}[t!]
		\centering
		\includegraphics[width = 1
		\textwidth]{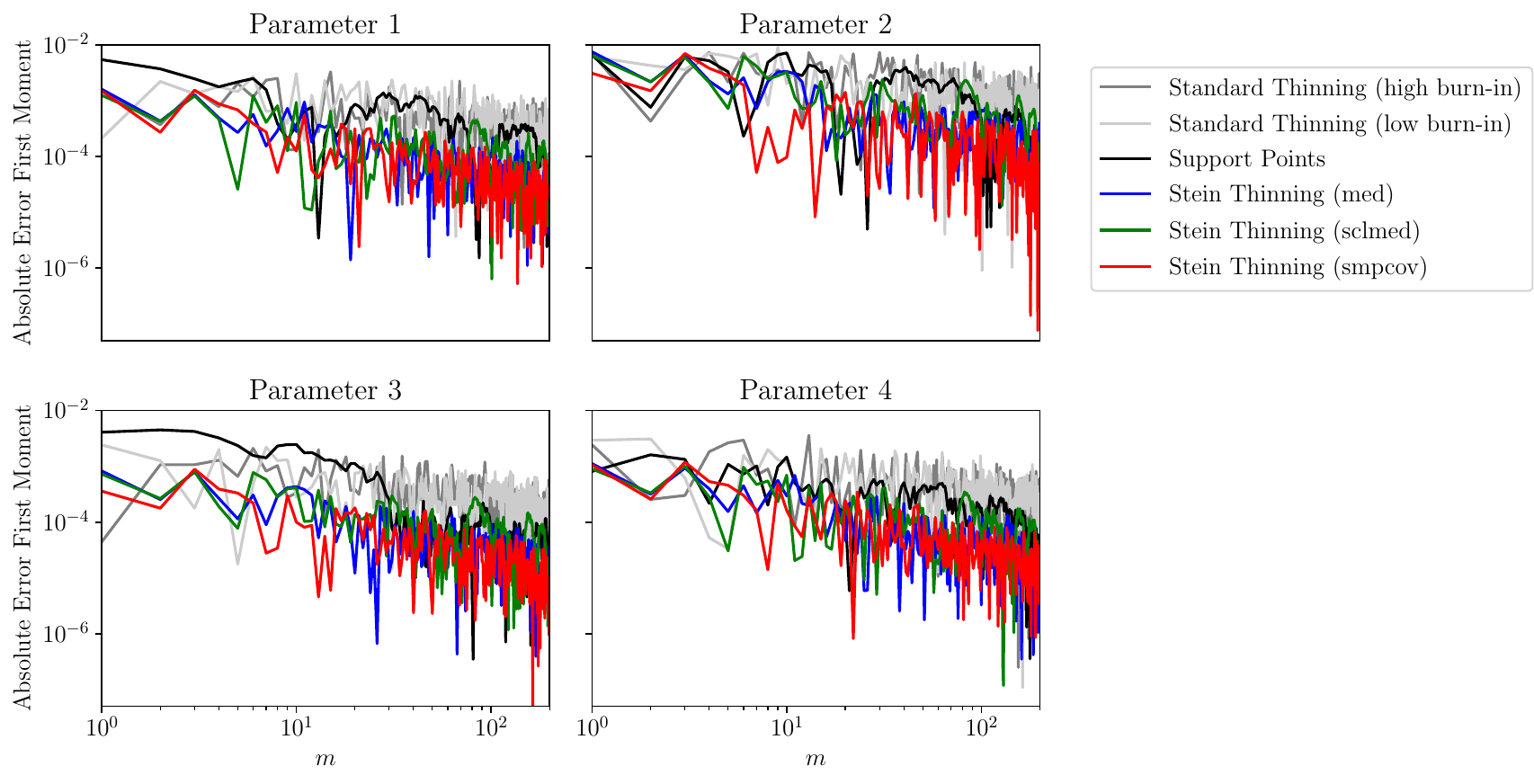}
		\caption{Absolute error of estimates for the posterior mean of each parameter in the Goodwin oscillator, based on output from \texttt{ADA-RW} MCMC.
		}
		\label{fig:Goodwin_moments_ADA_RW}
	\end{figure}

		\begin{figure}[t!]
		\centering
		\includegraphics[width = 1
		\textwidth]{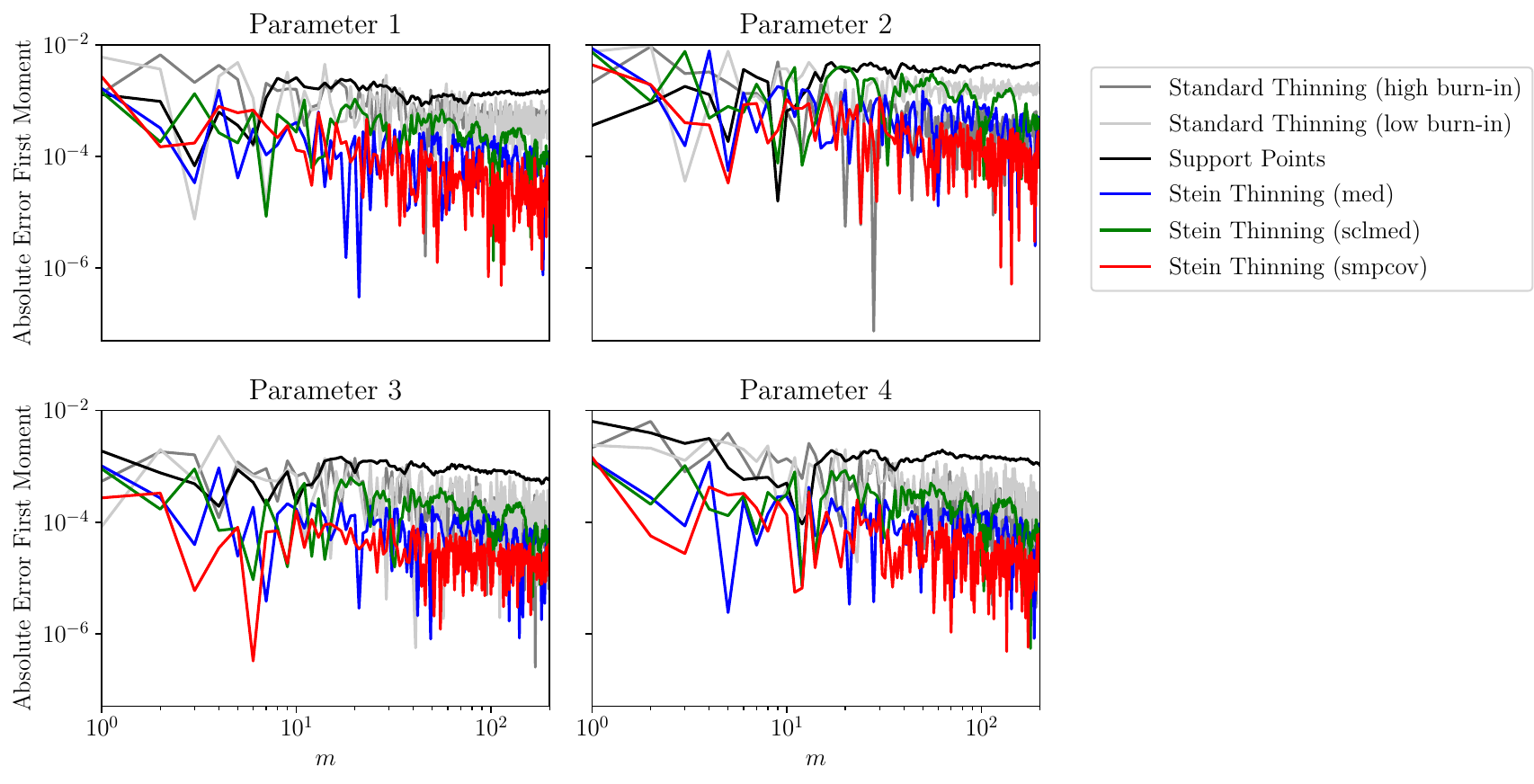}
		\caption{Absolute error of estimates for the posterior mean of each parameter in the Goodwin oscillator, based on output from \texttt{MALA} MCMC.
		}
		\label{fig:Goodwin_moments_MALA}
	\end{figure}

		\begin{figure}[t!]
		\centering
		\includegraphics[width = 1
		\textwidth]{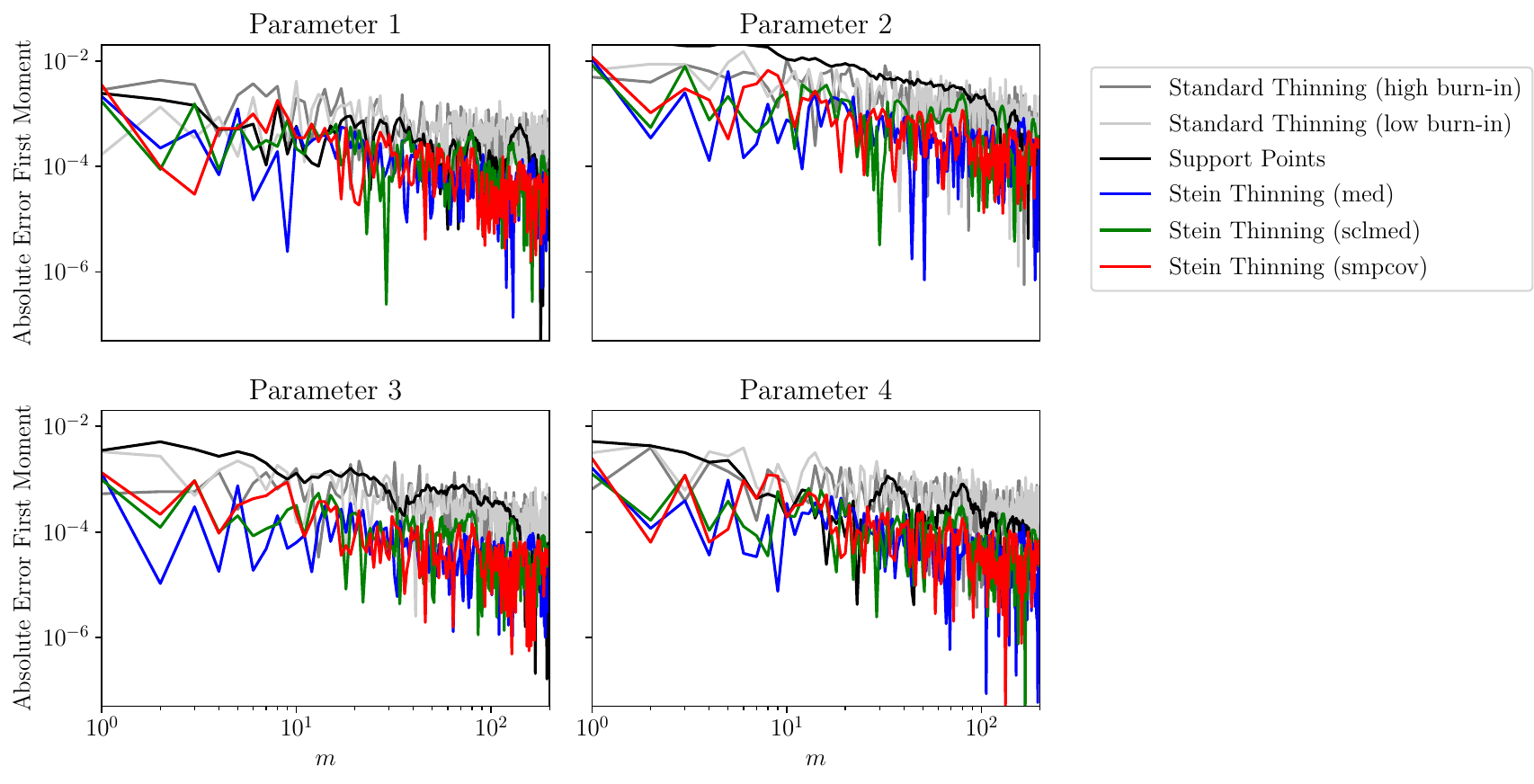}
		\caption{Absolute error of estimates for the posterior mean of each parameter in the Goodwin oscillator, based on output from \texttt{P-MALA} MCMC.
		}
		\label{fig:Goodwin_moments_PRECOND-MALA}
	\end{figure}

	\begin{figure}[t!]
		\centering
		\includegraphics[width =1
		\textwidth]{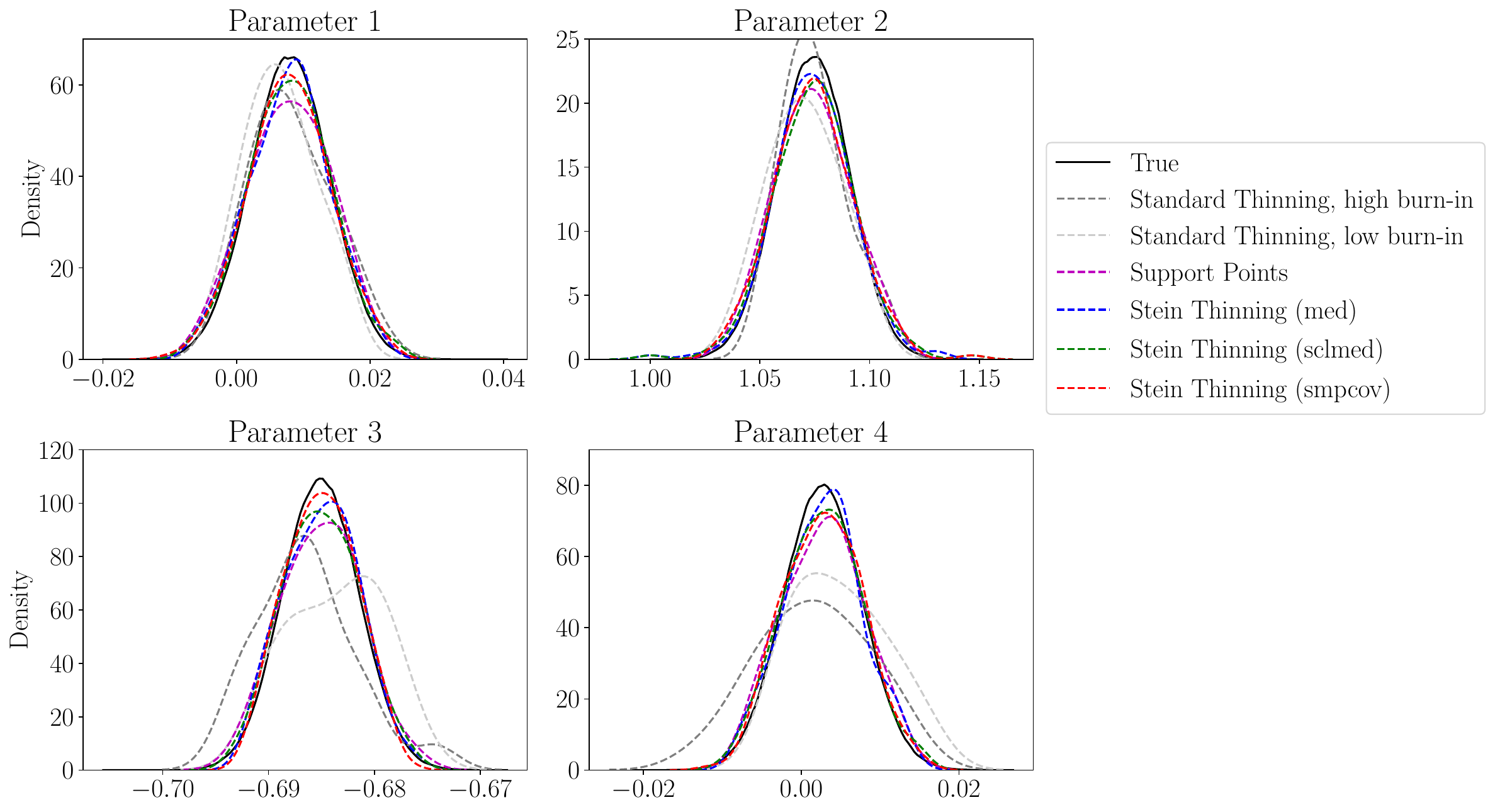}
		\caption{True and estimated marginal densities of the parameters in the Goodwin oscillator, using $m=20$ points selected from \texttt{RW} MCMC output. 
			}
		\label{fig:Goodwin_densities_RW}
	\end{figure}

	\begin{figure}[t!]
		\centering
		\includegraphics[width =1
		\textwidth]{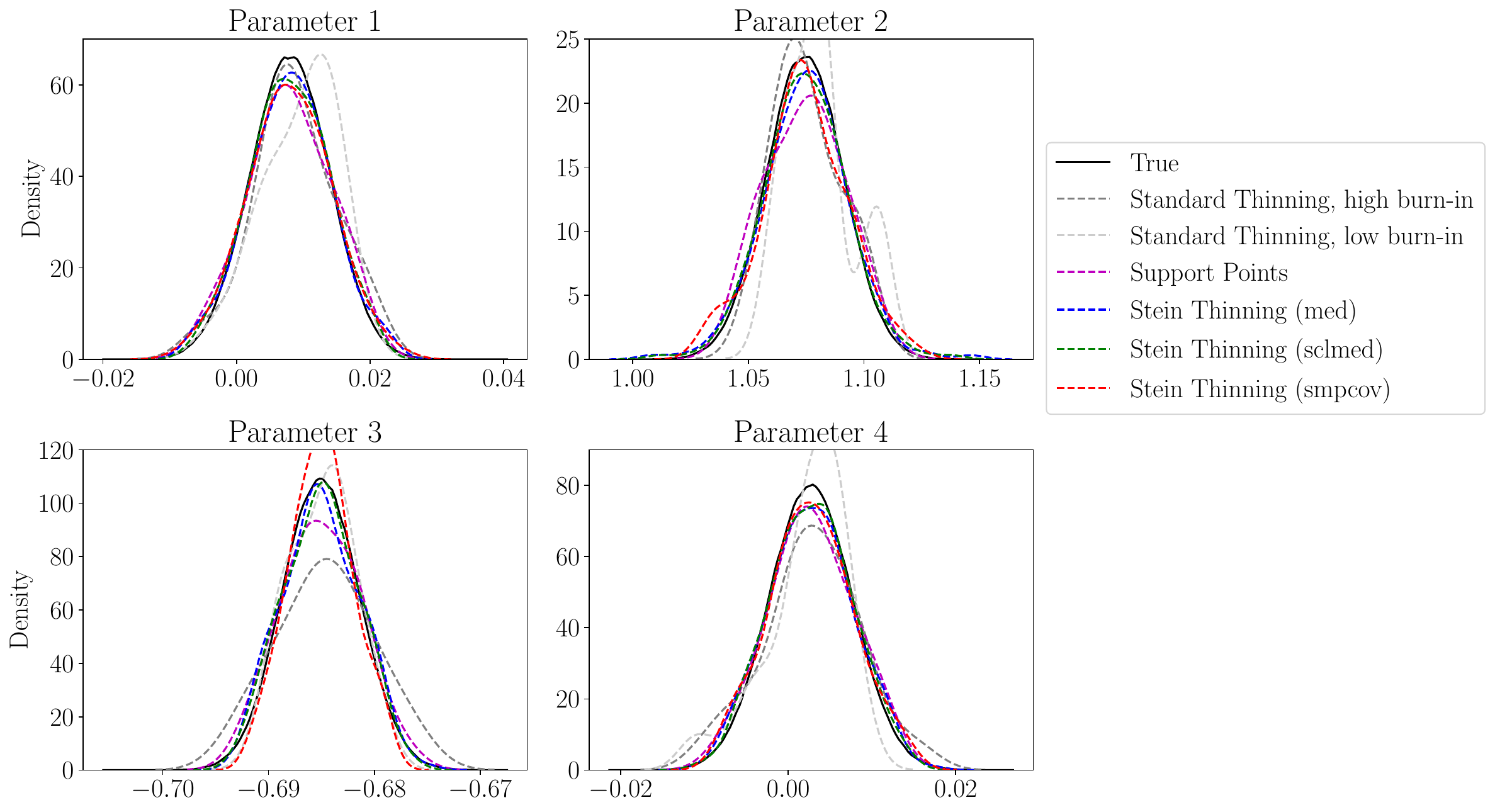}
		
		\caption{True and estimated marginal densities of the parameters in the Goodwin oscillator, using $m=20$ points, selected from \texttt{ADA-RW} MCMC output. 
			}
		\label{fig:Goodwin_densities_ADA-RW}
	\end{figure}

	\begin{figure}[t!]
		\centering
		\includegraphics[width =1
		\textwidth]{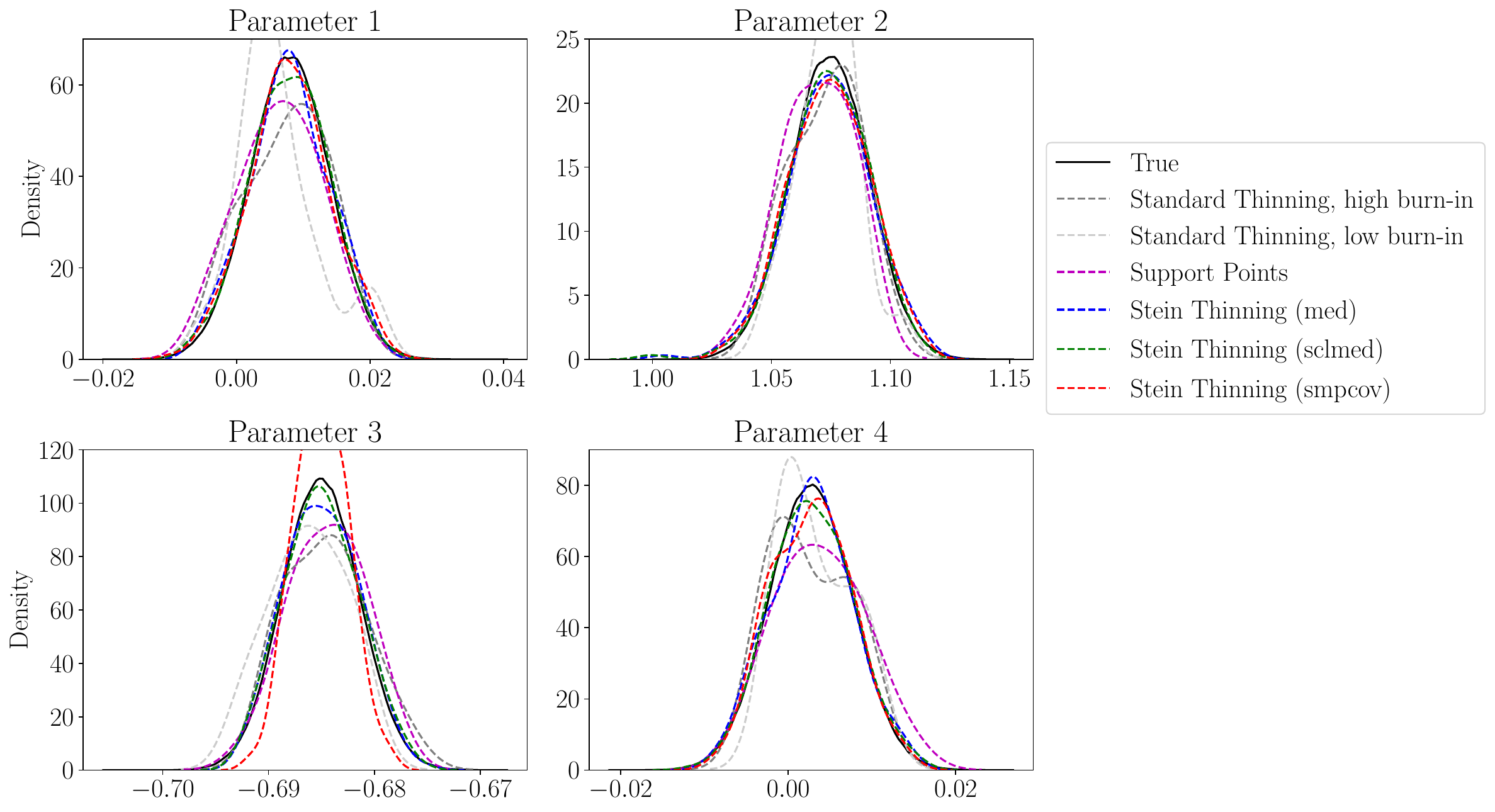}
		
		\caption{True and estimated marginal densities of the parameters in the Goodwin oscillator, using $m=20$ points, selected from \texttt{MALA} MCMC output. 
			}
		\label{fig:Goodwin_densities_MALA}
	\end{figure}

	\begin{figure}[t!]
		\centering
		\includegraphics[width =1
		\textwidth]{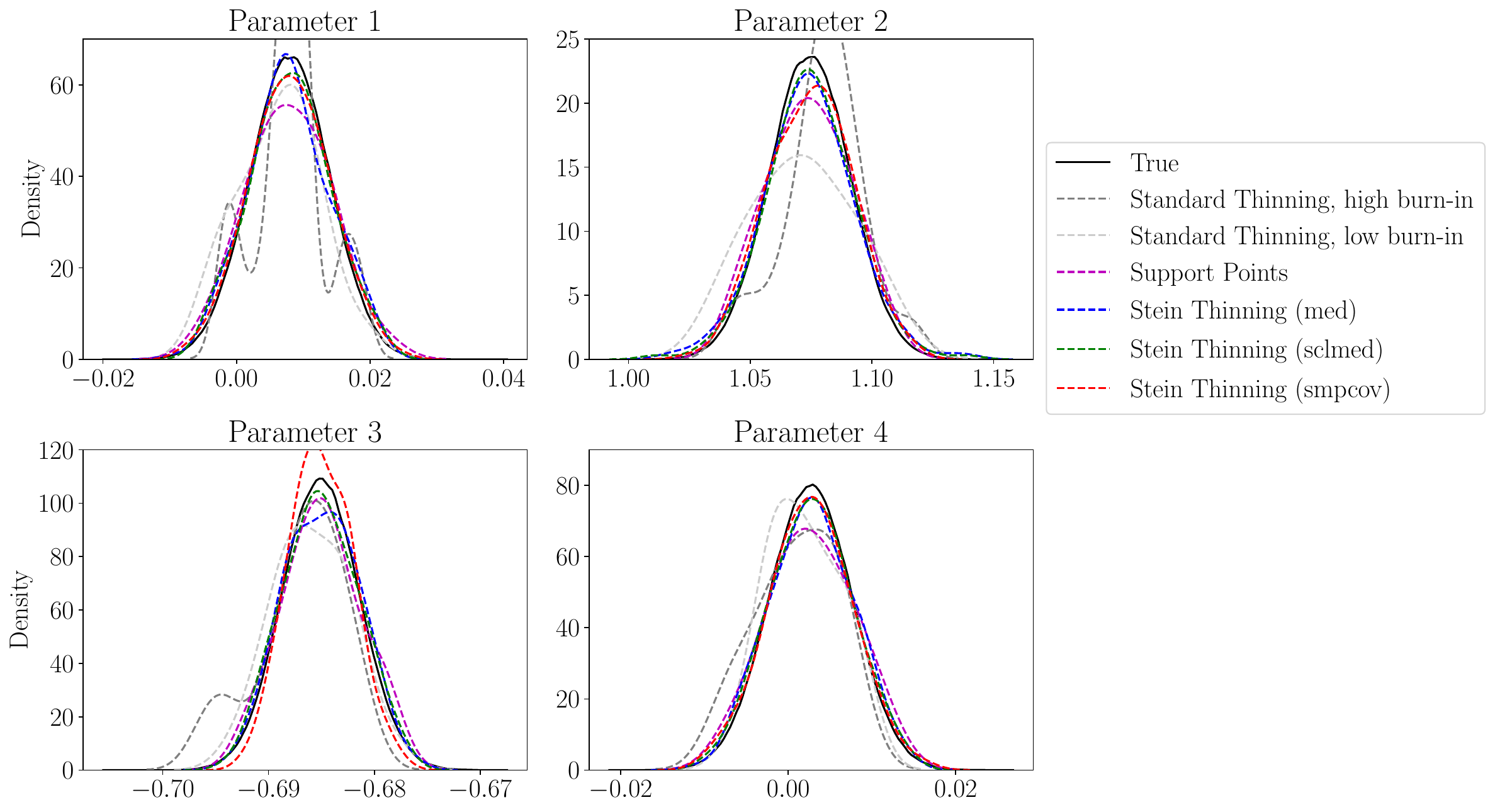}

		\caption{True and estimated marginal densities of the parameters in the Goodwin oscillator, using $m=20$ points, selected from \texttt{P-MALA} MCMC output. 
			}
		\label{fig:Goodwin_densities_PRECOND-MALA}
	\end{figure}

	\FloatBarrier

	\subsection{Lotka--Volterra}\label{app:addtional_Lotka}
	
	The Lotka--Volterra model describes the oscillatory evolution of prey ($u_1$) and predator ($u_2$) species in a closed environment. 
	The prey has an intrinsic mechanism for growth proportional to its abundance, described by a parameter $\theta_1 > 0$, whilst interaction with the predator leads to a decrease in the prey population at a rate described by a parameter $\theta_2 > 0$.
	Conversely, the predator has an intrinsic mechanism for decline proportional to its abundance, described by a parameter $\theta_3 > 0$, whilst interaction with the prey leads to an increase in the predator population at a rate described by a parameter $\theta_4 > 0$.
	The resulting system of ODEs is:
	\begin{align*}
	\frac{\mathrm{d} u_1}{\mathrm{d} t} &=  \theta_1 u_1 - \theta_2 u_1 u_2, 
	\\
	\frac{\mathrm{d} u_2}{\mathrm{d} t} &= \theta_4 u_1 u_2 - \theta_3 u_2. 
	\end{align*}
	To cast this model in the setting of \Cref{sec: methods} we set $x \in \mathbb{R}^4$ to be the vector whose entries are $\log(\theta_1), \dots, \log(\theta_4)$, so that we have a $d = 4$ dimensional parameter for which inference is performed.
	
	The experiment that we report considers synthetic data which are corrupted by Gaussian noise such that the terms $\phi_i$ in \eqref{eq:original_likelihood} have expression \eqref{eq:2d_likelihood}, 
	with $C = \text{diag}(0.2^2, 0.2^2)$.
	The initial condition was $u(0) = (1, 1)$ and the data-generating parameters were
	$x = \log(\theta)$, with $\theta = (0.67, 1.33, 1,1)$.
	The times $t_i$, $i = 1, \dots, 2400$, at which data were obtained were taken to be uniformly spaced on $[0,25]$.
	\Cref{fig:Lotka_data} displays the dataset. 
	A standard Gaussian prior $\pi(x)$ was used. 
	
	Exemplar trace plots for the MCMC methods are presented in \Cref{fig:Lotka_traceplots_z}.
	The over-dispersed initial states used for the $L$ chains are the same as those reported in \Cref{table:initial-points-Goodwin}, while the univariate and multivariate convergence diagnostics, computed  every 1000 iterations, are shown respectively in \Cref{fig:Lotka_conv_diagnostics_z} and \Cref{fig:Lotka_conv_diagnostics_z_multi}.
	The values of the thresholds $\delta(L,\alpha, \epsilon)$ are the same as those reported in \Cref{table:delta-Lotka}. 
	For each MCMC method, the estimated burn-in period is presented in Table \ref{table:burnin-Lotka}. 
	
	The additional results for the Lotka--Volterra model that we present in this appendix are as follows:
	
	\begin{itemize}
	    \item Figures \ref{fig:Lotka_SP_RW} (\texttt{RW}), \ref{fig:Lotka_SP_RW_ADA} (\texttt{ADA-RW}), \ref{fig:Lotka_SP_MALA} (\texttt{MALA}) and \ref{fig:Lotka_SP_MALA_PRECOND} (\texttt{P-MALA}) display point sets of size $m=20$ selected using traditional burn in and thinning methods, \texttt{Support Points} and \texttt{Stein Thinning}, based on MCMC output. 
	\end{itemize}

	\begin{figure}[t!]
		\centering
		\includegraphics[width = 0.45\textwidth]{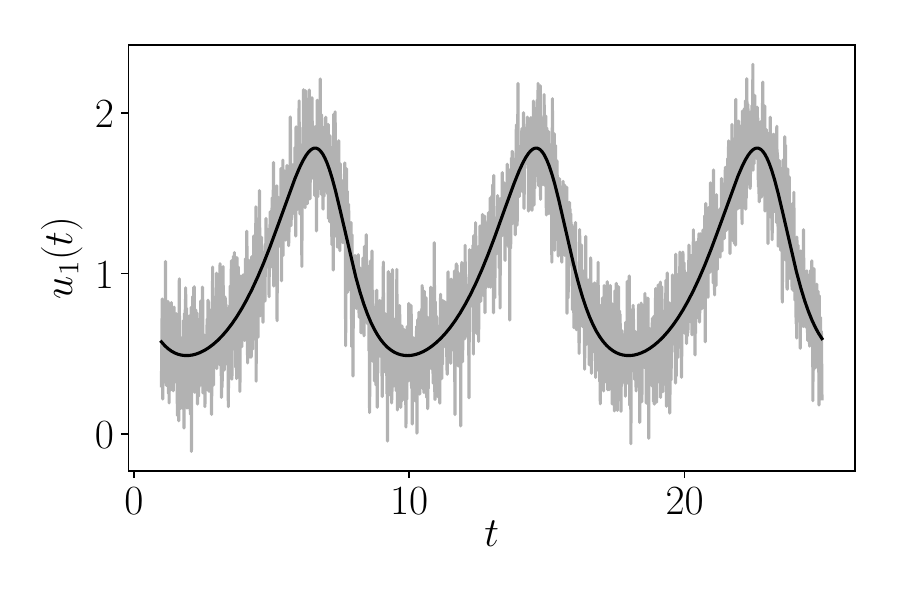}
		\includegraphics[width = 0.45\textwidth]{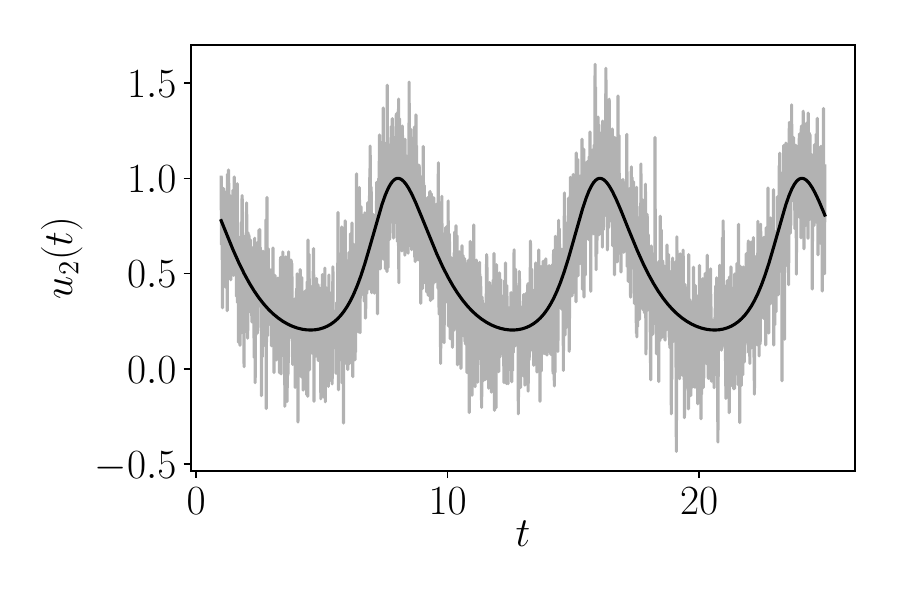}
		\caption{Data (gray) and  ODE solution corresponding to the true data-generating parameters (black) for the Lotka--Volterra model.}
		\label{fig:Lotka_data}
	\end{figure}
	
	\begin{figure}[ht!]
		\centering
		\includegraphics[width = 0.9\textwidth]{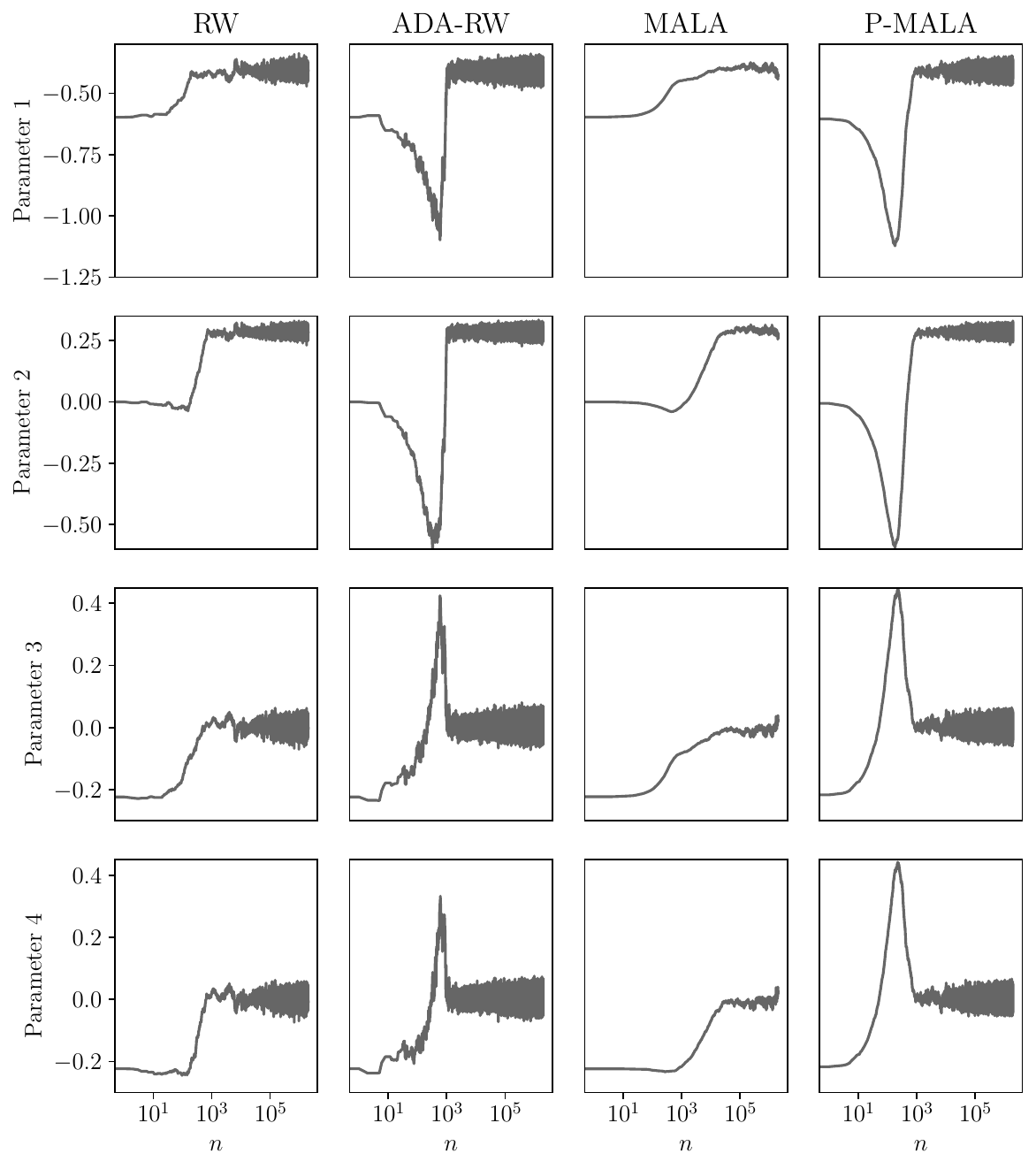}
	
		\caption{Trace plots for the parameters $x_i$ in the Lotka--Volterra model, plotted against the MCMC iteration number. 
			Each row corresponds to one of the four parameters, while each column corresponds to one of the four MCMC methods considered.
		}
		\label{fig:Lotka_traceplots_z}
	\end{figure}	
	
	\begin{table}[t!]
		\centering
		\footnotesize
		\begin{tabular}{|c||c | c|}
			\hline
			 &  \textbf{Initial State for Parameters}  & \textbf{Initial State for Parameters}  \\
			\textbf{Chain Number} & (\texttt{RW}, \texttt{MALA}, \texttt{P-MALA}) & (\texttt{ADA-RW})   \\
			\hline			
			\hline
			1 & (0.55, 1, 0.8, 0.8)  &   (0.55,	1,	0.8, 0.8) \\
			2 & (1.5, 1, 0.8, 0.8)    &  (0.55,	1,	0.8, 1.3)   \\
			3 & (1.3, 1.33, 0.5, 0.8) &  (1.3, 1.33, 0.5, 0.8)  \\
			4 & (0.55, 3, 3, 0.8)    &    (0.55,	1,	1.5,	1.5)    \\
			5 & (0.55, 1, 1.5, 1.5)   &    (0.55,	1.3,	1,	0.8)   \\
			\hline 
			
		\end{tabular}
		\caption{Initial states, $\theta = \exp(x)$, over-dispersed with respect to the posterior, for the $L = 5$ independent Markov chains used in the Lotka--Volterra model.  The parameters used to generate the data were $\theta = (0.67,	1.33,	1,	1)$. }
		\label{table:initial-points-Lotka}
	\end{table}

	\begin{figure}[t!]
		\centering
		\includegraphics[width = 1\textwidth]{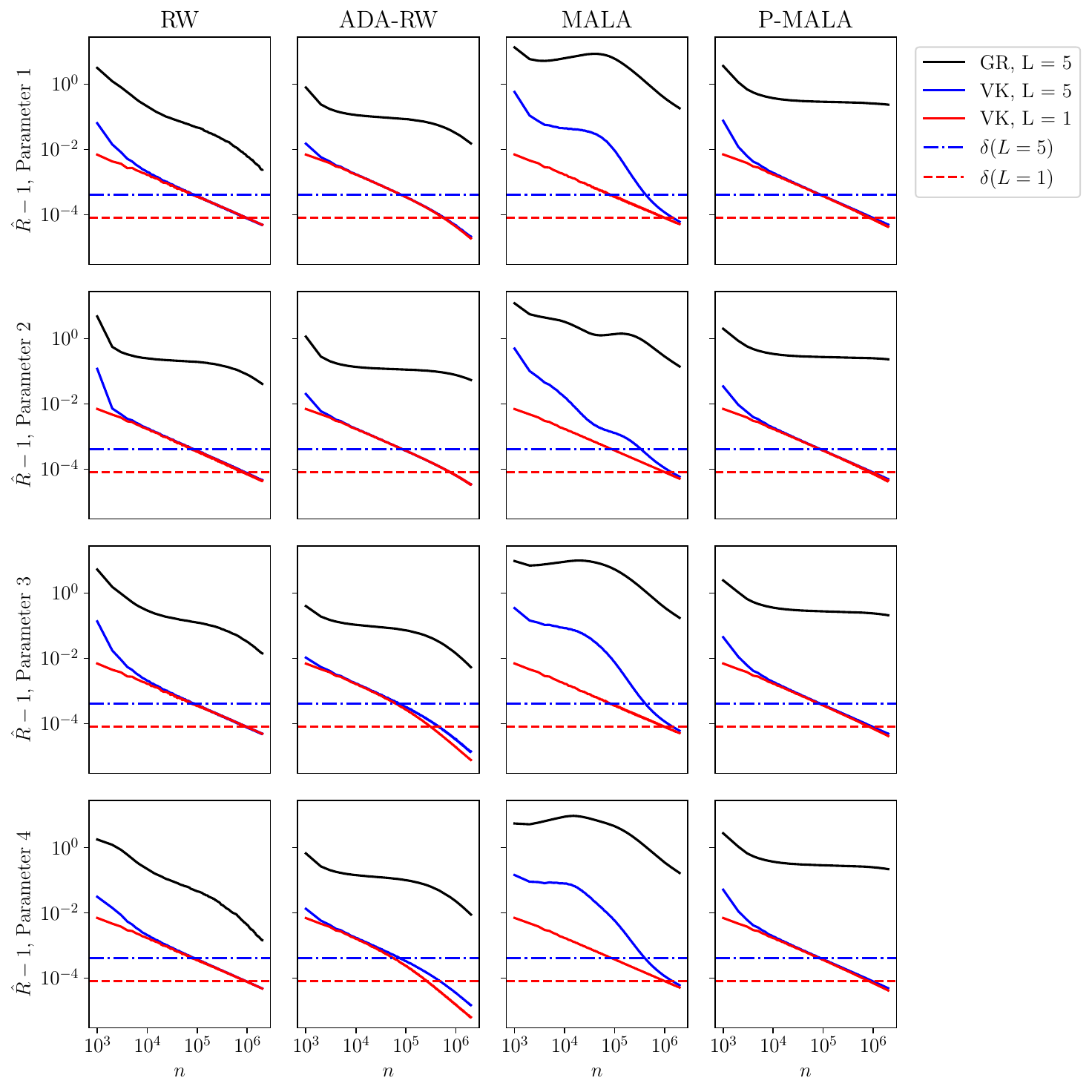}
		\caption{Univariate convergence diagnostics, for the Lotka--Volterra model, plotted against the MCMC iteration number. 
			The black line represents the GR diagnostic (based on $L=5$ chains), while the blue and red lines represent the VK diagnostic (based on $L = 5$ and $L=1$ chains, respectively). 
			The dash-dotted ($L=5$) and dashed ($L=1$) horizontal lines correspond to the critical values $\delta(L, \alpha, \epsilon)$, used to determine the burn-in period; see Table~\ref{table:delta-Lotka}.
			}
		\label{fig:Lotka_conv_diagnostics_z}
	\end{figure}

	\begin{figure}[t!]
		\centering
		\includegraphics[width = 1\textwidth]{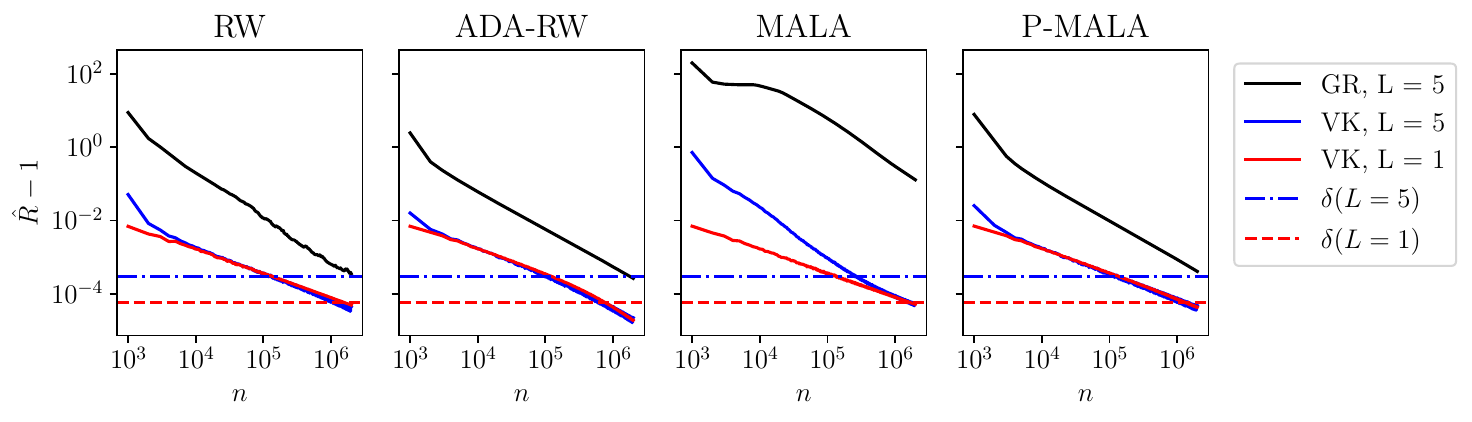}
		\caption{Multivariate convergence diagnostics for Lotka--Volterra, plotted against the MCMC iteration number. 
			The black line is the GR diagnostic (based on $L=5$ chains), while the blue and red lines are the VK diagnostic (based on $L = 5$ and $L=1$ chains, respectively). 
			The dotted ($L=5$) and dashed ($L=1$) horizontal lines correspond to the critical values $\delta(L, \alpha, \epsilon)$, used to determine the burn-in period; see Table~\ref{table:delta-Lotka}.}
		\label{fig:Lotka_conv_diagnostics_z_multi}
	\end{figure}

		\begin{table}[t!]
		\centering
		\footnotesize
		\begin{tabular}{|c||c||c | c | c|}
			\hline
			\textbf{MCMC Diagnostics} & \textbf{Sampler} &  $\hat{b}^{\text{GR}, 5}$ & $\hat{b}^{\text{VK}, 5}$ & $\hat{b}^{\text{VK}, 1}$   \\
			\hline 
			\hline 
				\multirow{4}{*}{\textbf{Univariate}} & \texttt{RW} &   $>n$ &  88,000& 954,000\\
			\cline{2-5}
			& \texttt{ADA-RW} &$>n$&  84,000& 764,000\\	
		    \cline{2-5}
			& \texttt{MALA} 	  &  $>n$&  424,000&  995,000\\
			\cline{2-5}
			& \texttt{P-MALA} 	&    $>n$&  90,000& 820,000\\
			\hline \hline
			\multirow{4}{*}{\textbf{Multivariate}} & \texttt{RW} & $>n$ & 119,00 & $1,512,000$ \\
			\cline{2-5}	
			& \texttt{ADA-RW} &   1,797,000  & 99,000& $743,000$\\
			\cline{2-5} 
			& \texttt{MALA} 	 &   $>n$ & 259,000  &  $1,573,000$\\
			\cline{2-5} 
			& \texttt{P-MALA} &  $>n$ & 114,000 & $1,251,000$ \\
			\hline
		\end{tabular}
		\caption{Estimated burn-in period for the Lotka--Volterra model, using the GR diagnostic based on $L$ chains, $\hat{b}^{\text{GR},L}$ ($L=5$), and the VK diagnostic based on $L$ chains,  $\hat{b}^{\text{VK},L}$, ($L = 1,5$). 		In each case both univariate and multivariate convergence diagnostics are presented; in the univariate case we report the largest value obtained when looking at each of the $d$ parameters individually to estimate the burn-in period. 			The symbol ``$>n$'' indicates the case in which a diagnostic did not go below the $1+\delta$ threshold. 	
		}
		\label{table:burnin-Lotka}
	\end{table}

		\begin{figure}[ht!]
		\centering
		\includegraphics[width = 0.4\textwidth]{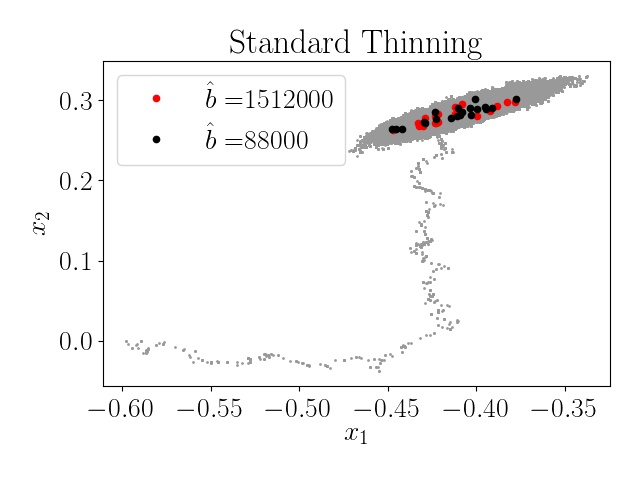}		
		\includegraphics[width = 0.4\textwidth]{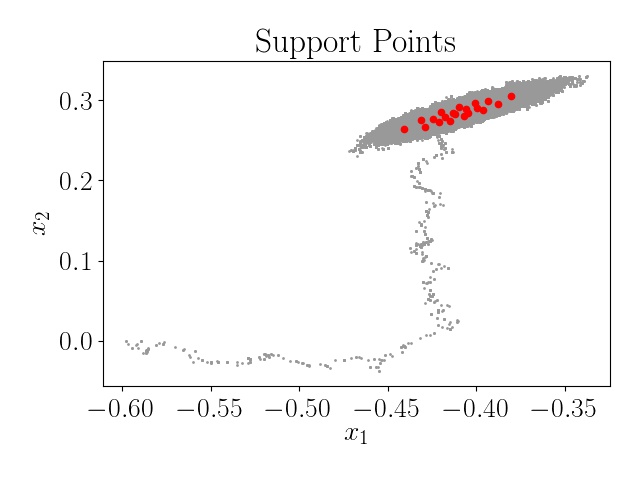}	
		
		\includegraphics[width = 0.4\textwidth]{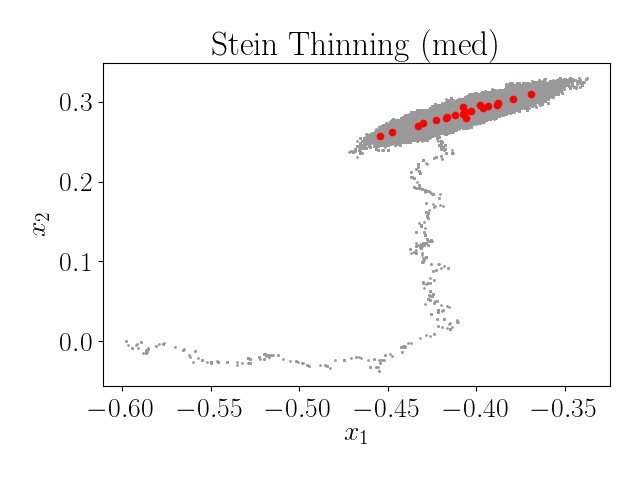}
		\includegraphics[width = 0.4\textwidth]{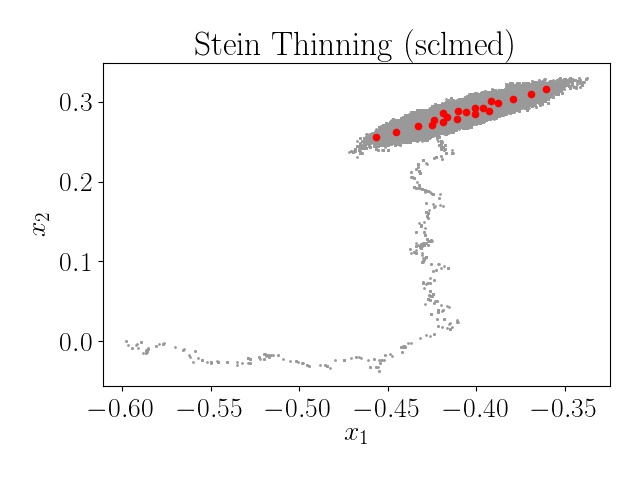}		
		\includegraphics[width = 0.4\textwidth]{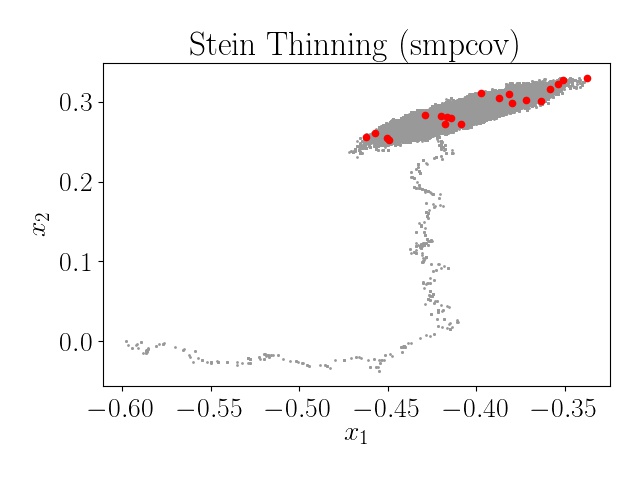}

		\caption{Projections on the first two coordinates of the \texttt{RW} MCMC output for the Lotka--Volterra model (grey dots), together with the first $m=4$ points selected through: traditional burn-in and thinning (the amount of burn in is indicated in the legend);  the \texttt{Support Points} method; \texttt{Stein Thinning}, for each of the settings \texttt{med}, \texttt{sclmed}, \texttt{smpcov}.
		}
		\label{fig:Lotka_SP_RW}
	\end{figure}	
	
	\begin{figure}[t!]
		\centering
		\includegraphics[width = 0.4\textwidth]{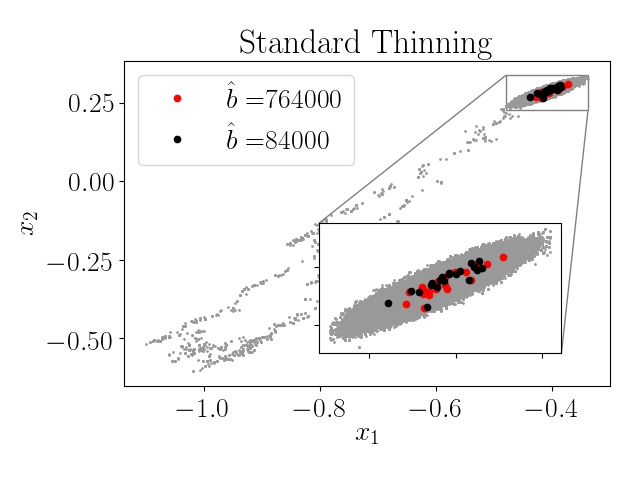}		
		\includegraphics[width = 0.4\textwidth]{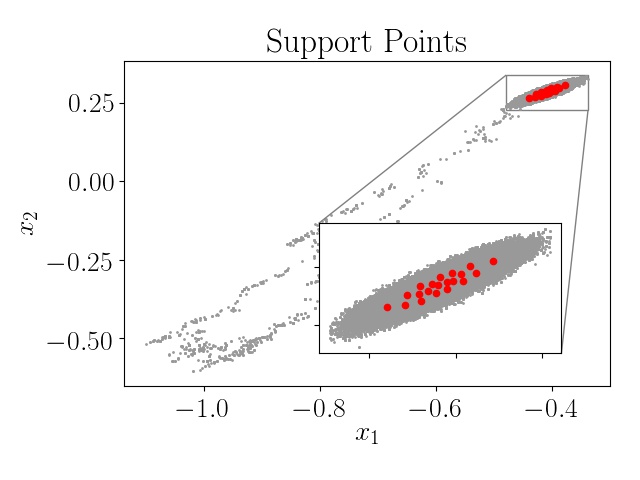}	
		\includegraphics[width = 0.4\textwidth]{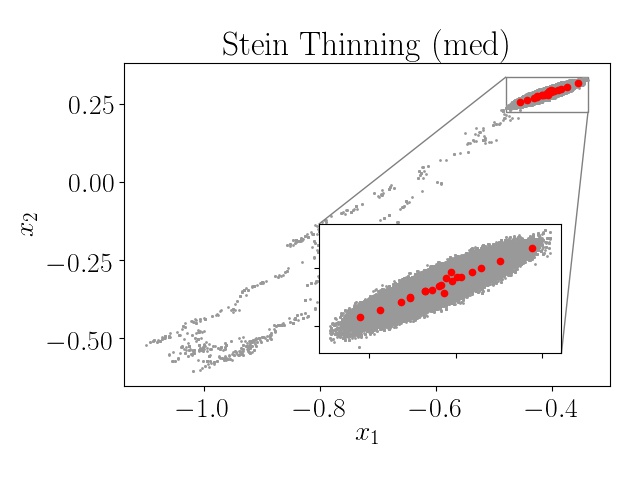}
		\includegraphics[width = 0.4\textwidth]{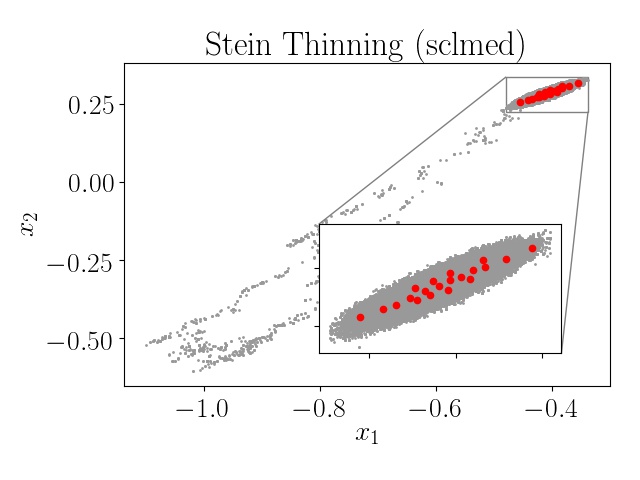}		
		\includegraphics[width = 0.4\textwidth]{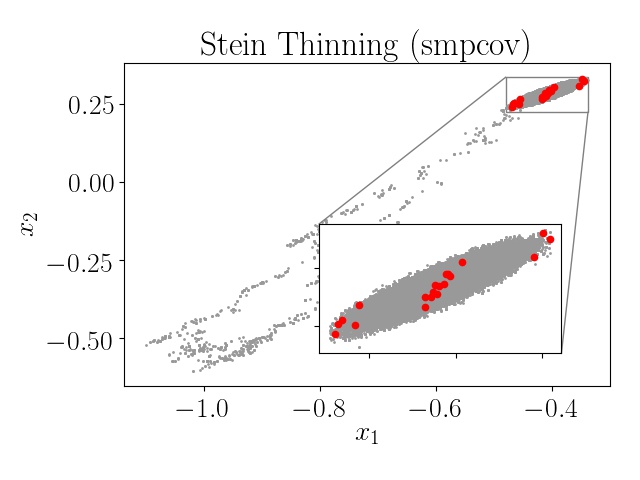}

		\caption{Projections on the first two coordinates of the \texttt{ADA-RW} MCMC output for the Lotka--Volterra model (grey dots), together with the first $m=4$ points selected through: traditional burn-in and thinning (the amount of burn in is indicated in the legend);  the \texttt{Support Points} method; \texttt{Stein Thinning}, for each of the settings \texttt{med}, \texttt{sclmed}, \texttt{smpcov}.
		}
		\label{fig:Lotka_SP_RW_ADA}
	\end{figure}

	\begin{figure}[t!]
		\centering
		\includegraphics[width = 0.4\textwidth]{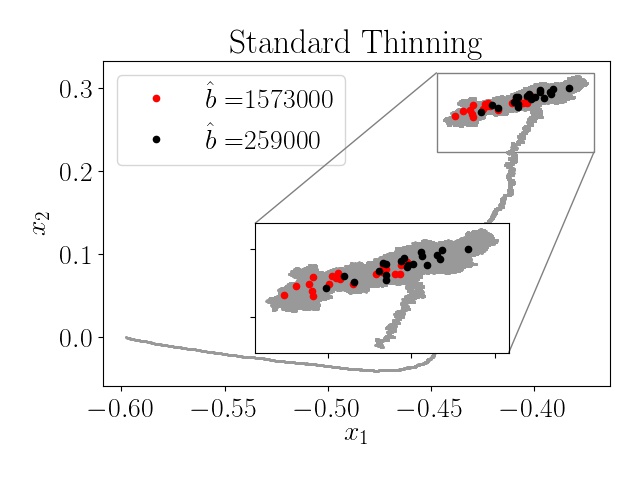}		
		\includegraphics[width = 0.4\textwidth]{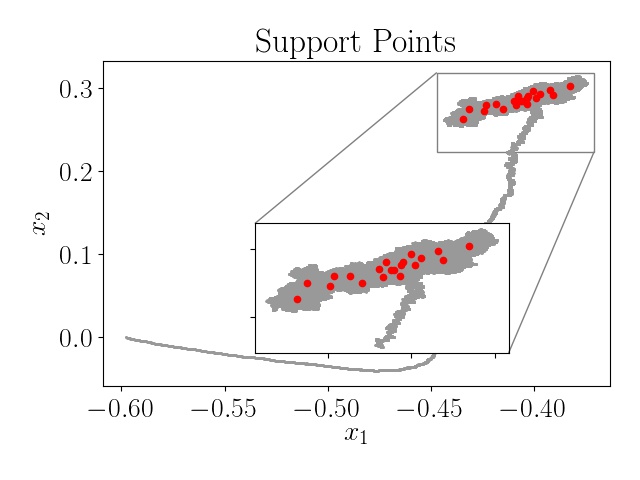}	
		\includegraphics[width = 0.4\textwidth]{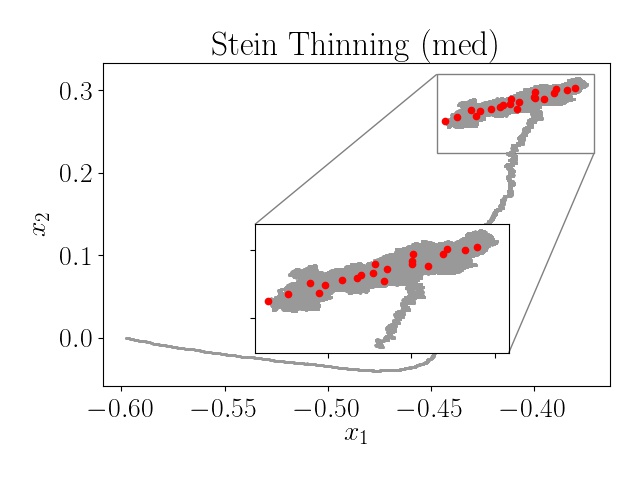}
		\includegraphics[width = 0.4\textwidth]{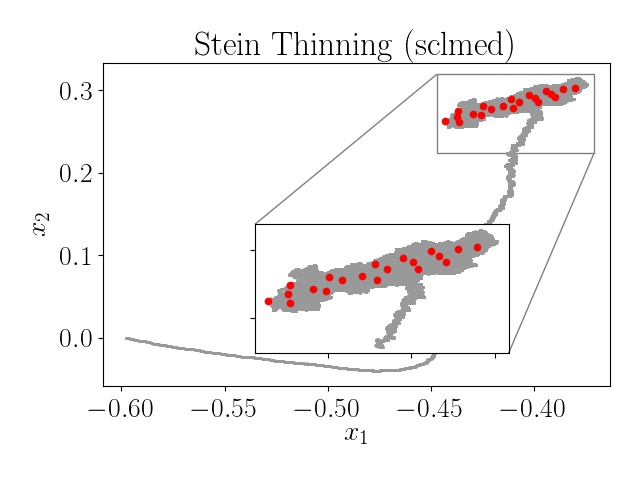}		
		\includegraphics[width = 0.4\textwidth]{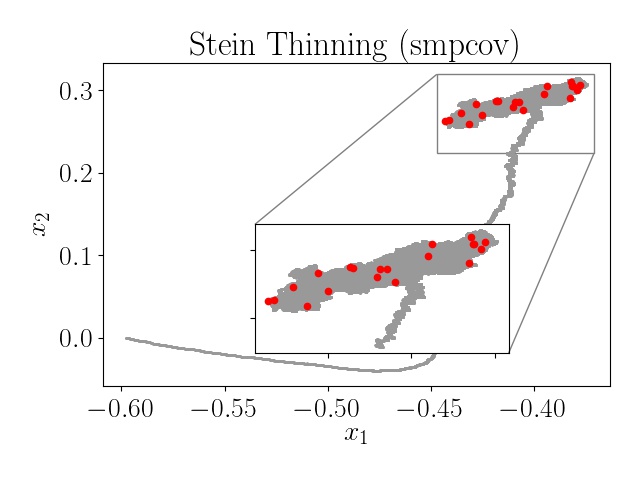}

		\caption{Projections on the first two coordinates of the \texttt{MALA} MCMC output for the Lotka--Volterra model (grey dots), together with the first $m=4$ points selected through: traditional burn-in and thinning (the amount of burn in is indicated in the legend);  the \texttt{Support Points} method; \texttt{Stein Thinning}, for each of the settings \texttt{med}, \texttt{sclmed}, \texttt{smpcov}.
		}
	\label{fig:Lotka_SP_MALA}
	\end{figure}

	\begin{figure}[t!]
		\centering
		\includegraphics[width = 0.4\textwidth]{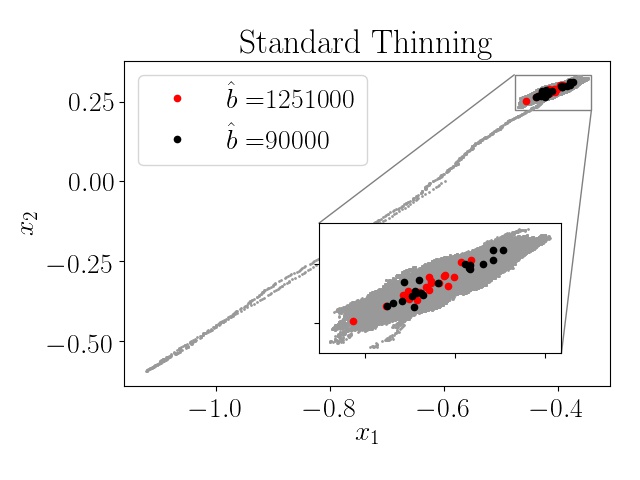}		
		\includegraphics[width = 0.4\textwidth]{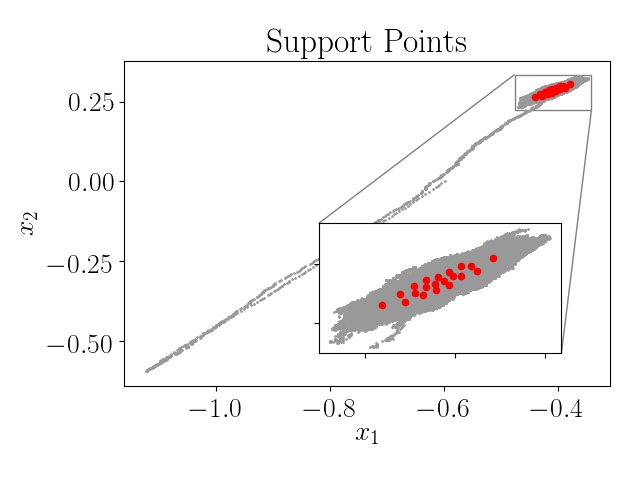}	

		\includegraphics[width = 0.4\textwidth]{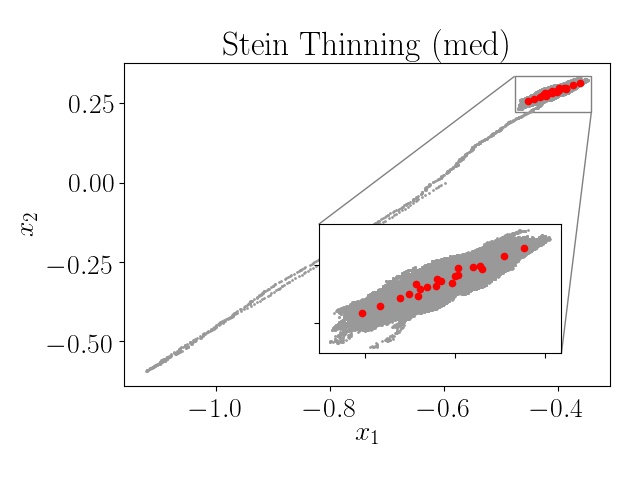}
		\includegraphics[width = 0.4\textwidth]{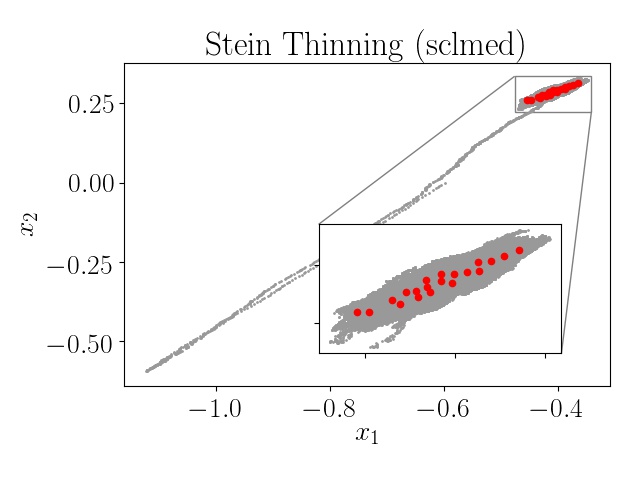}		
		\includegraphics[width = 0.4\textwidth]{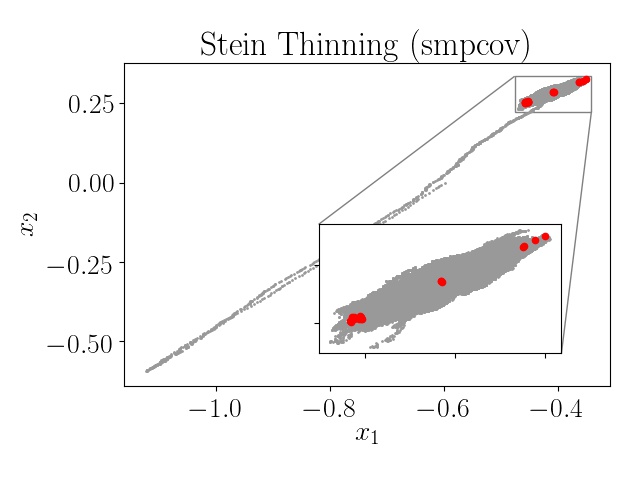}

		\caption{Projections on the first two coordinates of the \texttt{P-MALA} MCMC output for the Lotka--Volterra model (grey dots), together with the first $m=4$ points selected through: traditional burn-in and thinning (the amount of burn in is indicated in the legend);  the \texttt{Support Points} method; \texttt{Stein Thinning}, for each of the settings \texttt{med}, \texttt{sclmed}, \texttt{smpcov}.
		}
	\label{fig:Lotka_SP_MALA_PRECOND}
	\end{figure}

	\FloatBarrier

	\subsection{Calcium Signalling Model}\label{Appendix:Hinch}
	
	This appendix contains a detailed biochemical description of the calcium singalling model studied in \Cref{subsec: cardiac} of the main text, together with the experimental dataset that we collected.
	
	The \cite{Hinch2004} single cell model simulates the calcium transient evoked by membrane depolarisation in a cardiac cell. The model has a mathematical representation of the extracellular space and the intracellular compartment consisting of the sarcoplasmic reticulum (SR), dyadic space and cytosol. The major sarcolemmal calcium pathways are included: the L-type Ca channel (LCC), the plasmalemmal membrane calcium ATPase (PMCA) and the sodium-calcium exchanger (NCX). Inside the cell, the model has mathematical representations for calcium release from the SR to dyadic space through ryanodine receptors (RyR) and re-sequestration of calcium from the dyadic space into the SR by the SR ATPase (SERCA). Calcium buffering is also featured for the cytosol. A schematic representation of the cell model is given in Figure~\ref{fig:cell_scheme}.
	
	Membrane depolarisation is triggered by an electrical event. This causes calcium to enter through LCCs into the dyadic space, producing a local rise in Ca concentration, sufficient to activate RyRs.  This process engages a feedback, whereby Ca release from the SR causes more RyR opening events. As the released Ca diffuses into the cytosol, most of it becomes buffered, but some ions remain free and underpin the Ca transient. Recovery following Ca release is driven by SERCA, which re-sequesters Ca into the SR, and NCX and PMCA which extrude calcium across the sarcolemma. This returns the cell to is initial conditions, ready for the next electrical stimulation. 
	
	The Hinch model describes the nonlinear, time-dependent interaction of the four Ca handling transporters (LCC, PMCA, RyR and SERCA) and lumped buffering by a system of 7 ODEs whose parameter is $d=38$ dimensional. The first three differential equations provide a simplified four-state model describing the interaction between LCC and RyR within the dyadic space; here, only three states are simulated due to a conservation of mass constraint. The remaining four differential equations describe: calcium concentration in the sarcoplasmic reticulum and the cytosol, the calcium bound to cytosolic buffers and calcium current across the cellular membrane. Of these state variables, only the concentration of free calcium in the cytosol and the transmembrane current can be experimentally observed. 
	
	To provide a rich dataset for characterising calcium dynamics in a single cardiac myocyte, we applied three experimental protocols in sequence on a single myocyte.  During these protocols, we controlled membrane potential and measured membrane currents electrophysiologically and, after appropriate calibration, followed Ca fluorimetrically. The calcium handling proteins were interrogated by relating currents and Ca concentration in response to defined membrane potential manoeuvres, and in the presence of drugs to eliminate various confounding components.  The first voltage protocol interrogated LCC currents at different voltages, and measured their response in terms of SR release.  In the second protocol, a train of depolarisations then triggered Ca transients which provided information about SR release and their recovery provided a readout of SERCA, NCX and PMCA activities.  The third protocol consistent of rapid exposure to caffeine which emptied the SR and short-circuited SERCA.  This provided information about SR load, and the subsequent recovery is a readout of NCX and PMCA. Buffering was calculated from the quotient of measured Ca rise upon caffeine exposure and the amount of Ca released back-calculated from sarcolemmal current generated by NCX. The dataset contains 12998 observations of cytosolic free calcium concentration observed at a 60 Hz sampling frequency, and 22260 transmembrane current observations, both for a duration of 3 minutes. The data are displayed in \Cref{fig:Hinch_data}, where the different colours show the three parts of the biological protocol explained above.       
	The calcium signalling model in \Cref{fig:cell_scheme} is represented by a coupled system of~$7$ ODEs and depends on a $d=38$ dimensional parameter, which is to be estimated based on the experimental dataset.
	As just described, the data consist of measurements of calcium concentration in the cytoplasm and transmembrane current whilst the cell was externally stimulated,  so that only two of the state variables (in our case, $u_5$ and $u_7$) were observed (we denote the observations of these states $y^5$ and $y^7$, respectively).
	Our likelihood took the simple Gaussian form $\phi_i(u(t_i)) \propto \exp(-\frac{1}{2 \sigma_5^2} (y^5_i - u_5(t_i))^2 ) + \exp(-\frac{1}{2 \sigma_7^2} (y^7_i - u_7(t_i))^2 )$ with $\sigma_5 =  2.07 \times 10^{-8}$ and $\sigma_7 = 1.62 \times 10^{-10}$.
	The ODE was numerically solved using 
	\verb+CVODES+ \citep{hindmarsh2005sundials}
	and sensitivities were computed by solving the forward sensitivity equations; see \Cref{ap: MCMC methods}.
	Further details of the expert-elicited prior, the data pre-processing procedure and numerical details associated with the ODE solver will be reported in a separate manuscript, in preparation as of 12th July 2021, and are available on request.
	
	In the experiments that follow, \texttt{RW} MCMC was used both to target the posterior $P$ and to target a tempered distribution $Q$.
	The latter is equivalent to multiplying the measurement error standard deviations $\sigma_5$ and $\sigma_7$ by 8, and has the effect of rendering $Q$ more diffuse than $P$, in order that $Q$ is more favourable for MCMC.
	The specific value of 8 corresponded to the smallest amount of tempering required to achieve convergence within the available computational budget.
	A total of $n = 4 \times 10^6$ iterations were performed, and in each case the first $10^6$ iterations were used to adapt the scale of the Gaussian proposal distribution in the \texttt{RW} sampler, so that an acceptance rate close to 0.234 \citep{gelman1997weak} was achieved. 
	The first $10^6$ iterations were then discarded.
		
	The additional results for the calcium signalling model that we present in this appendix are as follows:
	
	\begin{itemize}
	\item \Cref{fig:Hinch_kdeA,fig:Hinch_kdeB,fig:Hinch_kdeC,fig:Hinch_kdeD} contain kernel density estimates for posterior marginals obtained by \texttt{Stein Thinning} applied to tempered \texttt{RW} MCMC output, versus standard \texttt{RW} MCMC output. 
	    \item Figures \ref{fig:Hinch_KSD_sclmed} and \ref{fig:Hinch_KSD_sclmed} present results for KSD based on \texttt{sclmed} and \texttt{smpcov} settings, to complement \Cref{fig:Hinch_KSD} in the main text.
	\end{itemize}

	\begin{figure}[t!] 
			\centerline{
				\includegraphics[width=0.8	\linewidth]{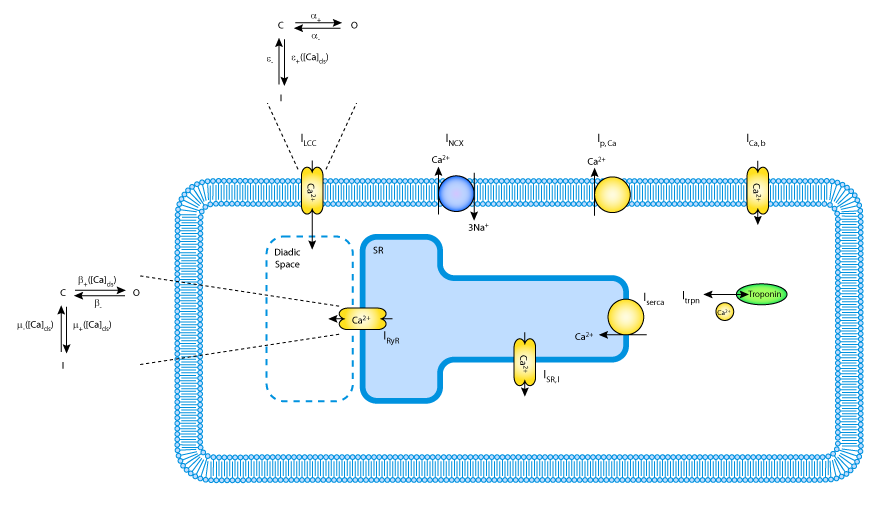}} 
			\caption{Calcium signalling model; a schematic representation due to \cite{Hinch2004}.
				The model consists of 6 coupled ordinary differential equations and depends upon 38 real-valued parameters that must be estimated from an experimental dataset.
			} 
			\label{fig:cell_scheme}
		\end{figure}

	\begin{figure}[t!]
			\centering
			\includegraphics[width = 0.45
			\textwidth]{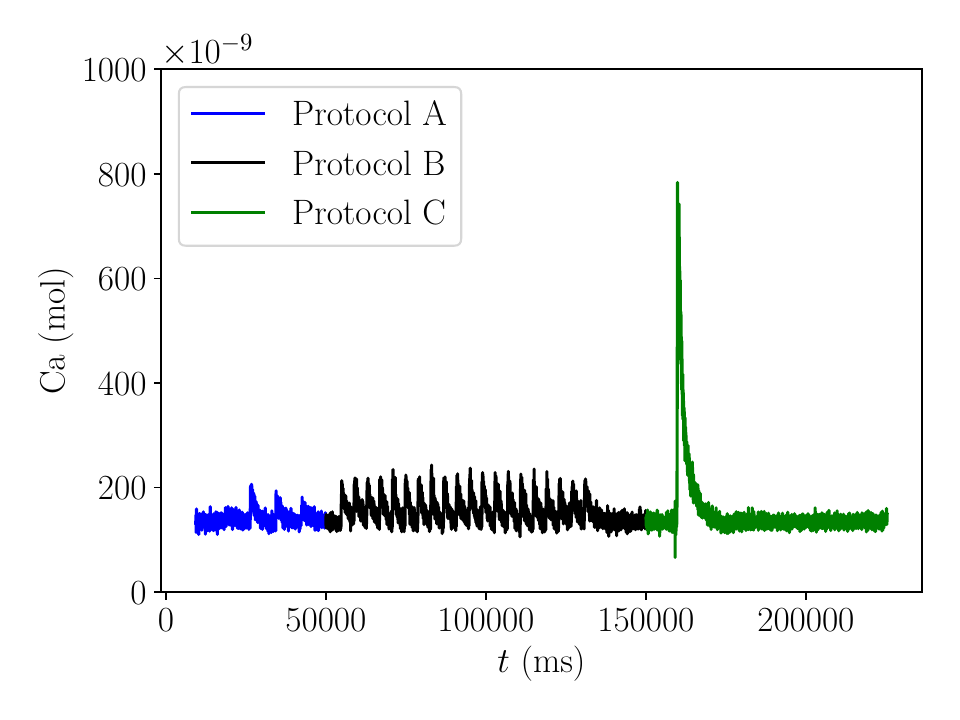}
			\includegraphics[width = 0.45
			\textwidth]{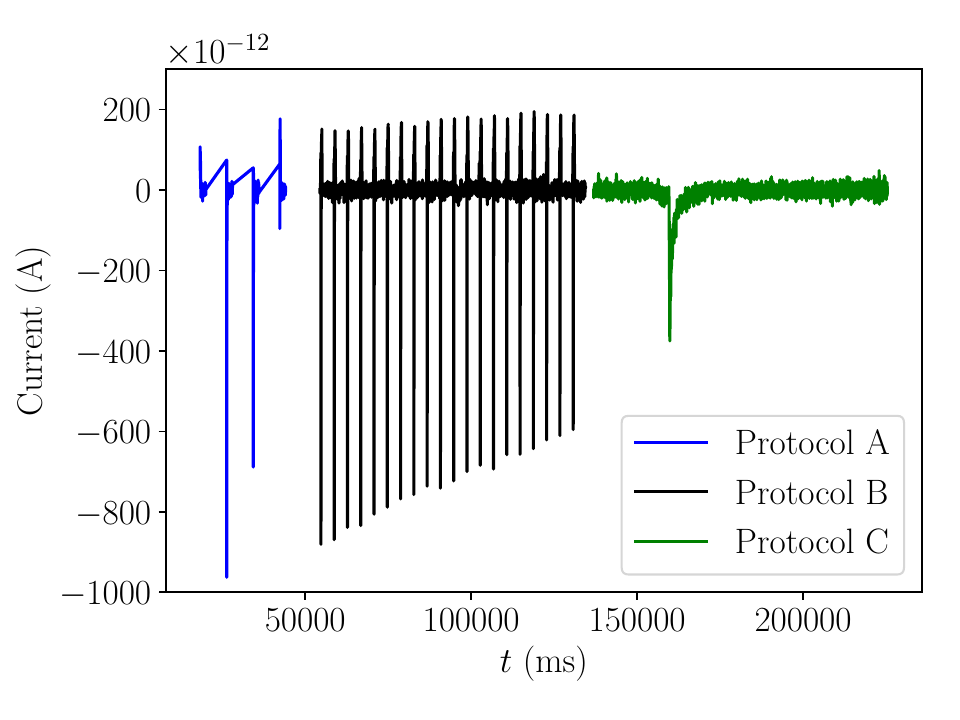}
			\caption{
				Calcium signalling data. The left panel shows calcium concentration (in mol) plotted against time (in ms), while the right panel shows transmembrane current (in A) plotted against time (in ms). The different colours show the data corresponding to the three different biological protocols.
			}
			\label{fig:Hinch_data}
	\end{figure}

 	\begin{figure}[t!]
			\centering
			\includegraphics[width = 0.65
			\textwidth]{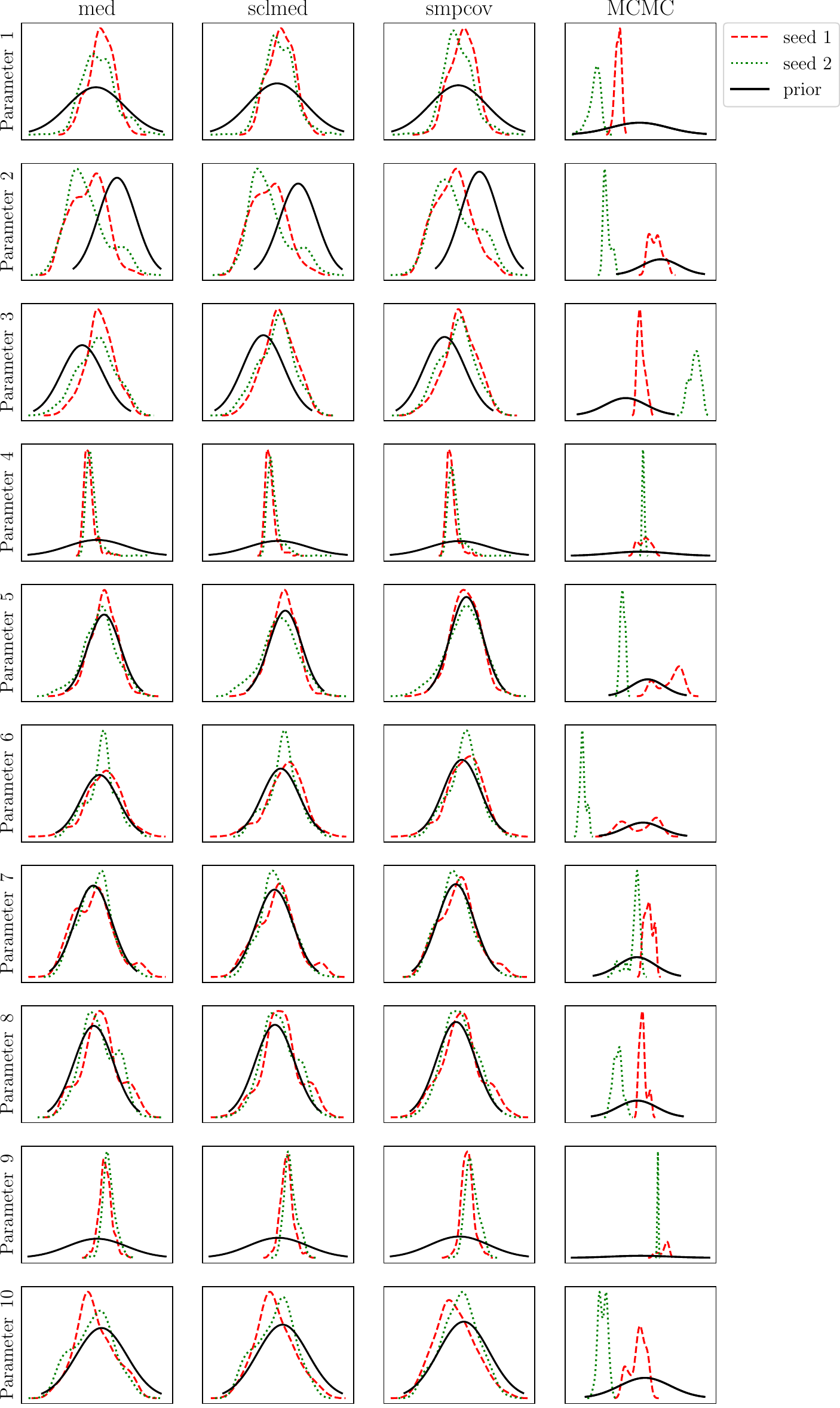}

			\caption{
				Kernel density estimates for posterior marginals in the calcium signalling model. 
				Stein thinning with \texttt{med}, \texttt{sclmed} and \texttt{smpcov} preconditioners (first three columns) was applied to tempered \texttt{RW} MCMC output (to obtain $m=500$ points). 
				These can be contrasted with the last column, where kernel density estimates based on standard \texttt{RW} MCMC are displayed.
				[In each case, two distinct random seeds were used. 
				For reference, the black curve represents the prior marginal.]
				}
			\label{fig:Hinch_kdeA}
		\end{figure}

	 	\begin{figure}[t!]
			\centering
			\includegraphics[width = 0.65
			\textwidth]{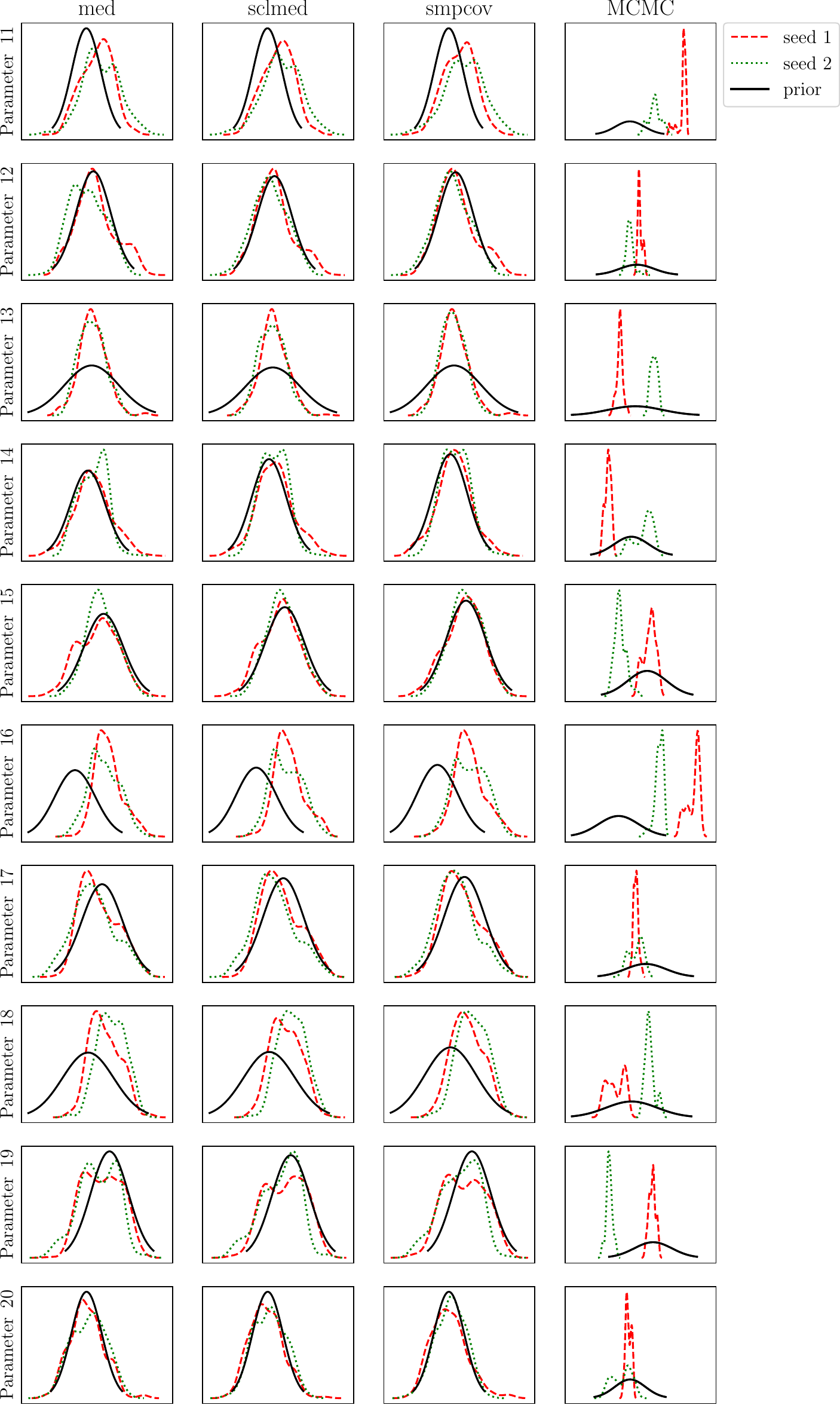}

			\caption{
				Kernel density estimates for posterior marginals in the calcium signalling model. 
				Stein thinning with \texttt{med}, \texttt{sclmed} and \texttt{smpcov} preconditioners (first three columns) was applied to tempered \texttt{RW} MCMC output (to obtain $m=500$ points). 
				These can be contrasted with the last column, where kernel density estimates based on standard \texttt{RW} MCMC are displayed.
				[In each case, two distinct random seeds were used. 
				For reference, the black curve represents the prior marginal.]
				}
			\label{fig:Hinch_kdeB}
		\end{figure}

	 	\begin{figure}[t!]
			\centering
			\includegraphics[width = 0.65
			\textwidth]{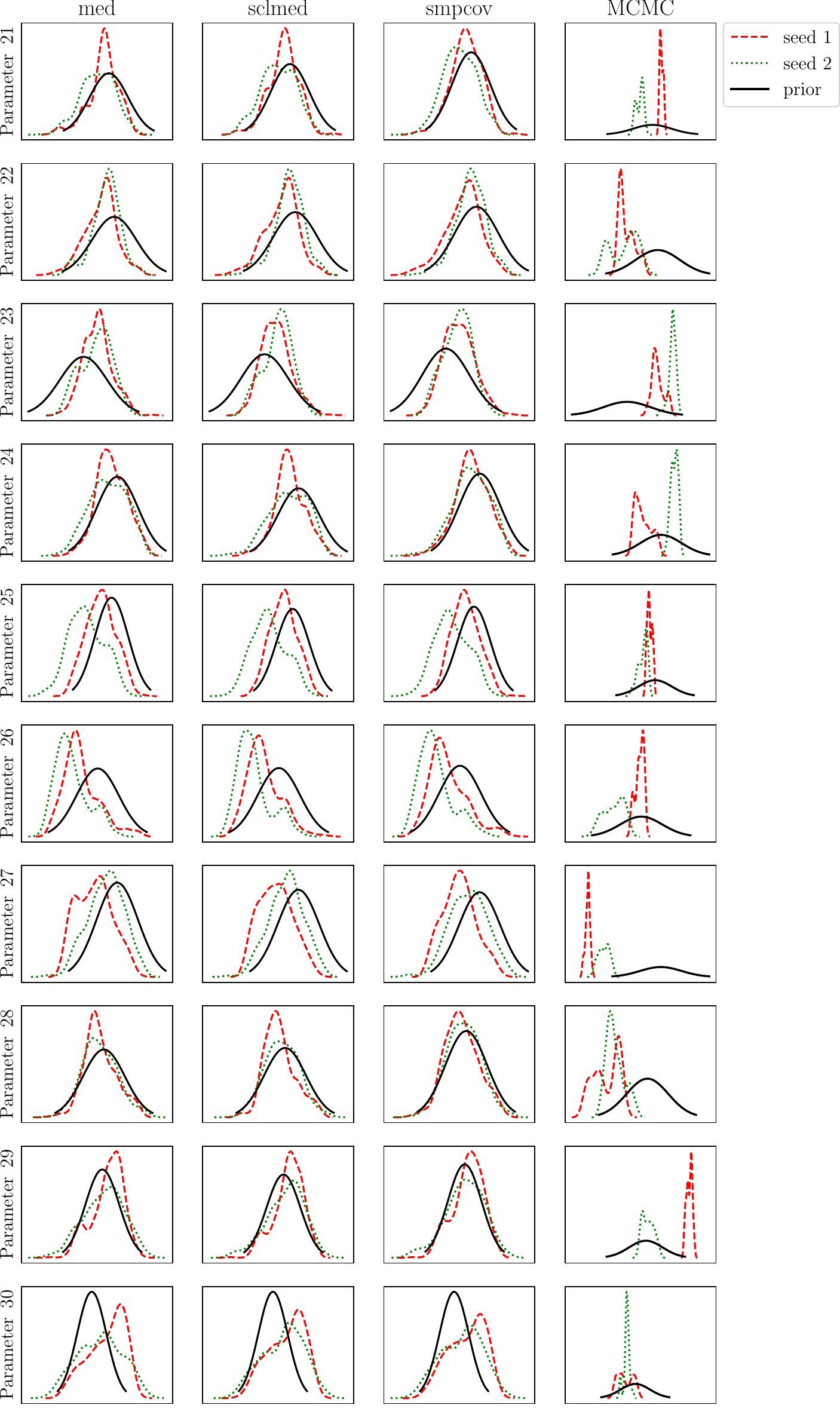}

			\caption{
				Kernel density estimates for posterior marginals in the calcium signalling model. 
				Stein thinning with \texttt{med}, \texttt{sclmed} and \texttt{smpcov} preconditioners (first three columns) was applied to tempered \texttt{RW} MCMC output (to obtain $m=500$ points). 
				These can be contrasted with the last column, where kernel density estimates based on standard \texttt{RW} MCMC are displayed.
				[In each case, two distinct random seeds were used. 
				For reference, the black curve represents the prior marginal.] 
				}
			\label{fig:Hinch_kdeC}
		\end{figure}	
	
	 	\begin{figure}[t!]
			\centering
			\includegraphics[width = 0.65
			\textwidth]{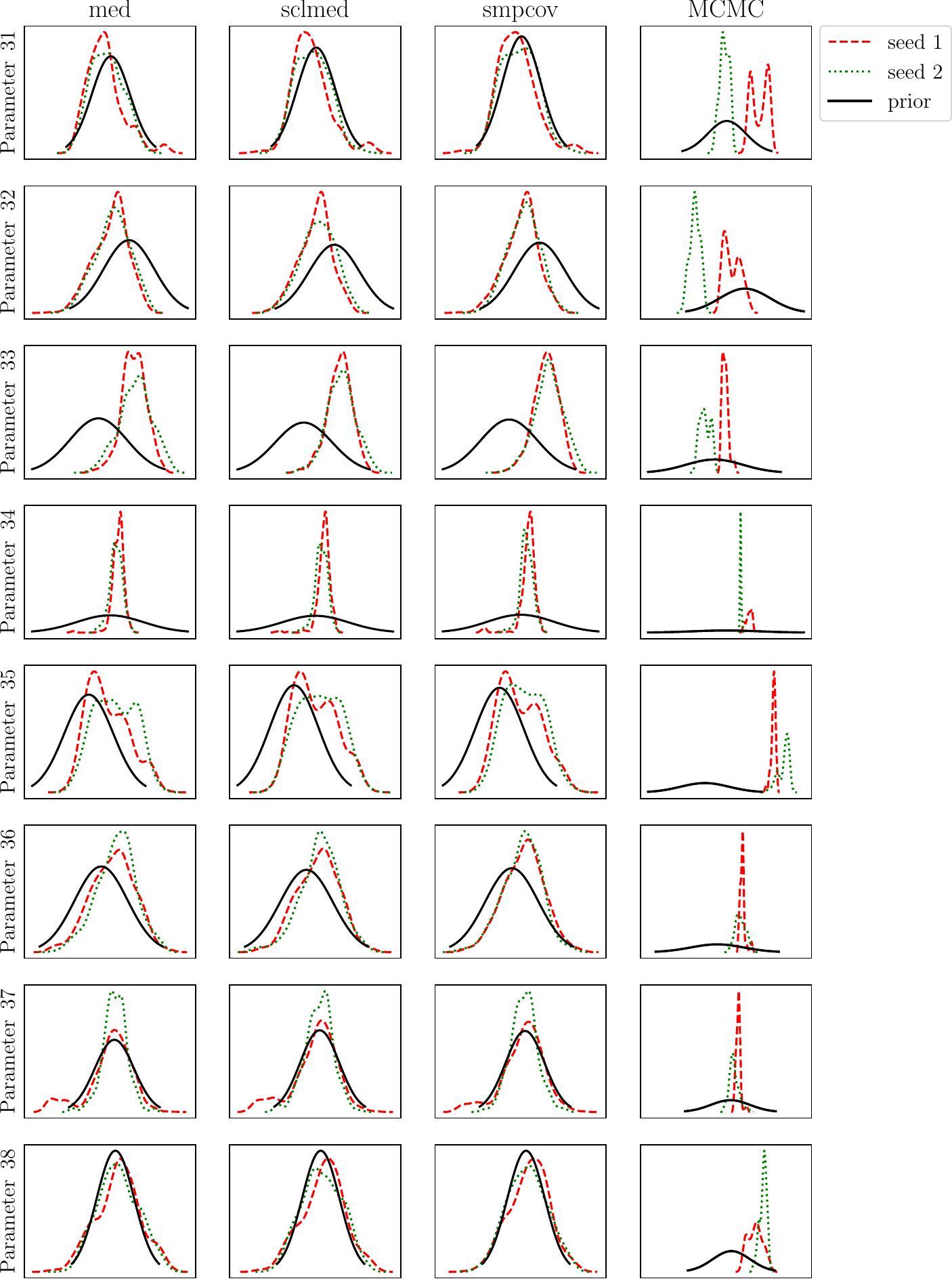}

			\caption{
				Kernel density estimates for posterior marginals in the calcium signalling model. 
				Stein thinning with \texttt{med}, \texttt{sclmed} and \texttt{smpcov} preconditioners (first three columns) was applied to tempered \texttt{RW} MCMC output (to obtain $m=500$ points). 
				These can be contrasted with the last column, where kernel density estimates based on standard \texttt{RW} MCMC are displayed.
				[In each case, two distinct random seeds were used. 
				For reference, the black curve represents the prior marginal.] 
				}
			\label{fig:Hinch_kdeD}
		\end{figure}

		\begin{figure}[t!]
			\centering
			\includegraphics[width = 0.7
			\textwidth,clip,trim = 0cm 0cm 0cm 0.6cm]{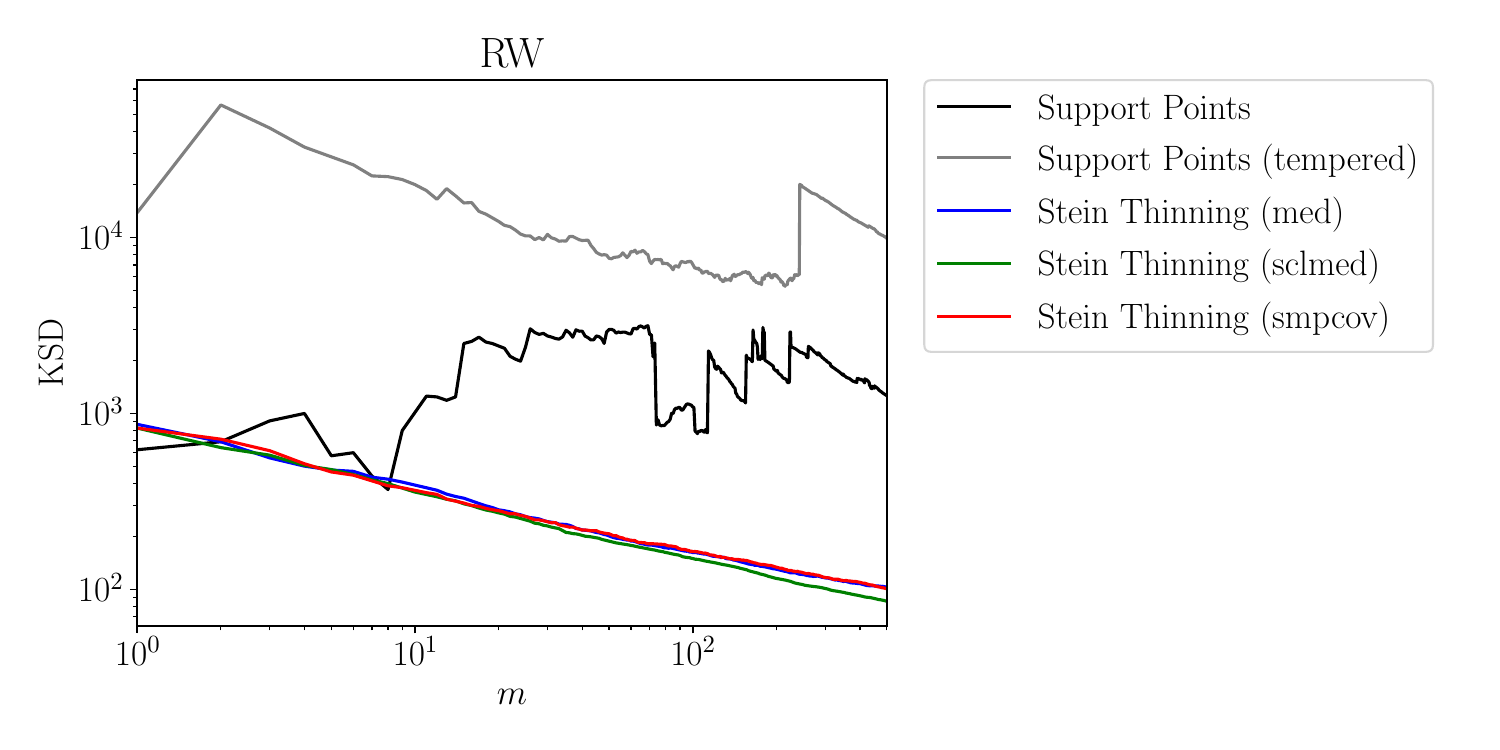}
			\caption{Calcium signalling model. Kernel Stein discrepancy (KSD) based on \texttt{sclmed}, for empirical distributions obtained using \texttt{Support Points} and \texttt{Stein Thinning}, based on output from \texttt{RW} MCMC applied to either $P$ or a tempered version of $P$.
			}
			\label{fig:Hinch_KSD_sclmed}
		\end{figure}	
	
		\begin{figure}[t!]
			\centering
			\includegraphics[width = 0.7
			\textwidth,clip,trim = 0cm 0cm 0cm 0.6cm]{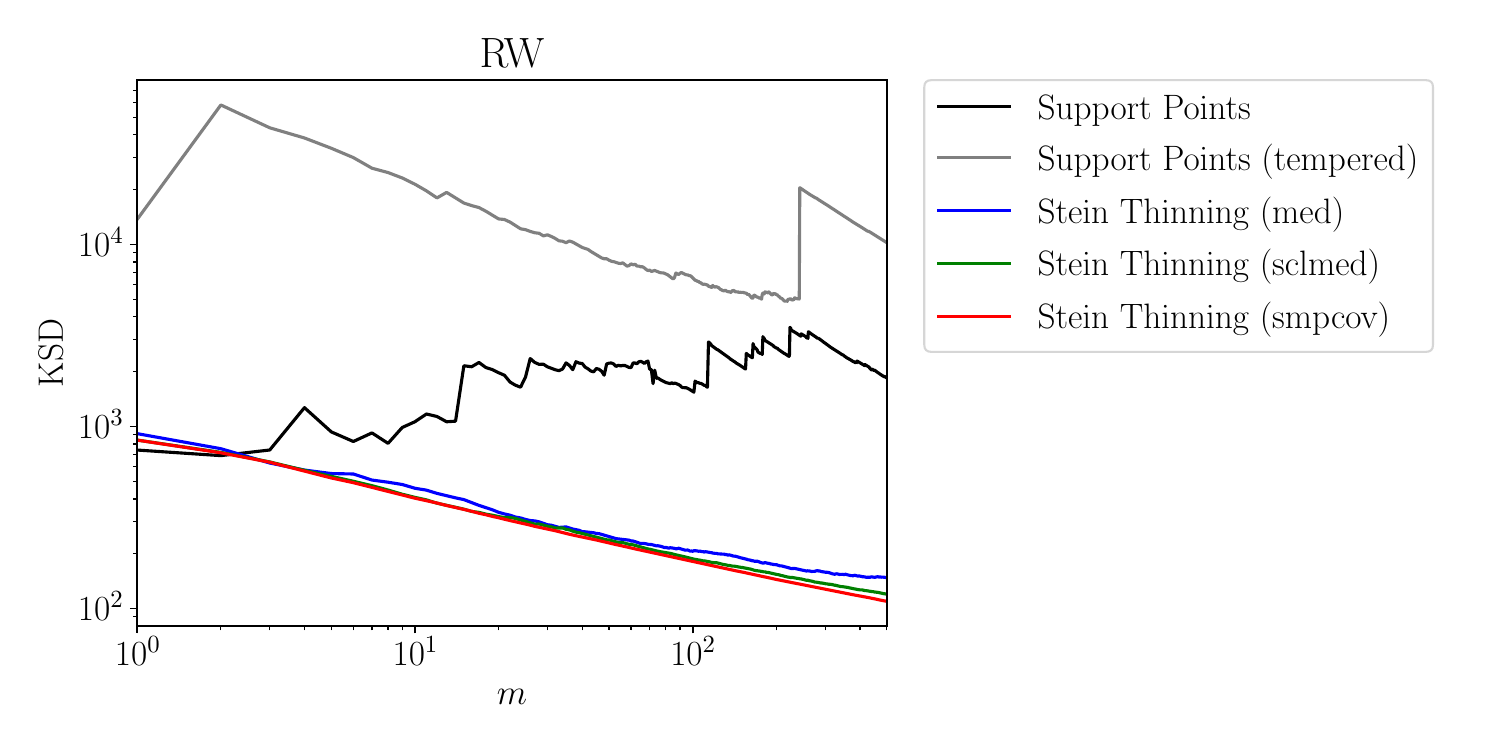}
			\caption{Calcium signalling model. Kernel Stein discrepancy (KSD) based on \texttt{smpcov}, for empirical distributions obtained using \texttt{Support Points} and \texttt{Stein Thinning}, based on output from \texttt{RW} MCMC applied to either $P$ or a tempered version of $P$.
			}
			\label{fig:Hinch_KSD_smpcov}
		\end{figure}

\end{document}